\documentclass[sigconf,breakonhyphens,nonacm=true]{acmart}
\usepackage[utf8]{inputenc}
\usepackage[T1]{fontenc}

\usepackage{amsmath}
\usepackage{amsthm}
\usepackage{slashed}
\usepackage[style=base, format=default]{subcaption}
\usepackage{booktabs}
\usepackage{ifthen}
\usepackage{algorithm}
\usepackage{graphicx}
\usepackage{etoolbox}
\usepackage[noend]{algpseudocode}
\usepackage{array}
\usepackage{enumitem}
\usepackage{xcolor}
\usepackage{ninecolors}
\usepackage{csquotes}
\usepackage{xspace}
\usepackage{bm}
\usepackage{tikz}
\usetikzlibrary{
	matrix,
	graphs,
	positioning,
	calc,
	arrows.meta,
	decorations.pathmorphing,
	decorations.pathreplacing,
	quotes,
	fit,
	backgrounds,
	shapes.misc
}
\usepackage{microtype}
\PassOptionsToPackage{hyphens}{url}
\usepackage{hyperref}
\hypersetup{breaklinks=true}

\makeatletter
\providecommand*{\toclevel@algorithm}{0}
\makeatother

\newcommand*{\versionswitch}[2]{\ifdefined\longversion #2\else #1\fi}
\def\longversion{}

\newcommand*{\citea}[1]{\citeauthor*{#1}~\cite{#1}}

\newcommand{\tch}[1]{\textcolor{blue5}{#1}}
\newcommand{\ch}{\color{blue5}}

\newcommand*{\Ps}{\ensuremath{P_s}\xspace}
\newcommand*{\Po}[1]{\ensuremath{P_{o_{#1}}}\xspace}

\newcommand*{\tmathc}[2][]{\ensuremath{\mathcal{#2}_{#1}}}
\newcommand*{\Aa}{\tmathc{A}\xspace}
\newcommand*{\Bb}{\tmathc{B}\xspace}
\newcommand*{\Cc}{\tmathc{C}\xspace}
\newcommand*{\Dd}{\tmathc{D}\xspace}
\newcommand*{\Ff}[1][]{\tmathc[#1]{F}\xspace}
\newcommand*{\Hh}[1][]{\ensuremath{\mathcal{H}#1}\xspace}
\newcommand*{\Ss}{\tmathc{S}\xspace}
\newcommand*{\Nn}[1][]{\tmathc[#1]{N}\xspace}
\newcommand*{\Oo}[1][]{\tmathc[#1]{O}\xspace}
\newcommand*{\Pp}[1][]{\tmathc[#1]{P}\xspace}
\newcommand*{\Rr}{\tmathc{R}\xspace}
\newcommand*{\Rrs}{\mathcal{R}}
\newcommand*{\Zz}{\tmathc{Z}\xspace}

\newcommand*{\Path}[2]{\ensuremath{(\Po{#1}, \ldots, \Po{#2})}\xspace}
\mathchardef\mhyphen="2D

\newcommand*{\fa}{\ensuremath{^{\rightarrow}}\xspace}
\newcommand*{\ba}{\ensuremath{^{\leftarrow}}\xspace}

\newcommand*{\Gg}{\emph{G}\xspace}
\newcommand*{\FormO}{\textsc{FormOnion}\xspace}
\newcommand*{\Proc}{\textsc{ProcOnion}\xspace}

\newcommand*{\Reply}{\textsc{FormReply}\xspace}
\newcommand*{\Recog}{\textsc{RecognizeOnion}\xspace}
\newcommand*{\RO}{\textsc{ROnion}\xspace}

\newcommand*{\FRN}{\Ff[\textit{RSOR}]}

\newcommand{\DocComment}[1]{\State \(\triangleright\) #1}
\newcommand{\CallPar}[2]{\textproc{#1}(#2)}

\newcommand*{\NewOnion}{\textsc{ProcessNewOnion}\xspace}
\newcommand*{\DelOnion}{\textsc{DeliverOnion}\xspace}
\newcommand*{\ForOnion}{\textsc{ForwardOnion}\xspace}

\newcommand*{\ToRelay}{\textsc{ProcessToRelay}\xspace}

\newcommand*{\DelPlain}{\textsc{DeliverMessage}\xspace}

\newcommand*{\MakeReply}{\textsc{SetupReply}\xspace}
\newcommand*{\LeakMess}{\textsc{LeakMessage}\xspace}

\newcommand*{\Tag}{\textsc{Tag}\xspace}
\newcommand*{\NewReply}{\textsc{ProcessNewReply}\xspace}
\newcommand*{\InitReply}{\textsc{InitiateReply}\xspace}
\newcommand*{\DelReply}{\textsc{DeliverReply}\xspace}
\newcommand*{\ByReply}{\textsc{BypassReply}\xspace}

\newcommand*{\LeakReply}{\textsc{LeakReply}\xspace}

\newcommand*{\Step}[1][]{%
	\upshape% Remove any font shape specifiers
	\textsc{ProcessNextStep}%
	\ifthenelse{\equal{#1}{}}% If empty, don't add a subscript
	{}%
	{$_{\!% Negative thin space to remove spacing from entering math mode after procedure name
		\ifthenelse{\equal{#1}{*}}% * (Asterisk) shouldn't be printed in italics
		{*}%
		{\textit{#1}}}$}%
	\xspace}

\newcommand*{\optsub}[1]{\ensuremath{\ifthenelse{\equal{#1}{}}{}{_{#1}}}}

\newcommand*{\nSMS}[1][]{\ensuremath{\ifthenelse{\equal{#1}{}}{(n, k)}{(#1)}\mhyphen{}\textit{SMR}\bar{S}}\xspace}
\newcommand*{\nSRS}[1][]{\ensuremath{\ifthenelse{\equal{#1}{}}{(n, k)}{(#1)}\mhyphen{}\textit{SR}\bar{S}}\xspace}

\newcommand*{\FLU}{\ensuremath{\textit{LU}\,\fa}\xspace}
\newcommand*{\BLU}{\ensuremath{\textit{LU}\,\ba}\xspace}
\newcommand*{\RTI}{\ensuremath{\textit{TI}\,^\leftrightarrow}\xspace}

\newcommand*{\TFLU}{\ensuremath{\textit{TLU}\,^\rightarrow}\xspace}
\newcommand*{\NBLU}{\ensuremath{\textit{SLU}\,\ba}\xspace}
\newcommand*{\NTI}{\ensuremath{\textit{STI}\,^\leftrightarrow}\xspace}

\newcommand*{\ROKEM}{Sphinx RO-KEM\xspace}
\newcommand*{\KEMCCA}{KEM-IND-CCA\xspace}
\newcommand*{\SKEMCCA}{Sphinx-KEM-IND-CCA\xspace}

\algblockdefx[Message]{Message}{EndMessage}
	[3]{\textbf{On message} \textsc{#1}(#2) \textit{from #3}}

\algtext*{EndMessage}

\algblockdefx[BlueProc]{BlueProc}{EndBlueProc}
	[2]{\textcolor{blue5}{\textbf{procedure}} \textsc{#1}(#2)}

\algtext*{EndBlueProc}

\newlength{\curindent}

\newlength{\oldabove}
\newlength{\oldashort}
\newlength{\oldbelow}
\newlength{\oldbshort}
\newcommand*{\compactequations}{
	\setlength{\oldabove}{\abovedisplayskip}
	\setlength{\abovedisplayskip}{0pt}

	\setlength{\oldashort}{\abovedisplayshortskip}
	\setlength{\abovedisplayshortskip}{0pt}

	\setlength{\oldbelow}{\belowdisplayskip}
	\setlength{\belowdisplayskip}{0pt}

	\setlength{\oldbshort}{\belowdisplayshortskip}
	\setlength{\belowdisplayshortskip}{0pt}
}
\newcommand*{\looseequations}{
	\setlength{\abovedisplayskip}{\oldabove}
	\setlength{\abovedisplayshortskip}{\oldashort}
	\setlength{\belowdisplayskip}{\oldbelow}
	\setlength{\belowdisplayshortskip}{\oldbshort}
}

\newtheorem{theorem}{Theorem}

\newtheorem{assumption}{Assumption}
\newtheoremstyle{enumdefinition}{}{}{}{}{\bfseries}{}{\newline}{}
\theoremstyle{enumdefinition}
\newtheorem{definition}{Definition}

\setcopyright{none}

\acmConference{}{}{}
\title[Provable Security for the Onion Routing and Mix Network Packet Format Sphinx]{Provable Security for the Onion Routing and\texorpdfstring{\\}{} Mix Network Packet Format Sphinx}
\author{Philip Scherer}
\affiliation{%
	\institution{KIT Karlsruhe}
	\city{}
	\country{}
}
\email{firstname.lastname@student.kit.edu}
\author{Christiane Weis}
\affiliation{%
	\institution{NEC Laboratories Europe}
	\city{}
	\country{}
}
\email{firstname.lastname@neclab.eu}
\author{Thorsten Strufe}
\affiliation{%
	\institution{KIT Karlsruhe}
	\city{}
	\country{}
}
\email{firstname.lastname@kit.edu}

\begin{document}
	\makeatletter
	\fancyhead[LE]{\@headfootfont\shorttitle}
	\fancyhead[RO]{\@headfootfont\@shortauthors}
	\makeatother
	\begin{abstract}
	\label{sec:abstract}
	%Topic
	Onion routing and mix networks
	are fundamental concepts to
	provide users
	with anonymous access
	to the Internet.
	% Problem
	Various corresponding solutions
	rely on the efficient
	Sphinx packet format.
	However,
	flaws in Sphinx's
	underlying proof strategy
	were found recently.
	It is thus currently unclear
	which guarantees Sphinx
	actually provides,
	and,
	even worse,
	there is no
	suitable proof strategy
	available.

	%Solution
	In this paper,
	we restore the security foundation
	for all these works
	by building a theoretical framework
	for Sphinx.
	We discover that
	the previously-used DDH assumption
	is insufficient for a security proof
	and show that
	the Gap Diffie-Hellman (GDH) assumption
	is required instead.
	We apply it
	to prove that a
	slightly adapted version
	of the Sphinx packet format
	is secure
	under the GDH assumption.
	Ours is the first work
	to provide a detailed,
	in-depth security proof
	for Sphinx in this manner.
	%Why good
	Our adaptations to Sphinx are necessary,
	as we demonstrate
	with an attack
	on sender privacy
	that would be possible otherwise.
\end{abstract}


	\maketitle
	\section{Introduction}
	\label{sec:introduction}

	The majority of today's Internet traffic 
	fails to protect
	the privacy of its users
	since the exposed IP addresses serve as identifying information.
	Onion routing (OR)%
	~\cite{onion_routing}
	and mix networks%
	~\cite{untrace-mail}
	are techniques
	that address this issue
	by hiding the users' IP addresses.
	OR and mix networks
	are similar
	in the sense that
	both use relays
	together with
	multiple layers of encryption
	to hide the link
	between the sender
	and the message
	and receiver.
	The sender
	wraps the message
	in multiple encryption layers,
	thus creating an \enquote{onion} packet.
	The layers
	are peeled off
	one by one
	while traversing
	the relays
	on the onion's path.
	As a critical distinction,
	mix networks
	protect against
	a global adversary by changing the order of packets 
	at each relay (hence \emph{mixing}
	the communications). This comes
	at the cost of
	introducing additional delays. OR networks avoid these delays, but become vulnerable to global passive adversaries.
	Proposed works for both
	OR and mix networks
	exist in two network models:
	In the integrated-system model,
	the receiver acts
	as the last relay.
	On the contrary,
	in the service model,
	the receiver is
	unaware of the OR or mix network.
	The last relay,
	which is called
	the exit relay,
	retrieves and forwards
	the message
	to the receiver
	as in the underlying protocol.

	To hide the link
	from senders
	to messages and receivers,
	protocol designers
	aim to make
	all incoming
	and outgoing packets
	at honest relays
	unlinkable for the adversary.
	This requires special care
	when designing the packet formats,
	as any part
	of the packet
	could include
	linkable information.
	Tor~\cite{tor} is being broadly applied against local adversaries but mix networks against global adversaries are increasingly developed.
	The most efficient
	and commonly used
	mix network packet format
	is Sphinx%
	~\cite{sphinx}.
	In fact,
	Sphinx is not only used
	as a foundation
	for mix networks%
	~\cite{loopix},
	but also for OR protocols%
	~\cite{hornet, taranet},
	and even inspired
	a recent improvement
	of Tor~\cite{walking_onions}.  	

	Sphinx~\cite{sphinx} works
	in the service model
	and assumes
	an additional party,
	the nymserver%
	\footnote{%
		While the name
		is the same,
		this server has
		nothing to do
		with the anonymization network
		\enquote{Nym}%
		~\cite{nym}.
	}.
	Senders send onions
	with reply information
	to the nymserver.
	Exit relays send
	anonymous reply messages
	via the nymserver,
	which uses the reply information
	to build
	a reply packet.
	Sphinx packets
	consist of a header
	and a payload.
	The header contains
	encrypted routing information
	and keys.
	The payload contains
	the encrypted message.
	Using only group exponentiations
	and well-known
	symmetric cryptography,
	Sphinx is highly efficient%
	~\cite{sphinx}.

	Sphinx's privacy has been
	proven previously
	using a proof strategy
	by \citeauthor{formal_onion}%
	~\cite{formal_onion}%
	\footnote{%
		In order to comply
		with ~\cite{formal_onion}
		and to honor
		Sphinx's applicability
		for OR protocols,
		we (ab-\hspace{0cm})use the onion terminology
		in the rest of this paper,
		while always meaning
		OR \emph{and} mix networks
		and considering
		a global adversary.
	}.
	Camenisch and Lysyanskaya's proof strategy
	first proposes
	an idealized version
	of OR
	in the form of
	an ideal functionality
	in the Universal Composability framework~\cite{uc}.
	This ideal functionality
	is effectively an abstract version
	of an OR protocol,
	from which privacy guarantees
	can be derived
	more readily.
	As proving that
	a protocol securely realizes
	this functionality directly
	is cumbersome,
	they also create
	a set of
	game-based properties
	which they claim
	imply realization
	of the ideal functionality%
	~\cite{formal_onion}.
	As it turns out,
	those properties,
	which were used
	in Sphinx's privacy proof,
	are insufficent
	to realize the ideal functionality%
	~\cite{break_onion}.

	While \citeauthor*{break_onion}
	propose new properties
	that indeed imply the ideal functionality%
	~\cite{break_onion},
	Sphinx is not able
	to achieve them
	for two reasons.
	First,
	Sphinx, along with
	many real-world applications,
	works in the service model,
	while the new proof strategy
	of \citeauthor*{break_onion}
	is in the integrated-system model.
	Secondly,
	Sphinx does not protect
	the integrity of the payload
	at each hop
	and thus allows
	for a malleability attack
	on the payload:
	If the adversary \emph{tags}
	(i.e.,
	flips bits of the payload)
	an onion
	leaving its sender,
	the exit relay processing the Sphinx packet
	will notice that
	the payload has been modified
	and drop the message.
	If Sphinx is used
	in the integrated-system model,
	this attack allows
	an adversarial receiver
	colluding with the first relay
	to learn which user
	was contacting it.
	As this violates
	the desired privacy goal,
	it follows
	that Sphinx
	does not achieve
	the integrated-system properties
	of the related work%
	~\cite{break_onion}.
	In the service model,
	this attack only allows
	an adversary to
	link the sender
	to the exit relay
	and completely destroys
	the message
	in the process.
	Hence,
	there is hope that the highly efficient
	Sphinx packet format is still secure to use
	as long as it is
	in its intended
	service model.
	Indeed,
	this question is highly relevant
	since all known protocols
	that prevent this attack
	while supporting
	anonymous replies
	incur extremely high overhead
	due to heavy,
	relatively new
	cryptographic primitives%
	~\cite{or_replies}.

	In addition
	to the problems
	noted above,
	we discover that
	the Decisional Diffie-Hellman (DDH) assumption
	used by Danezis and Goldberg
	to prove
	that Sphinx satisfies
	Camenisch and Lysyanskaya's game-based properties%
	~\cite{sphinx}
	is insufficient
	for Sphinx's security proof.

	In this paper,
	we hence set out to
	perform a thorough analysis
	and provide
	the missing privacy proofs
	for Sphinx.
	We first
	provide the
	necessary framework
	for the service model:
	A reusable game-based proof strategy
	which is of independent interest
	for future work
	on packet formats
	as well as
	as ideal functionality
	in the UC framework
	for use
	in analyzing the privacy guarantees
	of service-model OR protocols.
	We first define
	this new ideal functionality,
	which incorporates both the
	relaxed privacy accounting for
	payload malleability
	as well as
	the changes required
	for the service model.
	We then construct
	our game-based \enquote{onion properties}
	and prove that
	a protocol satisfying them
	implies that that protocol
	also realizes the ideal functionality.
	During the work
	on this proof,
	we also discover and fix mistakes and details
	in the proof 
	for the related work
	in the integrated-system model~\cite{or_replies}.

	Secondly,
	we turn to
	an analysis of Sphinx
	and realize that
	an adaptation
	must be made
	to the packet format
	and its operation
	in order to achieve
	provable security.
	As originally defined,
	Sphinx uses a nymserver
	to enable its reply functionality.
	However,
	the use of such third parties
	allows for
	an additional tagging attack
	based on
	payload malleability.
	For secure operation
	of Sphinx,
	we hence propose
	an adaptation
	of the Sphinx protocol
	that works without a nymserver,
	but still supports
	anonymous replies.
	This works by
	simply embedding the reply information
	in the packet's payload
	instead of sending it separately.
	Lastly,
	we discuss the effect
	of our privacy relaxation
	and detail criteria
	for the secure usage
	of Sphinx. 

	In summary,
	our main contributions are:
	\begin{itemize}
		\item the definition of
			repliable service OR schemes,
		\item the construction of
			an ideal functionality
			tailored to Sphinx
			as well as
			corresponding game-based properties,
		\item minor fixes
			in the proof
			for the related integrated-system model work,
		\item the discovery that
			the GDH assumption
			is required to prove Sphinx secure
			instead of the DDH assumption,
		\item the first detailed security proof
			for (a slightly adapted version of) Sphinx
			under the GDH assumption,
			and
		\item a discussion of criteria
			for secure usage
			of Sphinx.
	\end{itemize}

	\paragraph*{\textit{Outline}}
	\autoref{sec:background} introduces
	the required background
	on onion routing.
	\autoref{sec:reor} constructs
	the formal foundations
	for repliable OR in the service model,
	which 
	\autoref{part:sphinx} uses
	to analyze
	and adapt Sphinx.
	\autoref{sec:discussion}
	discusses the privacy
	achieved by
	our adapted Sphinx
	as well as
	relevant criteria
	under which
	the adapted Sphinx
	is considered secure.
	Finally,
	\autoref{sec:conclusion}
	concludes this paper.

	\section{Background and Related Work}
	\label{sec:background}

	We first introduce
	our privacy requirements,
	onion routing and mix networks
	and the general network model
	before providing background
	on the formal analysis
	in this work
	and the Sphinx packet format
	itself.

	\subsection{Onion Routing and Mix Network Packet Formats}
	Onion routing
	and mix networks
	aim to hide the sender
	of a packet
	in a set of users
	called the \emph{anonymity set}.
	They thereby
	prevent linking the sender
	both to the sent message
	and the receiver.  
	The networks
	employ multiple relays
	between the sender
	and receiver
	that process packets
	and forward them
	to the next relay
	or the receiver.
	As the message
	in the packet
	is typically wrapped
	into multiple layers
	of encryption,
	it is also called
	an \emph{onion}
	and every processing result
	on an onion's relay path
	is an \emph{onion layer}%
	~\cite{onion_routing}.

	Relay services
	are typically run
	by volunteers
	that want to
	help the sender
	increase its privacy
	against adversarial receivers,
	but also internet providers
	and possibly even
	nation-state adversaries.
	Due to this
	open nature
	of the network,
	it is assumed that
	a fraction of the relays
	is controlled by
	the adversary
	as well%
	~\cite{tor}.
	Mix networks
	aim to protect against
	a global adversary
	and therefore
	not only
	change the representation
	of the onion packet,
	but also reorder
	incoming packets
	before forwarding them.
	Onion routing networks,
	however,
	traditionally protect only
	against local adversaries
	and prioritize
	offering low-latency service
	over stronger protection%
	~\cite{tor}.
	In terms of
	the packet format,
	however,
	the networks are similar ---
	they share
	the underlying idea
	of layered encryption.
	For the sake of
	compatibility with related work,
	we use OR as a representative for both
	OR and mix network packet formats.
	We stress that
	we nonetheless target
	a global adversary.

	\subsection{Network Models and Functionalities}
	\label{ssec:nw}

	We distinguish between
	two network models
	for OR protocols:
	The integrated-system model
	assumes that
	the receiver is aware of
	and runs
	the OR protocol,
	while the service model
	assumes that
	the receiver can be
	unaware of
	the OR protocol.
	The last relay
	before the receiver,
	the \emph{exit relay},
	translates the packets accordingly.
	
	We also distinguish
	two different functionalities:
	\emph{Non-repliable} OR 
	only sends messages
	from senders to receivers
	in one direction and
	\emph{repliable} OR acts as
	two-way communication.
	Repliable OR schemes involve
	sending requests
	as \emph{forward} messages
	and receiving responses
	as \emph{reply} messages.

	By design,
	the anonymity of
	the sender
	is required to hold
	even against
	a malicious replying receiver.
	Thus,
	the receiver must not know
	to whom it sends its reply. 
	To achieve this,
	the sender can
	prepare an anonymized
	\enquote{return envelope}
	and include it in its onion.
	Some repliable OR protocols
	(including Sphinx) additionally
	require forward
	and reply packets
	to be indistinguishable
	from each other
	while they are
	moving through
	the network
	in order to
	increase the anonymity set
	of each packet%
	~\cite{sphinx}.
	
	\subsection{Formally Analyzing Mix Network Packet Formats} %Repliable Internal Onion Routing}
	\label{ssec:rior}

	The state of the art
	uses the
	Universal Composability (UC) framework~\cite{uc}
	to formalize OR security
	for packet formats.
	In UC,
	the desired security
	is defined by
	an ideal functionality \Ff,
	which is
	an optimal, abstract version
	of the protocol.
	\Ff performs all computations
	on a trusted third party
	and clearly specifies how
	the environment
	(which controls the honest parties)
	and the adversary
	are allowed to interact
	with that third party.
	The goal is then
	to show that a real protocol
	is as secure
	as \Ff
	and thus
	\emph{securely realizes} \Ff.
	This means
	all attacks
	against the real protocol
	must also work in \Ff
	(which is secure
	by definition).
	In this model,
	the environment controls
	all of the honest parties
	and the adversary controls
	the corrupted parties.
	To show that
	a protocol securely realizes \Ff,
	one constructs a simulator
	that translates the actions
	of the real-world adversary
	into the ideal-world \Ff 's
	adversarial entity's behavior
	and the ideal world's
	honest parties' actions
	into real-world
	protocol messages.
	The environment
	and the adversary
	are allowed to collaborate
	and the environment's view
	after the protocol ends
	is given to
	a distinguisher.
	Only if 
	the view of the environment
	when it is interacting with
	the simulator and  \Ff
	is indistinguishable from
	the view created from
	an interaction
	of the environment with
	the real-world adversary
	and the real protocol
	does the protocol
	securely realize \Ff%
	~\cite{uc}.

	\subsubsection{Overview of Analysis Works}
		
	The first approach
	to consider OR
	in the UC framework
	was proposed by
	Camenisch and Lysyanskaya.
	They create an ideal functionality
	and (as they claim)
	corresponding game-based
	security properties
	for the integrated-system model
	without replies%
	~\cite{formal_onion}.
	The reason for
	the creation of
	the game-based properties
	is that
	proving that
	a protocol securely realizes
	an ideal functionality
	is complex
	and involves multiple steps:
	One must construct
	a simulator for every adversary
	and then prove that
	the simulator simulates
	the adversary indistinguishably.
	Since many OR protocols
	operate in a similar way,
	Camenisch and Lysyanskaya
	reduce this overhead
	by defining their game-based properties.
	Their secure realization proof
	works for any OR protocol
	in their model
	that satisfies the game-based properties,
	taking advantage of
	the similarities
	between protocols.
	For authors of packet formats,
	this means that
	they only need to prove
	that their format
	satisfies the game-based properties
	in order to
	securely realize
	the ideal functionality%
	~\cite{formal_onion}.

	A series of protocols
	have based
	their security proofs
	on Camenisch and Lysyanskaya's
	game-based properties%
	~\cite{sphinx,hornet,taranet}.
	Kuhn et al.
	find flaws
	in the proposed properties
	as well as
	in the corresponding protocols%
	~\cite{break_onion}. 
	They fix
	the proof strategy
	by proposing
	new OR properties
	that imply the ideal functionality
	as originally proposed by
	Camenisch and Lysyanskaya.
	
	\citeauthor*{formal_replies}
	extend the ideal functionality
	and OR properties
	to include replies
	in the integrated-system model
	and propose
	a new OR scheme%
	~\cite{formal_replies}
	for this model. 

	\citea{or_replies}
	improve on
	\citeauthor*{formal_replies}'s work
	by improving the privacy in the schemes,
	ideal functionality,
	and the onion properties
	for repliable integrated-system-model OR.	

	\subsubsection{Summary of Integrated-System Framework}
	In the following,
	we present
	a high-level summary
	of the integrated-system framework%
	~\cite{or_replies},
	which we will adapt
	for our repliable service-model
	formalization.

	The ideal functionality \Ff[R]
	for repliable OR in
	the integrated-system model
	offers an interface
	to the environment and adversary
	that allows relays
	to send and forward onions.
	The honest and corrupted parties
	send messages
	to the trusted third party
	to trigger these actions.
	\Ff[R] then provides
	the appropriate outputs
	to the parties ---
	these correspond to
	the ideal outputs
	one would expect from
	a secure OR protocol.
	For example,
	the adversary is notified
	when an onion is forwarded
	by an honest relay
	and receives a temporary onion identifier
	that it can send back to \Ff[R]
	when it decides
	to deliver that onion
	to the next honest relay.
	However,
	\Ff[R]
	never constructs
	or outputs any \enquote{real} onions.
	Instead,
	the parties only receive
	information about the onions
	associated with
	the aforementioned
	temporary onion identifiers,
	These identifiers include
	no information about
	the sender,
	the receiver,
	the path,
	or the message
	of the respective onion.
	The identifiers are also
	replaced with
	a new random temporary identifier
	at each honest relay.
	This is done because
	any honest relay
	in an OR network
	breaks the link
	between the onion layers
	before and after itself.

	Proving secure UC-realization
	for each OR protocol individually
	involves a lot of redundant work.
	In line with earlier work,
	\citeauthor*{or_replies}
	thus define onion properties
	that are sufficient to prove
	that an
	integrated-system packet format
	securely realizes \Ff[R].
	These are
	\emph{Correctness},
	\emph{Forwards Layer Unlinkability} (\FLU),
	\emph{Backwards Layer Unlinkability} (\BLU),
	and \emph{Repliable Tail Indistinguishability} (\RTI)%
	~\cite{or_replies}.

	\emph{Correctness} requires that
	onions follow their set paths
	and decrypt correctly
	if they are honestly processed.

	In the \emph{Forwards Layer Unlinkability}
	(\FLU) game,
	onion layers
	between honest relays
	on the forward path
	of the onion
	are replaced with
	random onion layers
	taking the same path.
	The adversary is required
	to distinguish the replacement
	from the original.
	If the adversary
	cannot do so,
	the packet format ensures that
	onion layers
	on the forward path
	of an onion
	cannot be linked
	across an honest relay.

	In the \emph{Backwards Layer Unlinkability}
	(\BLU) game,
	onion layers
	on the \emph{reply path}
	(instead of
	the forward path)
	are replaced
	with random onion layers.
	Notably,
	it replaces them with
	random \emph{forward} onion layers.
	Again,
	the adversary must distinguish
	the replacement
	from the original.
	If it cannot do so,
	the packet format ensures
	unlinkability of reply onion layers.
	This property
	also implies indistinguishability of
	forward and reply onions.

	In the \emph{Repliable Tail Indistinguishability}
	(\RTI) game,
	onion layers
	going to a corrupted receiver relay
	are replaced with
	random onion layers
	with the same path
	and message contents.
	The adversary must
	once again tell the difference
	between the two.
	If it cannot,
	the packet format guarantees
	that different onions
	with the same message
	going to a corrupted receiver
	are indistinguishable.

	Note that,
	if all four properties
	are satisfied,
	every layer of an onion
	can be replaced
	with a random layer
	such that only
	the respective subpaths
	between honest relays
	stay the same
	(and the message,
	if the layer goes
	to a corrupted receiver).
	Effectively,
	this means that
	an adversary learns
	no more from
	any given sequence of onion layers
	than it would
	if given a temporary onion identifier
	and the corresponding subpath
	like in \Ff[R].

	These onion properties are used
	to enable the
	secure realization proof:
	When constructing the simulator
	that simulates the adversary,
	but interacts only with \Ff[R],
	the simulator
	does not learn
	the full path or message
	of onions
	from honest senders,
	but only the subpaths
	between two honest relays.
	This is by construction
	of \Ff[R].
	In order to translate
	the ideal-world onion
	into a real-world onion,
	it must thus replace
	the real onion layers
	that the adversary and environment
	would expect to see
	with random (forward) onion layers
	that only match
	the original onion
	in their subpath
	(and their message
	if an adversarial receiver
	gets the onion).
	The properties ensure
	that the replacement
	cannot be detected,
	allowing the proof
	to be completed
	in this manner.

\subsection{Sphinx}
	\label{ssec:sphinx}
	Sphinx~\cite{sphinx} is a
	compact repliable mix packet format
	in the service model
	following the adversary model
	and privacy goals
	described above.

	\subsubsection{The Sphinx Packet}
	\label{sssec:sphinx-packet}

	A Sphinx packet
	consists of a header $\eta$
	and a payload $\delta$.
	The header of the packet
	contains all of
	the routing information
	for the packet
	while the payload
	contains the message
	and the receiver address.
	The sender of the packet
	builds both components
	layer by layer,
	starting at the final
	innermost layer
	and adding one additional layer
	of encryption
	for each relay
	on the packet's \enquote{path}
	as chosen by
	the sender.
	If the path
	is shorter
	than the configured
	maximum path length,
	padding is added
	to the final layer
	of the header.
	The relays each remove
	one layer of that encryption
	when they process
	the packet.
	The final relay
	removes the last layer of encryption
	from the packet
	and discovers
	the message and receiver address.
	It then delivers
	the message
	to the receiver%
	~\cite{sphinx}.
	The packet format
	is described
	in more detail
	in \autoref{ssec:sphinx-packet-detail}.

	The separation
	of the header and payload
	allows for Sphinx's replies:
	A sender forms
	a repliable packet by
	creating a header
	for the reply
	as well as
	a key for the reply payload
	before sending
	the \enquote{forward} packet.
	The reply header and the key
	are sent to a third-party
	\textit{nymserver}
	using another forward packet
	and stored under a pseudonym there.
	After receiving a reply
	from the receiver,
	the exit relay of the packet
	sends the nymserver
	the reply message
	and the pseudonym
	it learned
	from the forward packet.
	The nymserver
	encrypts
	the reply payload
	using the key
	and attaches it
	to the reply header
	before sending it%
	~\cite{sphinx}.

	\subsubsection{Sphinx Security}
	\label{sssec:sphinx-insecure}

	Sphinx (as defined in ~\cite{sphinx})
	has two known flaws. 
	The first
	is due to the padding
	in the final layer
	of the header.
	In the original Sphinx definition,
	this padding
	consists of only 0 bits and
	depends on the chosen path length
	~\cite{break_onion, sphinx-impl}.
	As this pattern is recognizable,
	a corrupted exit relay
	learns information about the path length. 
	The Sphinx implementation~\cite{sphinx-impl}
	fixes this issue
	by using
	random bits instead.
	In the following,
	we consider
	the version of Sphinx
	that includes this fix.

	\paragraph{\textit{Payload Tagging}}
	Sphinx's second flaw
	concerns the integrity
	of the payload.
	Due to how Sphinx handles replies,
	it cannot use
	a standard integrity check 
	of the payload $\delta$
	at each relay%
	\footnote{
		Note that there are
		repliable OR schemes
		that solve this issue,
		but they are highly inefficient as they require
		significantly more complex
		cryptographic primitives%
		~\cite{or_replies}.
	}.
	Adversarial modifications
	of the payload thus
	go unnoticed
	until the packet
	reaches its exit relay,
	which notices
	(and drops)
	the modified packet
	during the integrity check
	on the payload.
	Crucially,
	this allows the adversary
	to link the packet
	it \enquote{tagged}
	(by, e.g.,
	flipping bits
	in its payload)
	to the packet
	that was dropped
	by the exit relay.
	The payload
	is completely destroyed
	in the process.
	If Sphinx is used
	in the integrated-system model,
	(as, e.g., in \cite{hornet}),
	this attack
	links sender and receivers
	and thereby completely breaks
	the scheme's security.
	However,
	if Sphinx is used
	in its intended service model,
	the attack only links
	the sender
	to its corresponding
	exit relay.
	This can be
	an acceptable leak
	in some settings,
	as we discuss later%
	\footnote{%
		Note that
		tagging the Sphinx header
		in this manner
		is not possible:
		The header's integrity
		is protected by a MAC
		that is checked
		at every relay
		for both forward
		and reply onions.
		The reason for
		this being easier to achieve
		than payload integrity
		is that
		the reply header
		can be built by the sender
		ahead of time,
		while the reply payload
		cannot (as the sender
		does not know
		the reply message).
	}.

	\section{Repliable Service Onion Routing}
	\label{sec:reor}

	We formalize repliable OR
	in the service model,
	hereafter referred to
	as repliable service OR
	(RSOR),
	and construct
	an ideal functionality
	and corresponding onion properties
	that allow for
	the payload malleability attack
	on Sphinx.
	We base
	our formalization
	on \citeauthor*{or_replies}'s work
	\cite{or_replies}.

\subsection{Defining Repliable Service Onion Routing}
	\label{sec:replies}

	\subsubsection{Notation}
	In general,
	we follow the notation
	of the corresponding related work
	(\cite{or_replies, sphinx})
	as much as possible.
	We represent
	messages with $m$,
	headers with $\eta$,
	and payloads,
	which include $m$
	along with other metadata,
	with $\delta$.
	A packet is the combination of  $\eta$ and $\delta$.
	We also refer to packets
	as onion layers
	and shorten the term \emph{onion layer}
	to \emph{onion} where appropriate.
	Onion paths $\mathcal{P}=(P_1,\dots, P_n)$
	consist of
	a sequence of relays
	with $P_i$ being
	the $i$-th relay's name
	and $P_n$ being
	the exit relay.
	$R$,
	the receiver,
	is not part of
	the onion's path.
	$PK_i$ and $SK_i$
	are the public and secret key
	of relay $P_i$.
	$O_i$ is the $i$-th onion layer,
	i.e.,
	the processing result
	that $P_{i-1}$ produces
	and sends to $P_i$.
	For reply information,
	we use the same notation
	with an additional
	superscript arrow:
	$x\ba$ indicates
	the reply counterpart
	to the forward component $x$.

	\subsubsection{Assumptions}
	\label{ssec:reor-assumptions}
	In this work,
	we make use of
	the following standard assumptions
	regarding the OR protocol,
	which we inherit
	from related work
	\cite{or_replies, break_onion}.
	Note that none of them
	require additional trust,
	but just limit
	the packet schemes
	our model applies to.
	To the best of our knowledge,
	every previously-proposed
	OR scheme and protocol
	adheres to
	these assumptions.

	\begin{assumption}
		\label{ass:path-length}
		A maximum path length (number of relays on the path)
		of $N$  is used (inclusive upper bound).
	\end{assumption}

	\begin{assumption}
		\label{ass:honest-sender}
		Honest senders choose
		acyclic paths
		(to increase the chance
		of picking
		at least one
		honest relay).
	\end{assumption}

	\begin{assumption}
		\label{ass:replays}
		Replay protection at honest relays
		drops onions
		whose headers
		are bit-for-bit identical
		to ones
		that have already been seen
		at that relay.
		Detecting these duplicates
		only fails with a
		negligible probability.
	\end{assumption}

	\begin{assumption}
		\label{ass:pki}
		A sender always knows
		the public keys $PK_i$
		of any relays $P_i$
		it uses
		for its onion's paths.
	\end{assumption}

	\begin{assumption}
		\label{ass:header-payload}
		An onion $O$
		consists of
		a header $\eta$
		and a payload $\delta$
		such that
		$O = (\eta, \delta)$.
	\end{assumption}

	\begin{assumption}
		\label{ass:secure-channels}
		Honest relays
		inside the OR network
		communicate with each other
		via secure channels%
		\footnote{%
			This assumption
			was implicitly introduced
			in the UC proof
			of%
			~\cite{break_onion}.
		}.
		Relays and receivers
		do not communicate
		via secure channels.
	\end{assumption}

	Assumption~\ref{ass:secure-channels}
	expands the secure-channel assumption
	for honest relays
	from~\cite{break_onion}
	by explicitly denying
	relay-receiver pairs
	from using secure channels
	to communicate
	since we do not know
	the state
	of the external network.

	Next,
	we introduce new assumptions
	related to the service model.

	\begin{assumption}
		\label{ass:crossfire-ignored}
		Receivers drop
		any onions
		sent to them.
		Similarly,
		relays drop
		any onions
		they get
		from links
		to receivers.
	\end{assumption}

	Since we assume
	that receivers
	are unaware
	of the OR network,
	they cannot process onions.
	Traffic between receivers
	is not part
	of our OR model.

	\begin{assumption}
		\label{ass:empty-paths}
		An onion
		in an RSOR packet format
		never has
		an empty forward path.
		If it is repliable,
		it does not have
		an empty backward path.
	\end{assumption}

	An onion
	with an empty forward path
	is effectively
	not an onion at all
	since the sender
	is also
	the exit relay.
	The packet format
	should thus not allow
	a valid onion
	to have an empty path.
	We use
	an empty backward path
	as a sentinel value
	for a non-repliable onion.

	\begin{assumption}
		\label{ass:unsolicited-replies}
		An honest relay
		will always drop
		an unsolicited reply
		(i.e.,
		a reply
		with a header
		the relay
		does not recognize
		as belonging to
		the final reply layer
		of an onion it sent).
	\end{assumption}

	This is not
	a surprising limitation.
	Honest relays
	will only process replies
	that they expect.
	Finally,
	we add one cryptographic assumption
	that we require
	for our game-based properties
	and UC realization proof:

	\begin{assumption}
		\label{ass:recreate-tag}
		Onion payloads
		are encrypted with
		a pseudorandom permutation (PRP).
	\end{assumption}

	With this assumption,
	any modification
	of a payload
	(i.e., a tagging attack)
	will completely randomize
	and thus destroy
	the payload contents,
	which we require
	for security against
	these attacks.

	\subsubsection{Formal RSOR Schemes}
	\label{ssec:reor-schemes}

	We build an RSOR scheme
	using the following algorithms (following
	~\cite{or_replies}):

	\compactequations
	\begin{itemize}
		\item Key generation algorithm $G$ for $P \in \Nn$:
			\[(PK, SK) \gets G(1^\lambda, p, P)\]

		\item Onion sending algorithm \FormO $(n, n\ba \leq N)$:
			\begin{gather*}
				O_i\! \gets\! \FormO(i, \bm{\Rrs}, m, \bm{R}, \Pp, \Pp\ba\!\!, PK_{\Pp}, PK_{\Pp\ba\!\!}), \\
				\Pp^{(\leftarrow)} = (P^{(\leftarrow)}_1, \ldots, P^{(\leftarrow)}_{n^{(\leftarrow)}}) \in \Nn^{n^{(\leftarrow)}}, \\
				PK_{\Pp^{(\leftarrow)}} = (PK^{(\leftarrow)}_1, \ldots, PK^{(\leftarrow)}_{n^{(\leftarrow)}}).
			\end{gather*}
			If $i > n$,
			$R$ is ignored
			and the onion $O_i$
			is formed like a reply
			with the message $m$.
			$\Pp\ba$ may be empty
			if no reply is desired.
			Otherwise,
			$\Pp\ba$ contains
			the \enquote{reply receiver}
			(which is the sender \Ps)
			as its final relay.

		\item Onion processing algorithm \Proc at $P$:
			\[(O', P') \gets \Proc(SK, O, P).\]
			$P$  processing $O$
			with its secret key $SK$
			results in an onion $O'$
			to send to
			the next relay $P'$.
			In case of an error,
			$(O', P') = (\bot, \bot)$.

			$(m, R) = \Proc(SK_E, O_n, E)$
			if $E$ is the exit relay of the onion.
			$m$ is the message for receiver $R$.\\
			$(m\ba, \bot) = \Proc(SK_s, O_{n+n\ba}, \Ps)$
			if sender \Ps
			(with private key $SK_s$)
			received the reply onion $O_{n+n\ba}$
			containing reply message
			 $m\ba$.

		\item Reply sending algorithm \Reply:
			\[(O\ba, P\ba) \gets \Reply(m\ba, O, E, SK).\]
			The reply onion $O\ba$
			to be sent to $P\ba$ is created
			from the reply message $m\ba$
			and the original onion $O$
			as processed by
			$O$'s exit relay $E$
			with the secret key $SK$.
			If an error occurs,
			$(O\ba, P\ba) = (\bot, \bot)$.
	\end{itemize}
	\looseequations

	\begin{definition}
		\label{def:rnre}
		An \emph{RSOR} scheme
		is a tuple of polynomial-time algorithms
		($G$, \FormO, \Proc, \Reply).
		%as defined above.
	\end{definition}

	Using an RSOR scheme,
	we can construct
	a corresponding protocol
	as follows:
	\compactequations
	\begin{enumerate}
		\item A sender
			selects the parameters
			for its onion,
			builds it
			using $\FormO$,
			and sends $O_1$
			to $P_1$.

		\item Each relay $P_i$
			processes the onion
			in turn
			using $\Proc$
			and sends $O_i$
			to $P_i$.

		\item The onion's exit relay $P_n$
			gets
			$(m, R)$ $\gets$ $\Proc($$SK_n$, $O_n$, $P_n)$.
			It generates
			a random \enquote{reply ID} $rid$,
			remembers $O_n$ in a map
			as $(rid, O_n)$,
			and sends $(m, rid)$
			to $R$.

		\item $R$ receives $(m, rid)$
			and decides
			to reply to the message.
			It sends its reply $(m\ba, rid)$
			back to $P_n$.

		\item $P_n$ gets the reply message
			and finds ($rid$, $O_n$)
			in its map.
			It calculates
			$(O\ba, P\ba) \gets \Reply(m\ba, O_n, P_n, SK_n)$
			and sends $O\ba$
			to $P\ba$.

		\item The reply onion
			follows its reply path
			like in step four
			until the sender receives
			$O\ba_{n\ba}$
			and processes it
			to receive $(m\ba, \bot)$.
	\end{enumerate}

	We introduce the concept
	of \enquote{reply IDs} ($rid$s)
	in the protocol description.
	These IDs are
	intended to be abstractions
	of the connections
	established by real-world
	network connection protocols
	like TCP ---
	To avoid the complexity
	actual connections
	would bring
	into our definitions,
	we represent them
	as a simple device
	that allows a receiver
	to send a reply
	directly to the relay
	it received the $rid$ from.
	Our protocols
	do not protect $rid$s,
	so they can be
	manipulated by adversaries
	(preventing this
	would require
	the establishment
	of secure channels
	between receivers and relays,
	which contradicts
	\autoref{ass:secure-channels}).

	In addition
	to the algorithms defined above,
	we also require
	the algorithm
	$\textsc{RecognizeOnion}($$i$, $O$, $\Rr$, $m$, $R$, $\Pp$, $\Pp\ba$, $PK_{\Pp}$, $PK_{\Pp\ba})$
	to be defined
	for RSOR
	analogously
	to its original definition
	by \citea{or_replies}:
	The algorithm compares
	the header of $O$
	to the $i$-th layer
	of the onion
	originally created
	from the given inputs.
	For our RSOR schemes,
	this means that
	the message
	is not checked
	since it is
	in the payload.

	\subsection{RSOR Ideal Functionalities}
	\label{sec:reor-ifs}
	In this section,
	we provide an overview
	of our ideal functionality \FRN
	for RSOR,
	focusing on the core ideas
	and issues
	in its construction.
	We base it on
	\citeauthor*{or_replies}'s
	\Ff[R]~\cite{or_replies}.
	We also explain
	the differences between
	\FRN and \Ff[R].
	\FRN
	is given in pseudocode
	in Appendix
	\ref{app:reor-if}.

	\subsubsection{Ideal Functionality}
	Recall from
	\autoref{ssec:rior}
	that,
	in an UC ideal functionality \Ff,
	all processing happens
	on a trusted party
	with a set of interfaces
	and procedures
	that the environment \Zz
	and the adversary \Ss
	interact with
	and receive information from.

	\paragraph{\textit{Fundamental concept}}
	Following the related work,
	onion layers are abstracted
	into a series of
	random temporary identifiers,
	the temp IDs ($tid$s).
	When an onion
	is forwarded through
	an honest relay,
	it receives a new $tid$,
	thus rendering
	the layers before the relay
	and after the relay unlinkable.
	Corrupted relays
	do not cause
	the $tid$
	to be replaced.
	As an onion
	is forwarded
	between honest relays,
	\Ss receives
	the $tid$s
	for the layers
	along with information
	on the layers' paths.
	This corresponds
	to the intuition that
	the adversary learns nothing
	about the content
	of the onion
	and cannot link
	the layers
	before and after
	an honest relay
	to each other.
	\FRN models
	all of the interactions
	between relays
	and between relays and receivers
	in the network
	including all possible
	adversarial capabilities
	in our adversary model.
	This includes corner cases
	and unusual adversary behavior.

	\paragraph{\textit{Detailed Interaction}}
	To send a new onion
	in \FRN,
	\Zz (for honest parties) or \Ss (for adversarial parties)
	notify \FRN
	with the desired receiver,
	message,
	and forward
	and reply paths.
	The resulting temp ID $tid$
	is given to \Ss,
	which controls the links
	and may choose to deliver
	the $tid$
	(i.e., the onion)
	to the next honest relay.
	If the onion
	was sent by a corrupted sender,
	\Ss receives
	all of the information
	on the onion
	for every path segment%
	\footnote{%
		This matches the behavior
		of real packet formats like Sphinx,
		where the sender
		calculates every layer itself
		and can thus recognize them.
	}.
	\Zz can decide
	when honest relays
	are done processing
	an onion and forward it
	to the next relay.

	While the onion
	is being forwarded
	through the network,
	\Ss can also choose
	to \emph{tag} it.
	This feature
	behaves like the tagging
	of a Sphinx packet
	(which is explained
	in \autoref{sssec:sphinx-insecure}).
	When the onion
	is forwarded
	by its last honest relay,
	it is either discarded
	(if it was tagged)
	or the message and receiver
	are leaked
	to \Ss.

	Reply handling
	comes in two variants:
	If the last honest relay
	is the exit relay
	of the onion,
	\FRN generates
	a reply ID $rid$
	(an abstraction
	of a connection
	that could be established
	by a protocol
	like TCP)
	and gives it to \Ss.
	\Ss can use
	the $rid$
	to provide the exit relay
	with a message
	for the reply onion.
	On the other hand,
	if the exit relay
	of the onion
	is corrupted,
	\Ss receives
	a $tid$
	it can use
	to send the reply onion
	with a reply message
	from any corrupted relay%
	\footnote{%
		A packet format
		cannot prevent
		the adversary
		from sending messages
		using any of the relays
		under its control.
	}.
	Since the channels
	between the relays
	and receivers
	are not secure,
	we have to assume that
	\Ss has complete control
	over message and reply delivery.
	Honest receivers
	can also initiate
	sending a reply
	to a message
	with an $rid$.
               
	\subsubsection{Comparison with Repliable Integrated-System OR}
	While we are able to
	reuse large parts
	of \citeauthor*{or_replies}'s \Ff[R]%
	~\cite{or_replies}
	to build \FRN,
	switching to the service model
	requires several additions
	to the existing \Ff[R].

	First,
	the new tagging feature
	marks onions
	to be discarded when
	they reach their
	last honest relay%
	\footnote{%
		Since the ideal functionality itself
		does not model
		adversarial processing,
		this is the case
		even when
		the last honest relay
		is not the exit relay.
		In that case,
		\Ss is provided with
		the remainder
		of the onion's path.
	}.
	Notably,
	tagging behaves asymmetrically:
	when a forward onion
	is tagged,
	the last honest relay
	discards it.
	On the reply path,
	a tagged onion
	is noticed
	by the reply receiver.
	If the reply receiver is honest,
	\Ss is not informed
	about the tag.

	Second, \FRN
	adds handling of
	message and reply delivery
	on the final link.
	Since these links
	are not secure,
	a real-life adversary
	is capable of
	manipulating the traffic
	on them in many ways.
	These include
	delivering messages
	and $rid$s to the wrong receivers
	or exit relays,
	swapping $rid$ and message pairs,
	impersonating
	exit relays to receivers
	and vice versa,
	and sending reply onions
	from other corrupted relays
	if the exit relay
	is corrupted.
	\FRN accounts for
	all of these
	adversarial behaviors
	in its interface
	to \Ss.

\subsection{RSOR Onion Properties}
	\label{sec:reor-props}

	Our properties are based on
	the onion properties
	originally defined
	by \citea{or_replies}
	with appropriate adjustments
	for the new algorithms
	and functionality.
	We detail one property
	and sketch
	the others here.
	For a formal definition
	of the other properties
	see Appendix~\ref{app:rnre-props-full}.

	\label{ssec:rnre-props}

	\subsubsection{RSOR-Correctness}
	This property requires that
	the scheme works as intended
	if no adversarial actions
	take place.
	Precisely,
	this means that,
	as the onion
	is processed using \Proc,
	each relay
	decrypts the correct
	address of the next relay
	and the next onion layer
	and the final layer
	decrypts to the message
	chosen by the sender.
	The same applies
	to the corresponding
	reply onion.

	\subsubsection{\texorpdfstring{Tagging-Forward Layer Unlinkability (See \autoref{fig:tflu})}{Tagging-Forward Layer Unlinkability}}
	\label{sssec:tflu}

	\TFLU is
	the RSOR equivalent
	of \FLU.
	It replaces onion layers
	between the honest sender
	and the first honest relay
	of an onion
	with random layers
	using that path segment.
	The introduction
	of tagging
	requires us to adjust
	this property.
	The original definition
	of \FLU
	provides the adversary
	with the challenge onion $O_1$
	or its replacement onion $\bar{O}_1$
	as produced
	by the sender
	and recognizes
	the processed layer $O_j$
	or $\bar{O}_j$
	when the adversary
	submits it
	to the processing oracle
	of the first honest relay.
	This recognition
	is based on
	the onion's header,
	so a tagged onion
	is still recognized.
	In the $b = 0$ case,
	the tagged $O_j$
	is processed normally
	and a tagged $O_{j+1}$
	is output to
	the adversary.
	However,
	in the $b = 1$ case,
	the tagged $\bar{O}_j$
	is recognized
	by the oracle
	and the original
	$O_{j+1}$ is output
	without being tagged.
	The adversary
	can simply finish
	the onion's processing
	to tell the difference.

	To alleviate this issue,
	\TFLU's oracle
	recognizes when
	the payload
	of the challenge onion
	has been tagged
	and \enquote{forwards}
	the tag
	by tagging $O_{j+1}$
	before sending it.
	If the honest relay
	is the exit relay,
	the oracle outputs nothing
	in this case.
	Since we assume that
	the payload is encrypted
	using a PRP,
	tagging the onion
	simply involves
	replacing the payload
	with randomness.

	\begin{figure}
		\centering
		\scalebox{0.72}{
		\begin{tikzpicture}
			[font=\footnotesize]
			\node[draw, red, minimum height=2cm] (P1) {$P_1$};
			\node[draw, left=of P1, minimum height=2cm] (Pr0) {\Ps};
			\node[anchor=north] (P0) at (P1.north) {\vphantom{$P_0$}};
			\node[anchor=south] (Pbkm1) at (P1.south) {\vphantom{$\bar{P}_0$}};

			\node[font=\small, left=1.7 of Pbkm1] (b) {(b)};
			\node[font=\small] (a) at (b.center |- P0.center) {(a)};

			\node[draw, red, right=0.5 of P1, minimum height=2cm] (Pjm1) {$P_{j-1}$};
			\node[draw, right=0.7 of Pjm1, minimum height=2cm] (Pj) {$P_j$};

			\node[draw, red, right=0.9 of Pj] (Pjp1) {$P_{j+1}$};
			\node[draw, red, right=0.6 of Pjp1] (Pnp1) {$P_{n+1}$};
			\node[draw, red, right=0.6 of Pnp1] (Pback) {$P\ba_{n\ba\!\!\!-1}$};
			\node[draw, right=0.6 of Pback] (P02) {\Ps};

			\node[below=0.2 of Pnp1] (m) {$(m, R)$};
			\node[below=0.2 of P02] (mback) {$m\ba$};

			\draw[->] (Pr0.east |- P0.center) -- node[auto] {$O_1$} (P1.west |- P0.center);
			\draw[->] (Pr0.east |- Pbkm1.center) -- node[auto] {$\bar{O}_1$} (P1.west |- Pbkm1.center);

			\draw[->, dashed] (P1.east |- P0.center) -- (Pjm1.west |- P0.center);
			\draw[->, dashed] (P1.east |- Pbkm1.center) -- (Pjm1.west |- Pbkm1.center);

			\draw[->] (Pjm1.east |- P0.center) -- node[auto] {$O_j$} (Pj.west |- P0.center);
			\draw[->] (Pjm1.east |- Pbkm1.center) -- node[auto] {$\bar{O}_{j}$} (Pj.west |- Pbkm1.center);

			\draw[->] (Pj) -- node[auto] {$O_{j+1}$} (Pjp1);

			\draw[->, dashed] (Pjp1) -- (Pnp1);
			\draw[->, dashed] (Pnp1) -- (Pback);
			\draw[->] (Pback) -- (P02);

			\draw[->] (Pnp1) -- (m);
			\draw[->] (P02) -- (mback);
		\end{tikzpicture}
		}
		\caption[The \emph{Forwards Layer Unlinkability} onion property]
		{
			The \TFLU
			onion property.
			The adversary is given
			either its chosen onion $O_1$ (a)
			or the random onion $\bar{O}_1$ (b)
			and must distinguish
			the two cases.
			Relays marked in
			\textcolor{red}{this style}
			are adversarial,
			while those
			in the normal style
			are honest.
			Omitted adversary-chosen paths
			are shown with dashed lines.
			The secondary output $O_{j+1}$
			after $P_j$
			is the same
			in both cases.
		}
		\label{fig:tflu}
	\end{figure}
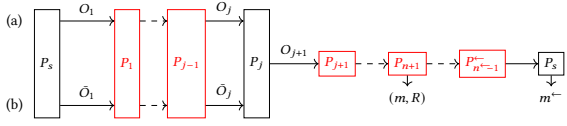

	The following full definition
	of \TFLU is derived from
	the original version
	by \citea{or_replies}.
	Our modifications
	are given in
	\tch{this style}.
	We also abbreviate
	\Recog as \RO
	for the sake of formatting.
	\begin{definition}
	\textit{{\ch Tagging}-Forward Layer Unlinkability}
	is defined as:

	\begin{enumerate}
		\item The adversary receives
			the router names $P_H$, \Ps
			and challenge public keys $PK_S$, $PK_H$,
			chosen by the challenger
			by letting $(PK_H, SK_H) \gets G(1^\lambda, p, P_H)$
			and $(PK_S, SK_S) \gets G(1^\lambda, p, \Ps)$.

		\item Oracle access:
			The adversary may submit
			any number of \textbf{Proc}
			and \textbf{Reply} requests
			for $P_H$ or \Ps
			to the challenger.
			For any $\textbf{Proc}(P_H, O)$,
			the challenger checks
			whether $\eta$ is
			on the $\eta_H$-list.
			If it is not on the list,
			it sends the output of
			$\Proc(SK_H, O, P_H)$,
			stores $\eta$
			on the $\eta_H$-list
			and $O$ on the $O_H$-list.
			For any $\textbf{Reply}(P_H, O, m)$,
			the challenger checks
			if $O$ is on the $O_H$-list
			and if so,
			the challenger sends ${\ch \Reply}(m, O, P_H, SK_H)$
			to the adversary.
			(Similar for requests
			on \Ps
			with the $\eta_S$-list).

		\item The adversary submits
			a message $m$,
			{\ch a receiver $R$,}
			a position $j$ with $1 \leq j \leq {\ch n}$,
			a path ${\ch\Pp} = (P_1, \ldots, P_j, \ldots, P_{\ch n})$
			with $P_j = P_H$,
			a path $\Pp\ba = (P_1\ba, \ldots, P\ba_{\ch n\ba} = \Ps)$
			and public keys
			for all relays $PK_i$
			($1 \leq i \leq {\ch n}$
			for the relays on the path
			and ${\ch n} < i$
			for the other relays).

		\item The challenger checks
			that the chosen paths are acyclic,
			the router names
			\tch{and public keys} are valid
			and that the same key is chosen
			if the router names are equal,
			and if so,
			sets $PK_j = PK_H$
			and $PK\ba_{\ch n\ba} = PK_S$
			and picks $b \in {0, 1}$
			at random.

		\item The challenger creates
			the onion $O_1$ with
			the adversary's input choice
			and honestly chosen randomness \Rr:
			\[{\ch \FormO(1, \Rr, m, R, \Pp, \Pp\ba, PK_{\Pp}, PK_{\Pp\ba})}\]
			and a replacement onion $\bar{O}_1$
			with the first part
			of the forward path ${\ch \bar{\Pp}} = (P_1, \ldots, P_j)$,
			a random message $\bar{m} \in M$,
			another honestly chosen randomness $\bar{\Rr}$,
			{\ch an honestly chosen random receiver $\bar{R}$},
			and an empty backward path $\bar{\Pp}\ba = ()$:
			\[{\ch \FormO(1, \bar{\Rr}, \bar{m}, \bar{R}, \bar{\Pp}, \bar{\Pp}\ba, PK_{\bar{\Pp}}, PK_{\bar{\Pp}\ba})}\]

		\item If $b = 0$,
			the challenger gives
			$O_1$ to the adversary. \\
			Otherwise,
			the challenger gives
			$\bar{O}_1$ to the adversary.

		\item Oracle access:
			If $b = 0$
			the challenger processes
			all oracle requests
			as in step 2). \\
			Otherwise,
			the challenger processes
			all requests
			as in step 2)
			except for:
			\begin{itemize}
				\item If $j < {\ch n}$:
					\begin{itemize}
						\item $\textbf{Proc}(P_H, O)$ with
							\[{\ch \RO(j, O, \bar{\Rr}, \bar{m}, \bar{R}, \bar{\Pp}, \bar{\Pp}\ba, PK_{\bar{\Pp}}, PK_{\bar{\Pp}\ba}),}\]
							{\ch and the expected payload $\delta_j$,}
							$\eta$ is not
							on the $\eta_H$-list
							and \[\Proc(SK_H, O, P_H) \neq {\ch (\bot, \bot)}\mathpunct{:}\]
							The challenger outputs $(P_{j+1}, O_c)$ with
							\begin{align*}
								O_c \gets {\ch \FormO(}{\ch j + 1, \Rr, m, R,}{\ch \Pp, \Pp\ba, PK_{\Pp}, PK_{\Pp\ba})}
							\end{align*}
							and adds $\eta$
							to the $\eta_H$-list
							and $O$
							to the $O_H$-list.

						\item $\textbf{Proc}(P_H, O)$ with
							\[{\ch \RO(j, O, \bar{\Rr}, \bar{m}, \bar{R}, \bar{\Pp}, \bar{\Pp}\ba, PK_{\bar{\Pp}}, PK_{\bar{\Pp}\ba})}\]
							{\ch but the incorrect payload $\delta'$,}
							$\eta$ is not
							on the $\eta_H$-list
							and \[\Proc(SK_H, O, P_H) \neq {\ch (\bot, \bot)}\mathpunct{:}\]
							The challenger outputs $(P_{j+1}, \ch{\tilde{O}_c})$
							with
							\begin{align*}
								O_c \gets {\ch \FormO(}{\ch j + 1, \Rr, m, R,}{\ch \Pp, \Pp\ba, PK_{\Pp}, PK_{\Pp\ba})}
							\end{align*}
							{\ch and $\tilde{O}_c$
							being $O_c$ with a tagged payload}
							and adds $\eta$
							to the $\eta_H$-list
							and $O$
							to the $O_H$-list.
					\end{itemize}

				\item If $j = {\ch n}$:
					\begin{itemize}
						\item $\textbf{Proc}(P_H, O)$
							with \[{\ch \RO(j, O, \bar{\Rr}, \bar{m}, \bar{R}, \bar{\Pp}, \bar{\Pp}\ba\!, PK_{\bar{\Pp}}, PK_{\bar{\Pp}\ba})},\]
							$\eta$ is not in the $\eta_H$-list
							and \[\Proc(SK_H, O, P_H) \neq {\ch (\bot, \bot)}\mathpunct{:}\]
							The challenger outputs $(m, {\ch R})$
							and adds $\eta$
							to the $\eta_H$-list
							and $O$
							to the $O_H$-list.

						\item $\textbf{Reply}(P_H, O, m\ba)$
							with \[{\ch \RO(j, O, \bar{\Rr}, \bar{m}, \bar{R}, \bar{\Pp}, \bar{\Pp}\ba\!, PK_{\bar{\Pp}}, PK_{\bar{\Pp}\ba})}\]
							$O$ is on the $O_H$-list
							and has not
							been replied before
							and \[{\ch \Reply}(m\ba, O, P_H, SK_H) \neq {\ch (\bot, \bot)}\mathpunct{:}\]
							The challenger outputs ${\ch (O_c, P_1\ba)}$
							with
							\begin{align*}
								O_c \gets {\ch \FormO(}{\ch j + 1, \Rr, m\ba, R,}{\ch \Pp, \Pp\ba, PK_{\Pp}, PK_{\Pp\ba}).}
							\end{align*}
					\end{itemize}
			\end{itemize}

		\item The adversary produces
			guess $b'$.
	\end{enumerate}
	${\ch \TFLU}$ is achieved
	if any PPT adversary \Aa
	cannot guess $b' = b$
	with a probability
	non-negligibly better than $\frac{1}{2}$.
	\end{definition}

	\subsubsection{\texorpdfstring{RSOR-Backwards Layer Unlinkability (See \autoref{fig:nblu})}{REOR-Backwards Layer Unlinkability}}

	The definition of
	the \NBLU property
	is analogous to
	the \TFLU property,
	but replaces the challenge onion
	on a path segment
	of the reply path.
	Additionally,
	the replacement onion
	is a forward onion,
	not a reply onion.
	This ensures that
	a scheme satisfying \NBLU
	will have
	indistinguishable
	forward and reply onions
	in the OR network.
	We do not need to
	adapt \NBLU
	to account for
	the tagging attack
	since the reply receiver
	processing oracle
	never produces output
	for onions
	with the challenge headers,
	so it cannot
	produce a wrong
	(untagged) output.

	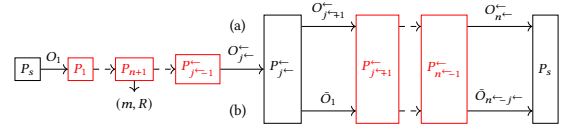
\begin{figure}
		\centering
		\scalebox{0.72}{
		\begin{tikzpicture}
			[font=\footnotesize]
			\node[draw, red, minimum height=2cm] (P1) {$P\ba_{j\ba\!\!\!+1}$};
			\node[draw, left=of P1, minimum height=2cm] (Pr0) {$P\ba_{j\ba}$};
			\node[anchor=north] (P0) at (P1.north) {\vphantom{$P_0$}};
			\node[anchor=south] (Pbkm1) at (P1.south) {\vphantom{$\bar{P}_0$}};

			\node[font=\small, left=2.2 of Pbkm1] (b) {(b)};
			\node[font=\small] (a) at (b.center |- P0.center) {(a)};

			\node[draw, red, right=0.4 of P1, minimum height=2cm] (Pjm1) {$P\ba_{n\ba\!\!\!-1}$};
			\node[draw, right=1.2 of Pjm1, minimum height=2cm] (Pj) {\Ps};

			\node[draw, red, left=0.8 of Pr0] (Pjp1) {$P\ba_{j\ba\!\!\!-1}$};
			\node[draw, red, left=0.4 of Pjp1] (Pnp1) {$P_{n+1}$};
			\node[draw, red, left=0.4 of Pnp1] (Pback) {$P_1$};
			\node[draw, left=0.5 of Pback] (P02) {\Ps};

			\node[below=0.2 of Pnp1] (m) {$(m, R)$};

			\draw[->] (Pr0.east |- P0.center) -- node[auto] {$O\ba_{j\ba\!\!\!+1}$} (P1.west |- P0.center);
			\draw[->] (Pr0.east |- Pbkm1.center) -- node[auto] {$\bar{O}_1$} (P1.west |- Pbkm1.center);

			\draw[->, dashed] (P1.east |- P0.center) -- (Pjm1.west |- P0.center);
			\draw[->, dashed] (P1.east |- Pbkm1.center) -- (Pjm1.west |- Pbkm1.center);

			\draw[->] (Pjm1.east |- P0.center) -- node[auto] {$O\ba_{n\ba}$} (Pj.west |- P0.center);
			\draw[->] (Pjm1.east |- Pbkm1.center) -- node[auto] {$\bar{O}_{n\ba\!\!\!-j\ba}$} (Pj.west |- Pbkm1.center);

			\draw[->] (Pjp1) -- node[auto] {$O\ba_{j\ba}$} (Pr0);

			\draw[->] (P02) -- node[auto] {$O_1$} (Pback);
			\draw[->, dashed] (Pback) -- (Pnp1);
			\draw[->, dashed] (Pnp1) -- (Pjp1);

			\draw[->] (Pnp1) -- (m);
			%\draw[->] (Pj) -- (mback);
		\end{tikzpicture}
		}
		\caption[The \emph{RSOR-Backwards Layer Unlinkability} onion property]
		{
			The \NBLU
			onion property.
			The adversary is initially given
			its chosen onion $O_1$.
			The oracle at the relay $P\ba_{j\ba}$
			will then return
			either its chosen onion $O\ba_{j\ba+1}$ (a)
			or the random onion $\bar{O}_1$ (b)
			and must distinguish
			the two scenarios.
			Relays marked in
			\textcolor{red}{this style}
			are adversarial,
			while those
			in the normal style
			are honest.
			Omitted adversary-chosen paths
			are shown with dashed lines.
		}
		\label{fig:nblu}
	\end{figure}

	\subsubsection{\texorpdfstring{RSOR-Tail Indistinguishability (See \autoref{fig:nti})}{REOR-Tail Indistinguishability}}

	The \NTI property
	is RSOR's counterpart
	for the \RTI property
	in the integrated-system model.
	It replaces the challenge onion
	with a random onion
	using the same path segment
	between an honest relay
	on the forward path
	and a second honest relay
	on the reply path.
	In \NTI,
	we do not allow
	the exit relay
	to be
	the honest relay
	since that situation
	is already covered
	by \NBLU.

	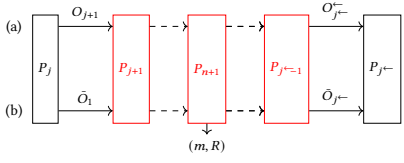
\begin{figure}
		\centering
		\scalebox{0.72}{
		\begin{tikzpicture}
			[font=\footnotesize]
			\node[draw, minimum height=2cm] (Pj) {$P_j$};
			\node[anchor=north] (Pjm1) at (Pj.north) {\vphantom{$P_{j-1}$}};
			\node[anchor=south] (Pbkm1) at (Pj.south) {\vphantom{$\bar{P}_{k-1}$}};

			\node[font=\small, left=0.2 of Pbkm1] (b) {(b)};
			\node[font=\small] (a) at (b.center |- Pjm1.center) {(a)};

			\node[draw, red, right=1 of Pj, minimum height=2cm] (Pjp1) {$P_{j+1}$};
			\node[draw, red, right=0.7 of Pjp1, minimum height=2cm] (Pnp1) {$P_{n+1}$};
			\node[draw, red, right=0.7 of Pnp1, minimum height=2cm] (Pback) {$P_{j\ba\!\!\!-1}$};
			\node[draw, right=1 of Pback, minimum height=2cm] (Pjback) {$P_{j\ba}$};

			\node[below=0.2 of Pnp1] (m) {$(m, R)$};

			\draw[->] (Pj.east |- P0.center) -- node[auto] {$O_{j+1}$} (Pjp1.west |- P0.center);
			\draw[->] (Pj.east |- Pbkm1.center) -- node[auto] {$\bar{O}_{1}$} (Pjp1.west |- Pbkm1.center);

			\draw[->, dashed] (Pjp1.east |- P0.center) -- (Pnp1.west |- P0.center);
			\draw[->, dashed] (Pjp1.east |- Pbkm1.center) -- (Pnp1.west |- Pbkm1.center);
			\draw[->, dashed] (Pnp1.east |- P0.center) -- (Pback.west |- P0.center);
			\draw[->, dashed] (Pnp1.east |- Pbkm1.center) -- (Pback.west |- Pbkm1.center);
			\draw[->, dashed] (Pnp1.east |- P0.center) -- (Pback.west |- P0.center);
			\draw[->, dashed] (Pnp1.east |- Pbkm1.center) -- (Pback.west |- Pbkm1.center);
			\draw[->] (Pback.east |- P0.center) -- node[auto] {$O\ba_{j\ba}$} (Pjback.west |- P0.center);
			\draw[->] (Pback.east |- Pbkm1.center) -- node[auto] {$\bar{O}_{j\ba}$} (Pjback.west |- Pbkm1.center);

			\draw[->] (Pnp1) -- (m);
		\end{tikzpicture}
		}
		\caption[The \emph{RSOR-Tail Indistinguishability} onion property]
		{
			The \NTI
			onion property.
			The adversary is given
			either its chosen onion $O_{j+1}$ (a)
			or the random onion $\bar{O}_{1}$ (b)
			and must distinguish
			the two cases.
			Relays marked in
			\textcolor{red}{this style}
			are adversarial,
			while those
			in the normal style
			are honest.
			Omitted adversary-chosen paths
			are shown with dashed lines.
		}
		\label{fig:nti}
	\end{figure}

	\subsubsection{Comparison with Integrated-System Properties}

	Besides the adaptions necessary
	for the different representation
	of the receiver
	and path,
	the main difference
	between the integrated-system properties
	and RSOR properties
	is the introduction
	of tagging for \TFLU.
	The other two properties
	do not require provisions
	for tagging mitigation:
	\NBLU remains unaffected
	because of
	the previously-described
	asymmetry of the tagging attack
	on the reply path,
	while \NTI
	never outputs
	a challenge onion
	from its oracles.

	\subsection{\texorpdfstring{Properties imply \FRN}{Properties imply FRSOR}}

	With the properties defined above,
	we can now define
	\emph{secure} RSOR schemes and protocols
	as those
	that satisfy the properties
	and behave
	in the way described
	in \autoref{ssec:reor-schemes}.

	\begin{definition}
		\label{def:sec-rnre-scheme}
		An RSOR scheme
		($G$, \FormO, \Proc, \Reply)
		that satisfies the four
		onion properties
		\textsl{RSOR-Correctness},
		\TFLU,
		\NBLU,
		and \NTI
		is a \emph{secure} RSOR scheme.
	\end{definition}

	\begin{definition}
		\label{def:sec-rnre-protocol}
		A \emph{secure RSOR protocol}
		is based on
		a secure RSOR scheme
		and behaves as defined
		in \autoref{ssec:reor-schemes}.
		A full definition
		of the protocol behavior
		is given in%
		\versionswitch{%
			~\cite{long-version}.
		}{%
			\ Appendix~\ref{app:reor-defs}.
		}
	\end{definition}

	\begin{theorem}
		\label{thm:rnre-props}
		A \emph{secure} RSOR protocol
		securely realizes \FRN.
	\end{theorem}

	The proof of this theorem
	is analogous to
	\citeauthor*{or_replies}'s original proof
	for integrated-system protocols
	and \Ff[R]%
	~\cite{or_replies}.
	We provide a brief sketch
	of the proof
	and the required changes here.
	The full proof
	is given in%
	\versionswitch{%
		~\cite{long-version}.
	}{%
		\ Appendix~\ref{app:frn-uc-proof}.
	}
	First,
	we construct a simulator \Ss
	that interacts
	with the adversary \Aa
	and \FRN,
	replicating
	the adversary's real-world actions
	in \FRN
	and vice versa.
	After that,
	we use
	a hybrid proof
	to show that
	any RSOR protocol
	that satisfies
	the RSOR properties
	is indistinguishable from
	our simulator
	interacting with \FRN.

	When interacting
	with \FRN and \Aa,
	\Ss must correctly translate
	the onions and messages
	sent by honest relays
	in \FRN into
	the real world
	without complete information
	on the contents
	or paths
	of the onions.
	The RSOR properties
	ensure that
	the simulator can create
	replacement onions
	with randomized contents
	and truncated paths
	whenever it needs
	to do so
	without \Aa or \Zz
	noticing the replacement.

	\subsubsection*{\texorpdfstring{Differences to \Ff[R]}{Differences to FR}}
	Introducing tagging
	requires additional logic
	in the simulator.
	When the adversary tags
	an onion
	from an honest sender,
	the simulator can tell
	that the payload
	has been manipulated
	and tag the onion
	in \FRN.
	If the onion
	is from a corrupted sender,
	\Ss notices the manipulation
	when it processes the onion
	at its honest exit relay
	and can tag it
	in \FRN
	at that point.
	This behavior
	allows the simulator
	to handle tagging attacks
	correctly.

	In addition,
	\Ss now also
	needs to handle
	communication with
	external receivers
	involving messages
	and reply IDs.
	Translating these
	between \FRN
	and the real world
	involves forwarding
	the appropriate communications
	while potentially
	impersonating exit relays
	or receivers.

	In performing our
	hybrid machine construction,
	we fix an error
	in \citeauthor*{or_replies}'s proof:
	They apply the \FLU
	(and \BLU) properties
	to path segments
	that do not start
	at the honest sender
	(or end at
	the honest reply receiver)%
	~\cite{or_replies}.
	However,
	the properties
	do not apply
	to these situations.
	We fix this problem
	by applying the \NTI property
	once
	in these situations
	to truncate the paths
	of the corresponding onions.
	Afterwards,
	we can apply
	the \TFLU or \NBLU properties
	as before.
	The same change
	(using \RTI)
	can be used to
	repair the proof
	in~\cite{or_replies}.

	\section{Analyzing Sphinx's Security}
\label{part:sphinx}

	We now aim to show that the Sphinx securely realizes
	\FRN.
	Therefore,
	we use our RSOR properties
	and show that
	Sphinx satisfies each of them.
	However,
	to be able to do so,
	we have to make
	two further adjustments
	to the Sphinx protocol:
	Firstly,
	we remove the nymserver,
	as it allows
	for attacks
	(in \autoref{sec:nymserver}).
	Secondly,
	we rely on
	a slightly stronger
	cryptographic assumption
	for the proof
	(in \autoref{ssec:GapDH}).

	\subsection{The Sphinx Packet}
	\label{ssec:sphinx-packet-detail}

	A Sphinx packet
	consists of
	a header
	$\eta = (\alpha, \beta, \gamma)$
	and a payload $\delta$
	(see \autoref{sssec:sphinx-packet}
	for an overview).
	$\alpha$ contains
	the shared secret keys
	encrypted for each relay,
	$\beta$ holds
	the routing information
	padded by the sender,
	and $\gamma$ is a MAC
	for $\beta$.

	Each relay $P_i$ has
	an asymmetric key pair
	of the form
	$(x_i, y_i := g^{x_i})$
	with $x_i \in \mathbb{Z}_q^*$
	and $g$ as a generator
	of a cyclic group $\mathcal{G}$
	with the prime order $q$.
	$q$ should be approximately
	$2^{2\kappa}$,
	where $\kappa$ is
	the security parameter%
	~\cite{sphinx}.
	The public keys
	are used along with
	random oracles
	in order to encapsulate
	the shared secret keys used
	for the symmetric cryptographic primitives
	in each layer
	of the Sphinx packet.
	We refer to
	this part
	of the packet format
	as Sphinx's
	random oracle-key encapsulation mechanism (RO-KEM).

	The RO-KEM's ciphertext
	is $\alpha$,
	the first component
	of the Sphinx packet header%
	~\cite{sphinx}.
	$\alpha$ is formed
	using the public keys
	of the relays
	on the sender's chosen path.
	First,
	the sender chooses
	a secret $x \in \mathbb{Z}_q^*$.
	$x$ is the only source
	of randomness
	in the Sphinx packet.
	Using $x$,
	the sender generates
	the $\alpha$-key encapsulation
	for the first onion layer
	(which is layer 0):
	$\alpha_0 = g^x$.
	The sender
	then derives
	the shared secret
	for the first relay:
	$s_0 = y_0^x$,
	where $y_0$ is the public key
	of the first relay%
	\footnote{%
		Note that
		our indices $i$
		for relay names
		are relative to
		a single packet's path
		for clarity.
	}.
	When the first layer
	of the onion reaches
	the first relay,
	it can derive $s_0$
	using its secret key:
	$s_0 = \alpha_0^{x_0}$.
	The keys and random values
	used in the
	cryptographic primitives
	in the Sphinx packet
	are all derived from $s_0$
	using the random oracles
	$h_b$,
	$h_\rho$,
	$h_\mu$,
	and $h_\pi$%
	~\cite{sphinx}.

	In particular,
	the random oracle $h_b$
	is used to build
	the key encapsulations
	for the following onion layers:
	To calculate $\alpha_1$,
	the sender lets
	$b_0 = h_b(\alpha_0, s_0)$
	and $\alpha_1 = g^{xb_0}$.
	The corresponding shared secret
	is similarly calculated
	as $s_1 = y_1^{xb_0}$,
	using the public key
	of the second relay
	on the onion's path.
	The $b_i$ are referred to
	as \emph{blinding factors}.
	When the first relay
	processes the first onion layer,
	it can calculate
	$\alpha_1 = \alpha_0^{b_0}$
	with the blinding factor
	$b_0$ that it gets from
	the same random oracle $h_b$.
	With this,
	the second relay
	can derive its shared secret
	using its private key.
	This process is repeated
	for each onion layer.
	Note that
	only $\alpha_0$
	is included in
	the \enquote{final first onion layer} ---
	the later $\alpha_i$
	are calculated
	by the relays themselves
	and replace
	the respective $\alpha_{i-1}$
	in the packet
	during processing%
	~\cite{sphinx}.

	Since we focus
	on Sphinx's RO-KEM
	in the following sections,
	we only give
	a short overview
	of the remaining
	packet structure here.
	For a more detailed
	and technical description,
	see Appendix~\ref{sapp:sphinx-packet-format}.
	
	The other two components
	of the Sphinx header
	are $\beta$ and $\gamma$.
	In each onion layer $i$,
	$\gamma_i$ is simply
	a MAC $\mu$ of $\beta_i$
	keyed with $h_\mu(s_i)$.
	Here,
	$h_\mu$ is
	one of the random oracles
	keyed with
	that layer's shared secret.
	$\beta_i$ contains
	the address of
	the next relay $P_{i+1}$,
	the next MAC $\gamma_{i+1}$,
	and a prefix
	of the next $\beta_{i+1}$
	in that order.
	This information
	is XORed with the output
	of a PRG $\rho$
	keyed with $h_\rho(s_i)$.
	When processing
	an onion layer,
	the relay $P_i$
	creates $\beta_{i+1}$
	from $\beta_i$
	by appending a string
	of zero bits
	to $\beta_i$,
	XORing the PRG output
	onto the result,
	and cutting off
	the relay address
	and MAC
	at the start
	of the end result.
	The construction
	of Sphinx's padding scheme
	means that
	the XORing the zero-bit extension
	with the PRG output
	results in
	exactly the missing suffix
	of $\beta_{i+1}$.

	Finally,
	the Sphinx packet payload $\delta$
	is simply constructed
	thorugh multiple layers
	of encryption
	with a pseudorandom permutation (PRP) $\pi$,
	which is keyed with
	$h_\pi(s_i)$
	in the $i$-th layer.
	The payload
	is constructed by the sender
	in reverse order,
	layering encryption
	from the final onion layer
	to the first.
	The innermost layer
	contains
	a zero padding
	of length equal
	to the security parameter,
	the receiver address,
	and the message
	in that order.
	When a relay
	processes an onion,
	it simply removes
	one layer of encryption
	from the payload.
	The final relay
	performs an integrity check
	by checking that
	the zero padding
	at the beginning of the payload
	is intact.

	\subsection{Nymserver}
	\label{sec:nymserver}
	The original Sphinx definition
	by \citea{sphinx}
	uses a \emph{nymserver}
	to hold the reply headers
	for the onions
	in the network.
	To create a repliable onion,
	the sender first sends
	a non-repliable onion
	containing the reply header
	and a symmetric key
	for the payload
	to the nymserver
	under a pseudonym.
	The sender is responsible
	for embedding that pseudonym
	in the forward onion
	for the exit relay
	to find.
	If the receiver
	decides to reply,
	it sends its reply message
	back to the exit relay.
	The exit relay
	sends the reply message
	and the pseudonym
	to the nymserver.
	The nymserver
	finds the reply header
	associated with
	the pseudonym
	in its database,
	encrypts the payload,
	attaches it
	to the reply header,
	and sends the reply onion%
	~\cite{sphinx}.

	This nymserver construction
	is insecure in the presence
	of an attacker that can
	tag or drop onions
	and controls the nymserver
	or observes its traffic:
	A sender that wants to
	send a repliable onion
	actually sends two onions,
	one going to the nymserver.
	The attacker tags (or drops)
	one of these onions,
	hoping that it picked
	the nymserver onion.
	If it is successful,
	the reply header
	is not stored
	in the nymserver.
	When the exit relay
	(on behalf of the receiver,
	which is trivally linkable
	to the exit relay)
	sends the message
	and pseudonym
	to the nymserver,
	the attacker can observe
	that no onion
	is produced
	by the nymserver%
	\footnote{%
		If the attacker
		controls the nymserver itself,
		it can simply see
		that the requested pseudonym
		does not exist.
	},
	thus learning the connection
	between the sender
	and the receiver.

	To fix this problem
	and create a version of Sphinx
	that can be proven secure,
	we adapt Sphinx
	to include the reply header
	and payload symmetric key
	in the forward onion payload
	directly%
	\footnote{%
		Note that
		there are other feasible mitigations
		for this attack,
		e.g.,
		sending multiple reply headers
		ahead of time
		to make
		linking the missing reply header
		to the original sender
		more difficult
		for an attacker.
		However,
		our adaptation
		both completely prevents the attack
		and simplifies
		the Sphinx protocol
		while being very simple itself.
		We thus consider it
		the most appropriate fix
		for the problem.
	}.
	In the original definition
	of Sphinx,
	a forward payload
	has the contents
	$0_\kappa\|R\|m$
	and a reply payload
	has the contents
	$0_\kappa\|m\ba$%
	\footnote{%
		Messages are always padded
		to the full length
		of the PRP's message space.
	}%
	~\cite{sphinx}.
	In our adaptation,
	forward payloads
	are formed as
	$0_\kappa\|R\|\eta_0\|\tilde{k}\|m$
	with $R$
	as the receiver address,
	$\eta_0$ as the reply header,
	and $\tilde{k}$
	as the symmetric key.
	Non-repliable forward payloads
	and reply payloads
	contain
	zero paddings
	$\text{pad}\fa_{\kappa,N}$
	and $\text{pad}\ba_{\kappa,N}$
	of the appropriate lengths
	instead of
	$\eta_0\|\tilde{\kappa}$
	and $R\|\eta_0\|\tilde{\kappa}$
	respectively.
	This means
	a non-repliable forward payload
	contains
	$0_\kappa\|R\|\text{pad}\fa_{\kappa,N}\|m$
	and a reply payload
	contains
	$0_\kappa\|\text{pad}\ba_{\kappa,N}\|m\ba$.
	This change fixes
	the issue
	while removing a third party
	from the protocol
	and simplifying Sphinx.
	With this change,
	the exit relay
	is now responsible for
	taking the reply message
	from the receiver,
	embedding it
	in the reply payload,
	and sending
	the completed reply onion.

	\subsection{Sphinx Key Encapsulation Mechanism (KEM)}
	Before we move on
	to the RSOR properties,
	we discuss the security
	of Sphinx's random oracle (RO-)KEM,
	which is used
	to form the $\alpha$
	in each header.
	To simplify its analysis
	in isolation,
	we define the KEM
	separately from
	the rest of the packet format
	and prove that
	it satisfies
	a modified version
	of the IND-CCA property
	for KEMs (\KEMCCA)
	as defined by
	\citea{cramer-shoup}%
	\footnote{%
		The property
		we use here
		is defined
		in Section 7.1.2.
	}
	where the challenger
	outputs additional information
	that we require
	in our later onion property proofs.
	In those proofs,
	we will make use
	of our \SKEMCCA property
	in order to
	randomize the blinding factors
	and symmetric keys
	used in the challenge onion.
	In this section,
	we abbreviate the concatenation
	of the three random oracles
	$h_\rho$,
	$h_\mu$,
	and $h_\pi$
	as $h_*$
	for legibility
	since they operate identically
	with regards to the KEM.

	\begin{definition}[\ROKEM]
		The \ROKEM
		is a three-tuple of polynomial-time algorithms
		(\textsc{KeyGen}, \textsc{Encap}, \textsc{Decap})
		with:

		\begin{itemize}
			\item Key generation:
				\[\textsc{KeyGen}(1^\kappa) := (PK = g^x, SK = x)\]
				with $x \gets^R \mathbb{Z}^*_q$
				and $g$ as the public generator
				of the group $\mathcal{G}$.
			\item Encapsulation:
				\begin{equation*}
					\textsc{Encap}(1^\kappa, PK=g^x) := ((h_*(PK^{x'}), h_b(g^{x'}, PK^{x'})), g^{x'})
				\end{equation*}
				for a random $x' \in \mathbb{Z}^*_q$,
				with $h_*$ and $h_b$
				being the random oracles
				used to key the components
				of the Sphinx header.
			\item Decapsulation:
				\begin{equation*}
					\textsc{Decap}(1^\kappa, SK=x, \alpha=g^{x'}) := (h_*(\alpha^{SK}), h_b(\alpha, \alpha^{SK})),
				\end{equation*}
				where $\alpha$ is an encapsulation
				produced by \textsc{Encap}%
				~\cite{sphinx}.
		\end{itemize}

	\end{definition}

	\subsubsection{\SKEMCCA}
	The basic \KEMCCA game
	as defined by
	\citea{cramer-shoup}
	is unfortunately insufficient
	for our later proofs,
	where we require information
	on the RO outputs used in
	other (non-challenge) layers.
	We thus define
	a modified \SKEMCCA
	that outputs all of the information
	required to build
	a Sphinx packet
	while embedding the KEM challenge
	at an adversary-chosen index.

	\begin{definition}[\SKEMCCA]
		\begin{enumerate}
			\item[]
			\item First, the challenger chooses
				$(PK, SK) \gets \textsc{KeyGen}(1^\kappa)$
				and sends $PK$
				to the adversary.
			\item Oracle access:
				The adversary can submit requests
				to the decapsulation oracle \Oo
				and the random oracles
				$h_*$ and $h_b$.
			\item The adversary submits
				\begin{itemize}
					\item $n - 1$ public keys
						$y_0, \ldots, y_{j-1}, y_{j+1}, \ldots, y_{n-1}$
						with $n < N$.
						These are the public keys
						for the non-honest relays
						on the \enquote{KEM's path},
					\item and a position $j$ with $0 \leq j < n$.
				\end{itemize}
			\item The challenger checks
				that the $y_i$ are all distinct
				and valid public keys.
			\item The challenger creates
				the KEM challenge
				for the adversary
				by choosing a random
				$x' \in \mathbb{Z}^*_q$
				and generating
				the first $j$ encapsulations
				$\alpha_0$ through $\alpha_{j-1}$
				and secrets $s_i$
				like for a Sphinx header:
				$\alpha_i \gets g^{x'b_0\cdots b_{i-1}}$,
				$s_i \gets y_i^{x'b_0\cdots b_{i-1}}$,
				$b_i \gets h_b(\alpha_i, s_i)$.
			\item The challenger sends
				the adversary its \enquote{auxiliary information}%
				\footnote{%
					Note that this is
					all of the information
					required to build any layer
					of the Sphinx packet
					preceding the
					challenge layer $j$.
				}:
				\begin{itemize}
					\item The first encapulation $\alpha_0$,
					\item the $h_*$ outputs
						$h_*(s_0)$, \ldots, $h_*(s_{j-1})$,
					\item and the blinding factors
						$b_0$, \ldots, $b_{j-1}$.
				\end{itemize}%
			\item The challenger provides
				the adversary with
				the KEM challenge:
				It picks $b \in \{0,1\}$
				at random.
				If $b = 0$,
				the challenger lets
				$b_j = h_b(\alpha_j, s_j)$
				and gives the adversary
				$(\alpha_j, h_*(s_j), b_j)$.
				Otherwise,
				the adversary gets
				$(\alpha_j, r_1, b_j)$
				for $r_1 \gets^R \{0,1\}^{3\kappa}$
				and $b_j \gets^R \mathbb{Z}^*_q$.
				Finally,
				the challenger generates
				the rest of the KEM layers
				$\alpha_{j+1}$, \ldots, $\alpha_n$
				with $s_{j+1}$, \ldots, $s_n$
				and $b_{j+1}$, \ldots, $b_n$
				the same way as the previous layers
				(using the corresponding $b_j$)
				and gives the adversary
				$h_*(s_{j+1})$, \ldots, $h_*(s_n)$
				and $b_{j+1}$, \ldots, $b_n$.
			\item Oracle access:
				The adversary gets access
				to the same \Oo,
				$h_*$, and $h_b$ oracles.
			\item The adversary submits
				its guess $b'$
				to the challenger.
		\end{enumerate}
		\SKEMCCA is achieved
		if any PPT adversary \Aa
		cannot guess $b' = b$
		with a probability
		non-negligibly better
		than $\frac{1}{2}$.
	\end{definition}

	\subsubsection{Security}
	\label{ssec:GapDH}
	In order to show
	that the \ROKEM
	satisfies our security property,
	we can perform
	a reduction proof
	to a Diffie-Hellman assumption.
	The original Sphinx definition
	uses the DDH assumption
	for the \ROKEM
	in order to show
	that the blinding factors
	and symmetric keys
	are indistinguishable
	from randomness
	for the adversary%
	~\cite{sphinx}.
	However,
	a more detailed analysis
	reveals that
	the DDH assumption
	is insufficient
	to prove the \ROKEM secure:
	In a reduction
	from \SKEMCCA
	to the DDH assumption,
	the DDH attacker must
	simulate both
	the decapsulation oracle \Oo
	as well as the random oracles
	$h_*$ and $h_b$
	consistently
	for the \SKEMCCA attacker.
	Doing so correctly
	for adversary-chosen inputs
	involves identifying
	which encapsulations $\alpha$
	and secrets $s$
	belong together.
	Since being able
	to do so efficiently
	would already break DDH,
	the redution
	is not possible
	in this form.
	We use
	the \emph{Gap Diffie-Hellman} assumption instead.
	It states that
	the CDH problem is hard
	even given an oracle
	that solves the DDH problem.
	It is generally assumed
	that the GDH assumption
	holds in the standard elliptic curve groups,
	which Sphinx already uses
	~\cite{gap-problems}.
	For our reduction,
	this means
	that we are reducing
	\SKEMCCA to the CDH problem,
	but our CDH attacker
	additionally receives
	a DDH oracle \Oo[G].
	Using \Oo[G],
	the CDH attacker
	can correctly identify
	matching secrets
	and encapsulations.
	It follows that:

	\begin{theorem}
		The \ROKEM
		satisfies \SKEMCCA
		under the GDH assumption.
	\end{theorem}

	\begin{proof}[Proof Sketch]
		We use a \SKEMCCA attacker \Aa
		on the \ROKEM
		to construct
		a GDH attacker $\Bb^{\Oo[G]}$
		with the DDH oracle
		\Oo[G].
		\Bb uses its CDH challenge
		$(g, g^{x_1}, g^{x_2})$
		as a public key $PK = g^{x_1}$
		and a challenge $\alpha = g^{x_2}$.
		After getting
		the challenge index $j$
		and the other public keys
		from \Aa,
		\Bb sets
		$\alpha_j = \alpha$
		and constructs
		the rest of the KEM's path
		by choosing
		random blinding factors $b_i$
		and $h_*$ outputs
		for every layer,
		programming its choices
		into the ROs.
		\Bb also randomly chooses
		a bit $b \in \{0,1\}$
		to determine which scenario
		it simulates.
		The only difference
		between the two
		is whether
		the random oracle outputs
		on layer $j$
		are programmed
		into the ROs.
		\Bb then simulates
		the oracles,
		keeping the outputs consistent
		using its DDH oracle \Oo[G].
		In order for \Aa
		to tell the scenarios apart,
		it must request
		$h_*$ or $h_b$
		with $g^{x_1x_2}$,
		allowing \Bb
		to win the GDH game.
		For the full proof see Appendix~\ref{app:Sphinx-KEM-Proof}.
		\renewcommand{\qedsymbol}{}
	\end{proof}

	\subsection{Sphinx Security Analysis}	

	In order to prove
	that Sphinx
	securely realizes \FRN,
	we show
	that it satisfies
	our new properties.
	For the sake of brevity,
	we only sketch
	the proof for \TFLU here
	along with short proof outlines
	for \NBLU and \NTI.
	\versionswitch{%
		Extended explanations
		of the proofs
		are given in
		Appendix~\ref{sapp:nbluproof}
		and Appendix~\ref{sapp:ntiproof}.
		The full proofs for \NBLU
		and \NTI
		can be found in%
		~\cite{long-version}.
	}{%
		The full proofs for
		all three properties
		can be found in
		Appendix~\ref{sapp:tfluproof},
		Appendix~\ref{sapp:nbluproof}
		and Appendix~\ref{sapp:ntiproof}.
	}
	RSOR-Correctness follows
	from inspection
	of the Sphinx scheme.

	\begin{theorem}
		\label{thm:tflu}
		Sphinx satisfies \TFLU
		under the GDH assumption.
	\end{theorem}

	\begin{proof}[Proof Sketch]
		To show that Sphinx
		satisfies \TFLU,
		we prove that an adversary
		cannot distinguish
		the $b = 0$ scenario
		of the \TFLU game
		from the $b = 1$ scenario
		through a hybrid argument
		starting at $b=0$
		and ending at $b=1$.
		We gradually move
		from the $b=0$ scenario
		to the $b=1$ scenario
		by constructing one hybrid
		after another.
		The first hybrid game
		is simply the $b=0$ scenario.
		Each hybrid changes
		the challenge onion
		or the challenger's behavior
		in an indistinguishable way
		until the final hybrid
		is identical to
		the $b=1$ scenario.
		We summarize
		the hybrids' construction here,
		see Appendix~\ref{sapp:tfluproof}
		for details.
		For clarity,
		we also separate the proof
		into two cases:
		One where $j = n$
		and the other
		where $j < n$.
		In the \TFLU game,
		$j$ determines
		the index of the honest relay
		on the challenge onion's path.

		\noindent\textbf{Case 1 $(j = n)$:}
		In this case,
		the honest relay
		on the forward path
		is also
		the exit relay
		of the onion
		and thus also
		the sender of the reply onion.
		This case demands
		that the entire forward onion
		is replaced with an onion
		using the same forward path,
		but containing
		a random message
		and receiver
		as well as
		having an empty reply path.

		To perform this replacement,
		we take advantage
		of the innermost layer
		of the PRP protecting
		the payload's contents:
		Since the adversary
		cannot tell
		what the contents
		of the payload are,
		we can replace
		the message,
		receiver,
		and reply header
		with a random message,
		a random receiver,
		and padding (i.e., an empty reply path)
		respectively.
		
		The original reply header
		is still used
		to form the actual reply onion
		at the honest exit relay.
		We use the PRP's security
		and the implicit integrity check
		in the forward payload
		(via the zero padding)
		to show that
		the adversary
		does not notice this change
		and cannot manipulate
		the payload
		in order to
		to affect the reply header.
		With this,
		the forward onion
		given to the adversary
		is completely independent
		of the challenge onion
		and contains
		a random message,
		a random receiver,
		and an empty reply path.
		The reply header of
		the challenge onion
		is still used
		on the reply path.
		This is identical
		to the $b=1$ scenario.

		\noindent\textbf{Case 2 $(j < n)$:}
		Now,
		we consider the case
		where the honest relay
		on the forward path
		is not the exit relay.
		This case
		is more complex
		since the forward path
		of our replacement onion
		is not the same
		as the forward path
		of the challenge onion.
		Our hybrids thus need to
		truncate the forward path
		of the challenge onion
		at the honest relay
		in order to move
		to the $b=1$ scenario.
		In addition,
		our hybrids must also handle
		a potential tagging attack
		by the adversary correctly.
		Proceed as follows:

		\begin{enumerate}
			\item First,
				we need to handle
				tagging attacks.
				If the adversary tags
				the challenge onion
				in the $b=0$ scenario,
				the payload will be mangled
				by the PRP decryption step
				in the honest relay's processing.
				The key aspect
				of this mangling
				is that the PRP's output
				in that case
				is indistinguishable from
				a random bitstring
				of the same length.
				We take advantage
				of this fact
				in the first hybrids:
				Instead of processing
				the challenge payload normally
				at the honest relay,
				the hybrid checks
				whether the payload has been modified.
				If not,
				the next layer
				of the challenge onion
				with the correct
				next payload layer
				is output.
				If the payload
				has been modified (i.e., tagged),
				then the payload output
				is replaced with
				a random bitstring.
				This change reduces to
				the PRP's security
				at the honest relay
				and the reply receiver.

				With the new payload handling,
				we have effectively decoupled
				the onion layers
				before and after
				the honest relay:
				Any adversary modification
				to the layers
				before the honest relay
				either results
				in a failure in processing
				(due to the MAC
				if the header is modified)
				or in the payload
				being replaced
				with randomness
				(if the payload is modified).
				We can thus change
				the contents
				of both the header and payload layers
				before the honest relay
				and replace the changed onion
				with the original challenge onion
				at the honest relay oracle
				without the adversary
				being able to
				\enquote{sneak} information
				through the honest relay.
				The contents
				of the payload layers
				before the honest relay
				are replaced with
				a random message,
				a random receiver,
				and an empty reply path
				like in the $j = n$ case.

			\item Next,
				we \enquote{detach}
				the layers of the challenge onion's header
				before the honest relay
				(which we will refer to as $A$)
				from the layers
				after it
				(referred to as $B$).

				As a first step,
				the final layer of $A$
				is no longer processed
				at the honest relay.
				The first layer of $B$
				is always output
				as the next layer's header
				instead.
				We detach
				$A$ from $B$
				in multiple steps:
				First,
				replace the innermost contents
				of $A$'s header layers
				with the contents
				of a final Sphinx header layer
				such that
				the $A$ layers
				are now formed
				as if the onion's path
				ended at the honest relay.
				We can do so
				since the PRG
				protects that header information
				until the honest relay.
				Second,
				the KEM keys
				used to build $B$
				are replaced with
				a new instance
				of the KEM
				that starts at
				the honest relay.
				We use the randomness
				of the blinding factor
				multiplied onto the exponent
				at the honest relay
				to secure this step.
				Finally,
				the padding in $B$,
				which still contains
				$A$'s padding,
				is changed:
				The bits corresponding to
				$A$'s padding
				are replaced with random bits
				in the last header layer
				of $B$.
				This is possible
				due to Sphinx's padding construction
				and the PRG's security.
				After these steps,
				the $A$ header layers
				are completely independent
				of the $B$ header layers.

			\item In a last step,
				we adjust $B$'s padding
				and KEM construction
				so that the $B$ layers
				are built like part
				of the complete challenge onion again.
				Now,
				$A$ is an independent forward onion
				with a truncated path,
				an empty reply path,
				and a random message
				and receiver,
				while $B$
				is built like
				the original challenge onion.
				This corresponds to
				the $b=1$ scenario
				of \TFLU.
		\end{enumerate}
		\renewcommand*{\qedsymbol}{}
	\end{proof}

	\noindent\NBLU:
	The \NBLU proof works
	similarly to the \TFLU proof.
	First,
	the part
	of the challenge reply onion
	after the honest relay
	is \enquote{detached}
	from the first part.
	Then,
	the second part's header and payload contents
	are adapted into
	those of a forward onion.

	\noindent\NTI:
	For the \NTI proof,
	we have to truncate
	the forward and reply paths
	of the challenge onion
	to move from one scenario
	to the other.
	Truncating the forward path
	is like performing
	hybrid \Hh[_9]
	from the \TFLU proof
	and adjusting the padding accordingly.
	Truncating the reply path
	is analogous to hybrid \Hh[_7].

	Given that Sphinx
	satisfies
	each of the RSOR properties,
	it follows that

	\begin{theorem}
		\label{thm:sphinx-if}
		Nymserverless Sphinx securely realizes \FRN
		under the GDH assumption.
	\end{theorem}

	\section{Discussion}
	\label{sec:discussion}

	In this section we argue
	that the relaxation
	used in our
	ideal functionality
	is acceptable in practice
	under certain conditions. We further give
	detailed advice on
	when we consider
	the usage of Sphinx
	secure.

	\subsection{\texorpdfstring{Relaxed Security Requirements of \FRN}{Relaxed Security Requirements of FRSOR}}
	\label{sec:relaxedSec}

	We stress that
	\FRN still prevents
	all tagging attacks
	except for the malleability attack
	on the payload.
	Thus,
	if an adversary
	is able to
	link layers
	of an honest sender's onion
	that do not involve the exit relay, 
	both \FRN
	and our properties
	are not achieved.

	However, we
	also emphasize that
	the reduction in security
	due to allowing
	the malleability attack
	can be critical
	and RSOR protocols
	should only be used
	under the following conditions:

	First,
	the exit relay choice
	must not leak critical information. 
	Thus,
	the exit relay
	cannot depend on
	the receiver%
	\footnote{%
		This might be considered
		in order to
		have an exit relay
		that is
		topologically close
		to the receiver.
	} 
	or the included message.
	However, if
	(as in many protocols)
	the exit relays are chosen
	uniformly at random
	or randomly according to
	their capacities,
	linking the sender
	to the exit relay
	does not provide the adversary
	with useful new information.

	Second, senders must not
	react to receiving
	tagged payloads
	in a way
	that an adversary
	can distinguish
	from the reception
	of a well-formed,
	correct payload.
	Otherwise,
	the payload tagging attack
	becomes a threat
	once again:
	A corrupted exit relay
	could tag
	a reply payload
	and observe
	that one sender
	reacts differently
	to receiving
	a reply onion
	and link a sender-receiver pair.
	In particular,
	if a packet format is used
	as part of
	a larger protocol,
	an honest sender
	receiving a reply message
	must not trigger
	any output
	to the adversary.

	Third,
	one must not use
	sessions visible to
	the exit relay.
	This includes techniques
	that allow the linking
	of multiple onions
	for the same sender-receiver pair
	to each other.
	Otherwise,
	one packet from the sender
	can be tagged
	while the other packets
	are not.
	This allows for
	conclusions about
	the receiver
	and content
	of the tagged onion
	via the onions
	which are linked
	to the tagged one
	because of the session.
	
	\vspace{-0.5\baselineskip}
	\subsection{Using Sphinx in a Network}

	We want to give
	practical advice
	on the cases in which
	we consider
	the usage of Sphinx
	in its intended RSOR model
	a secure choice.
	First of all,
	Sphinx should only be used
	with the changes
	we apply
	in this paper.
	Precisely,
	one must include
	the fix for path padding
	(random bits instead of zero bits)
	and the Sphinx reply header
	has to be included
	in the forward payload
	to avoid attacks
	based on the nymserver.

	Further,
	it is important
	to ensure that
	all of the conditions
	mentioned for security
	in \autoref{sec:relaxedSec}
	are met.
	This means that the exit relay
	is chosen randomly,
	independently of the communication parties
	and contents,
	senders' reactions
	to well-formed replies
	and tagged replies
	are indistinguishable,
	and multiple onions
	between the same communication parties
	cannot be linked
	to each other
	at the relays.

	\section{Conclusion}
	\label{sec:conclusion}

%Sphinx
	With this paper,
	we provide
	the privacy formalization
	for repliable service-model OR protocols
	with our ideal functionality \FRN
	and the four new onion properties
	RSOR-Correctness,
	Tagging-Forward Layer Unlinkability,
	RSOR-Backwards Layer Unlinkability,
	and RSOR-Tail Indistinguishability.
	We prove that
	these properties
	imply \FRN together.
	Our formalization
	pays close attention
	to consider all
	the new edge cases
	of the service model
	and to relax
	the security
	in an acceptable way
	to allow for
	payload malleability. 

	For Sphinx,
	we correct
	the cryptographic group assumption
	for the Sphinx scheme
	from DDH to GDH.
	Additionally,
	we realize that
	a security proof
	is not possible
	in the presence of
	the nymserver.
	We propose to
	include the reply header
	in the forward payload instead.
	With our formal groundwork,
	we are then able to
	prove this adapted version
	of Sphinx
	secure according to \FRN.
	To our knowledge,
	we are the first
	to provide a security proof
	for Sphinx
	at our level of detail.
	We thereby ensure
	that the OR and mix networks
	that base their protocols
	on the Sphinx packet format
	can rely on
	a thoroughly-analyzed
	foundation again.

	Of course,
	there is still progress
	in OR,
	mix networks,
	and packet formats
	to be expected
	in future works.
	Authors of new OR
	and mix network protocols
	benefit from our investigation
	of the criteria
	for using Sphinx
	in a secure way
	to decide whether or not
	to base their protocols
	on Sphinx.
	In addition,
	future works
	on OR and mix network packet formats
	profit from
	our formalization
	in the service model,
	especially by using
	our new onion properties
	to build their algorithms
	and prove their privacy.

	\bibliographystyle{ACM-Reference-Format}
	\begin{acks}
		We thank Dennis Hofheinz
		for pointing us
		towards the GDH assumption.
	\end{acks}
	\bibliography{sources-long}
\vfill
\newpage
\appendix

	\section{RSOR onion properties}
	\label{app:rnre-props-full}
	The following definitions
	are copied
	from \citea{or_replies}.
	Our modifications
	are given in
	\tch{this style}.
	We also abbreviate \Recog
	as \RO
	for the sake of
	formatting.

	\subsection{RSOR-Correctness}
	\label{sapp:rnre-correctness}

	\emph{{\ch RSOR}-Correctness} is defined as%
	\footnote{
		This definition
		was originally proposed by
		\citea{formal_onion}
		in a slightly different format.
	}:

	Let ($G$, \FormO, \Proc, \tch{\Reply})
	be a \tch{RSOR} scheme
	with maximal path length N
	\tch{and polynomial $|\Nn|$ and $|D|$.}
	Then for all $n, n\ba\!\! <\!\! N, \lambda \in \mathbb{N}$,
	all choices of
	the public parameter $p$,
	all choices of
	randomness \Rr,
	all choices of
	\tch{receiver $R$},
	all choices of
	forward paths
	{\ch$\Pp = (P_1, \ldots , P_n)$
 	and backward paths
	$\Pp\ba = (P_1\ba, \ldots, P_n\ba)$},
	all $(PK^{(\leftarrow)}_i, SK^{(\leftarrow)}_i)$
	generated by $G(1^\lambda, p, P^{(\leftarrow)}_i)$,
	all messages $m$, $m\ba$,
	all possible choices
	of internal randomness
	used by \tch{\FormO},
	\Proc,
	and \tch{\Reply},
	the following needs to hold:

	\noindent
	\textbf{Correctness of forward path.}
	\vspace{-\abovedisplayskip}
	\begin{gather*}
		Q_i = P_i,\text{ for }1 \leq i \leq n \text{ and }Q_1 := P_1, \\
		O_1 \gets {\ch \FormO(1, \Rr, m, R, \Pp, \Pp\ba, PK_{\Pp}, PK_{\Pp\ba})}, \\
		(O_{i+1}, Q_{i+1}) \gets \Proc(SK_i, O_i, Q_i).
	\end{gather*}
	\textbf{Correctness of request reception.}
	\vspace{-\abovedisplayskip}
	\begin{equation*}
		(m, {\ch R}) = \Proc(SK_{\ch n}, O_{\ch n}, P_{\ch n}).
	\end{equation*}
	\textbf{Correctness of backward path.}
	\vspace{-\abovedisplayskip}
	\begin{align*}
		Q_i\ba &= P_i\ba \text{ for } 1 \leq i \leq n \\
		\text{ and } (O_1\ba, Q_1\ba) &\gets \Reply(m\ba, O_{\ch n}, P_{\ch n}, SK_{\ch n}), \\
		(O_{i+1}\ba, Q_{i+1}\ba) &\gets \Proc(SK_i\ba, O_i\ba, Q_i\ba).
	\end{align*}
	\textbf{Correctness of reply reception.}
	\vspace{-\abovedisplayskip}
	\begin{equation*}
		(m\ba, \bot) = \Proc(SK\ba_{\ch n\ba}, O\ba_{\ch n\ba}, P\ba_{\ch n\ba}).
	\end{equation*}

	\subsection{\texorpdfstring{RSOR-Backw. Layer Unlinkability (\NBLU)}{REOR-Backward Layer Unlinkability}}
	\label{sapp:nblu}

	\textit{{\ch RSOR}-Backward Layer Unlinkability} is defined as:
	\begin{enumerate}
		\item The adversary receives
			the router names $P_H$, \Ps
			and challenge public keys $PK_S$, $PK_H$,
			chosen by the challenger
			by letting $(PK_H, SK_H) \gets G(1^\lambda, p, P_H)$
			and $(PK_S, SK_S) \gets G(1^\lambda, p, \Ps)$.

		\item Oracle access:
			The adversary may submit
			any number of \textbf{Proc}
			and \textbf{Reply} requests
			for $P_H$ or \Ps
			to the challenger.
			For any $\textbf{Proc}(P_H, O)$,
			the challenger checks
			whether $\eta$ is
			on the $\eta_H$-list.
			If it is not on the list,
			it sends the output of
			$\Proc(SK_H, O, P_H)$,
			stores $\eta$
			on the $\eta_H$-list
			and $O$ on the $O_H$-list.
			For any $\textbf{Reply}(P_H, O, m)$,
			the challenger checks
			if $O$ is on the $O_H$-list
			and if so,
			the challenger sends ${\ch \Reply}(m, O, P_H, SK_H)$
			to the adversary.
			(Similar for requests
			on \Ps
			with the $\eta_S$-list).

		\item The adversary submits
			a message $m$,
			{\ch a receiver $R$,}
			a position $j\ba$ with $0 \leq j\ba \leq {\ch n\ba}$,
			a path ${\ch\Pp} = (P_1, \ldots, P_{\ch n})$
			where $P_{\ch n} = P_H$ if $j\ba = 0$,
			a path $\Pp\ba = (P_1\ba, \ldots, P_{j\ba}\ba, \ldots, P\ba_{\ch n\ba} = \Ps)$
			with the honest relay $P_H$
			at backward position $j\ba$
			if $1 \leq j\ba \leq {\ch n\ba}$,
			and the second honest relay \Ps
			at position ${\ch n\ba}$,
			and public keys
			for all relays $PK_i$
			($1 \leq i \leq {\ch n}$
			for the relays on the path
			and ${\ch n} < i$
			for the other relays).

		\item The challenger checks
			that the chosen paths are acyclic,
			the router names
			\tch{and public keys} are valid
			and that the same key is chosen
			if the router names are equal,
			and if so,
			sets $PK\ba_{j\ba} = PK_H$
			(resp. $PK_{\ch n}$ if $j\ba = 0$),
			$PK\ba_{\ch n\ba} = PK_S$
			and sets bit $b$
			at random.

		\item The challenger creates
			the onion $O_1$ with
			the adversary's input choice
			and honestly chosen randomness \Rr:
			\[{\ch \FormO(1, \Rr, m, R, \Pp, \Pp\ba, PK_{\Pp}, PK_{\Pp\ba})}\]
			and sends $O_1$ to the adversary.

		\item The adversary gets oracle access
			as in step 2)
			except if:
			\begin{enumerate}
				\item The request is\ldots
					\begin{itemize}
						\item for $j\ba > 0$:
							$\textbf{Proc}(P_H, O)$
							with
							\begin{align*}
								{\ch \RO(}{\ch n + j\ba\!, O, \Rr, m, R,}{\ch \Pp, \Pp\ba\!, PK_{\Pp}, PK_{\Pp\ba}\!)} = True,
							\end{align*}
							$\eta$ is not
							on the $\eta_H$-list
							and \[\Proc(SK_H, O, P_H) \neq {\ch (\bot, \bot)}\mathpunct{:}\]
							stores $\eta$
							on the $\eta_H$
							and $O$
							on the $O_H$-list and \ldots

						\item for $j\ba = 0$:
							$\textbf{Reply}(P_H, O, m\ba)$
							with
							\begin{align*}
							{\ch \RO(}{\ch n, O, \Rr, m, R,}{\ch\Pp, \Pp\ba, PK_{\Pp}, PK_{\Pp\ba})} = True,
							\end{align*}
							$O$ is
							on the $O_H$-list
							and no onion
							with this $\eta$
							has been replied to before
							and \[{\ch \Reply}(m\ba, O, P_H, SK_H) \neq {\ch (\bot, \bot)}\ldots\]
					\end{itemize}
					\ldots then:
					The challenger picks
					the rest of the return path
					${\ch \bar{\Pp}} = (P_{j\ba + 1}\ba, \ldots, P\ba_{\ch n\ba})$,
					an empty backward path $\bar{\Pp}\ba = ()$,
					and a random message $\bar{m}$,
					another honestly chosen randomness $\bar{\Rr}$,
					{\ch an honestly chosen random receiver $\bar{R}$},
					and generates
					an onion $\bar{O}_1$:
					\[{\ch \FormO(1, \bar{\Rr}, \bar{m}, \bar{R}, \bar{\Pp}, \bar{\Pp}\ba, PK_{\bar{\Pp}}, PK_{\bar{\Pp}\ba})}\]

					\begin{itemize}
						\item If $b = 0$,
							the challenger calculates
							\begin{align*}
								(O_{j\ba\!+1}, P_{j\ba\!+1}\ba)\! =
								\begin{cases}
									\Proc(SK_H, O, P_{j\ba}\ba\!)&, j\ba > 0, \\
									{\ch \Reply}(m\ba\!\!, O, P_{j\ba}\ba\!, SK_H) &, j\ba = 0
								\end{cases}
							\end{align*}
							and gives $O_{j\ba+1}$
							for $P_{j\ba+1}\ba$
							to the adversary.

						\item Otherwise,
							the challenger gives
							$\bar{O}_1$
							for $P_{j\ba+1}\ba$
							to the adversary.
					\end{itemize}

				\item $\textbf{Proc}(\Ps, O)$
					with $O$ being
					the challenge onion
					as processed
					for the final receiver
					on the backward path,
					i.e.:
					\begin{itemize}
						\item for $b = 0$:
							\begin{align*}
								{\ch \RO(}{\ch n + n\ba\!\!, O, \Rr, m, R,}{\ch\Pp, \Pp\ba\!\!, PK_{\Pp}, PK_{\Pp\ba})} = True
							\end{align*}

						\item for $b = 1$:
							\begin{align*}
								{\ch \RO(}{\ch n\ba\!\! - j\ba\!\!, O, \bar{\Rr}, \bar{m}, \bar{R},}
								{\ch\bar{\Pp}, \bar{\Pp}\ba\!\!, PK_{\bar{\Pp}}, PK_{\bar{\Pp}\ba})} = True
							\end{align*}
					\end{itemize}
					\ldots then the challenger
					outputs nothing.
			\end{enumerate}

		\item The adversary
			produces guess $b'$.
	\end{enumerate}
	{\ch \NBLU} is achieved
	if any PPT adversary \Aa
	cannot guess $b' = b$
	with a probability
	non-negligibly better
	than $\frac{1}{2}$.

	\subsection{\texorpdfstring{RSOR-Tail Indistinguishability (\NTI)}{REOR-Tail Indistinguishability}}
	\label{sapp:nti}

	\textit{{\ch RSOR}-Tail Indistinguishability} is defined as:

	\begin{enumerate}
		\item The adversary receives
			the router names
			$P_H$, $P_H\ba$, \Ps,
			and challenge public keys
			$PK_S$, $PK_H$, $PK_H\ba$,
			which are chosen by the challenger
			by letting
			$(PK_H, SK_H) \gets G(1^\lambda, p, P_H)$,
			$(PK_H\ba, SK_H\ba) \gets G(1^\lambda, p, P_H\ba)$,
			$(PK_S, SK_S) \gets G(1^\lambda, p, \Ps)$.

		\item Oracle access:
			The adversary may submit
			any number of
			\textbf{Proc} and \textbf{Reply} requests
			for $P_H$, $P_H\ba$, or \Ps
			to the challenger.
			For any $\textbf{Proc}(P_H, O)$,
			the challenger checks
			whether $\eta$ is
			on the $\eta_H$-list.
			If it is not on the list,
			it sends the output of
			$\Proc(SK_H, O, P_H)$,
			stores $\eta$
			on the $\eta_H$-list
			and $O$
			on the $O_H$-list.
			For any $\textbf{Reply}(P_H, O, m)$,
			the challenger checks
			if $O$ is
			on the $O_H$-list
			and if so,
			the challenger sends
			${\ch \Reply}(m, O, P_H, SK_H)$
			to the adversary.
			(Similar for requests
			on $P_H\ba$, \Ps).

		\item The adversary submits
			a message $m$,
			{\ch a receiver $R$,}
			a path ${\ch \Pp} = (P_1, \ldots, P_j, \ldots, P_{\ch n})$
			with the honest relay
			$P_H$ or $P_H\ba$
			at position $j, 0 \leq j < {\ch n}$,
			a path $\Pp\ba = (P_1\ba, \ldots, P_{\ch n}\ba)$
			with the honest relay $P_H\ba$
			at position $1 \leq j\ba \leq {\ch n\ba}$
			and public keys
			for all relays $PK_i$
			($1 \leq i \leq {\ch n\ba}$ for the relays
			on the path
			and ${\ch n} < i$
			for the other relays).

		\item The challenger checks that
			the given paths
			are acyclic,
			the router names
			\tch{and public keys} are valid
			and that the same key is chosen
			if the router names are equal,
			and if so,
			sets $PK_j = PK_H$
			(or $PK_j = PK_H\ba$,
			if the adversary chose $P_H\ba$
			at this position as well),
			$PK\ba_{j\ba} = PK_H\ba$,
			$PK\ba_{\ch n\ba} = PK_S$
			and sets bit $b$
			at random.

		\item The challenger creates
			the onion $O_{j+1}$
			with the adversary's
			input choice
			and honestly chosen randomness \Rr:
			\[{\ch \FormO(j + 1, \Rr, m, R, \Pp, \Pp\ba, PK_{\Pp}, PK_{\Pp\ba})}\]
			and a replacement onion
			$\bar{O}_1$
			with the path
			from the honest relay $P_H$
			to the corrupted {\ch exit relay}
			${\ch \bar{\Pp}} = (P_{j+1}, \ldots, P_{\ch n})$
			and the backward path
			ending at $j\ba$:
			$\bar{\Pp}\ba = (P_1\ba, \ldots, P\ba_{j\ba})$;
			and another honestly chosen randomness $\bar{\Rr}$:
			\[{\ch \FormO(1, \bar{\Rr}, m, R, \bar{\Pp}, \bar{\Pp}\ba, PK_{\bar{\Pp}}, PK_{\bar{\Pp}\ba})}\]

		\item If $b = 0$:
			The challenger sends
			$O_{j+1}$
			to the adversary. \\
			Otherwise:
			The challenger sends
			$\bar{O}_1$
			to the adversary.

		\item Oracle access:
			the challenger processes
			all requests
			as in step 2)
			except for\ldots \\
			\ldots $\textbf{Proc}(P_H\ba, O)$
			with $O$ being
			the challenge onion
			as processed for
			the honest relay
			on the backward path,
			i.e.:
			\begin{itemize}
				\item for $b = 0$:
					\begin{align*}
					{\ch \RO(}{\ch n + j\ba, O, \Rr, m, R,}{\ch\Pp, \Pp\ba, PK_{\Pp}, PK_{\Pp\ba})} = True
					\end{align*}

				\item for $b = 1$:
					\begin{align*}
						{\ch \RO(}{\ch(n - j) + j\ba, O, \bar{\Rr}, m, R,}
								  {\ch\bar{\Pp}, \bar{\Pp}\ba, PK_{\bar{\Pp}}, PK_{\bar{\Pp}\ba})} = True
					\end{align*}
			\end{itemize}
			\ldots then the challenger
			outputs nothing.

		\item The adversary produces
			guess $b'$.
	\end{enumerate}
	{\ch \NTI} is achieved
	if any PPT adversary \Aa
	cannot guess $b' = b$
	with a probability
	non-negligibly better
	than $\frac{1}{2}$.

	\newpage
	\section{Ideal Functionality}
	\label{app:reor-if}
	Modifications to
	\citeauthor*{or_replies}'s \Ff[R]
	are highlighted
	in \textcolor{blue5}{this style}
	in the pseudocode.

	\newcommand*{\algspace}{\vspace{0.15cm}}
	\newcommand*{\ShSend}{\textsc{Out.Cor.Sender}}
	\newcommand*{\ShStep}{\textsc{Proc.NextStep}}
	\begin{algorithm}[ht]
		\caption{Ideal Functionality \FRN (1)}
		\label{alg:frn-1}
		\begin{algorithmic}
			\footnotesize
			\DocComment{\textbf{Data structures:}}
			\State {Bad: Set of corrupted relays and receivers}
			\State {$L_o$: List of onions processed by adversarial relays}
			\State {$B_i$: List of onions held by relay $P_i$}
			\State {\ch $B^r_i$: List of receiver replies held by relay $P_i$}
			\State {\ch $L_{tag}$: List of messages tagged by the adversary}
			\State {$Back$: Map from $tid$s to reply paths and forward IDs}
			\State {$ID_{fwd}$: Map from a reply onion ID to a forward onion ID}
			\State {\ch $Rep_i$: Map of reply identifiers to $tid$s at exit relay $P_i$}

%			\DocComment{\textbf{Notation:}}
%			\State {\Ss: Adversary or simulator}
%			\State {\Zz: Environment}
%			\State {$\Pp = \Path{1}{n}$: Onion path}
%			\State {$O = (sid, \Ps, {\color{blue5}R / P_r}, m, \Pp, \Pp\ba, i, d):$ Onion = (session ID, sender, receiver or receiving relay, message, forward path, reply path, distance traveled, direction)}
%			\State {$N$: Maximum onion path length}%
%			\Statex
%
			\vspace{0.2cm}
			%\DocComment{$\Ps\; {\color{blue5}\in \Nn}$ creates and sends a new onion}
			\Message{\NewOnion}{${\color{blue5}R}, m, \Pp, \Pp\ba$}{\textit{\Zz or \Ss via \Ps}}
				\If{$|\Pp| > N \textbf{ or } |\Pp\ba| > N$} \textbf{reject}
				\Else
					\State $sid \gets^R$ session ID
					\State $O \gets (sid, \Ps, {\color{blue5}R}, m, \Pp, 0, f)$
					\State \Call{\ShSend}{$\Ps, sid, {\color{blue5}R}, m, \Pp, \Pp\ba, \textsl{start}, f$}
					\State \Call{\ShStep}{$O$}
				\EndIf
			\EndMessage

			\algspace
			%\DocComment{\textcolor{blue5}{$P_i \in \Nn$ creates and sends the reply onion for $tid$}}
			\BlueProc{\NewReply}{$m, tid$}
				\If{$(tid, \ldots) \notin Back$} \textbf{reject}
				\Else
					\State {\color{blue5}$(tid, \Ps, \Pp, \Pp\ba, sid', \_)$} $\gets Back$
					\State $sid \gets^R$ session ID
					\State Store $(sid, sid')$ in $ID_{fwd}$
					\State $O \gets (sid, P_{\ch i}, \Ps, m, \Pp\ba, (), 0, b)$
					\State \Call{\ShSend}{$P_{\ch i}, sid, \Ps, m, \Pp, \Pp\ba, \textsl{start}, b$}
					\State \Call{\ShStep}{$O$}
				\EndIf
			\EndBlueProc

			\algspace
			%\DocComment{\Ss delivers the onion $tid$ to the next relay}
			\Message{\DelOnion}{$tid$}{\Ss}
				\If{$(tid, \_, \_) \in L_o$}
					\State $(tid, O = (sid, \Ps, {\color{blue5}R / P_r}, m, \Pp, \Pp\ba, i, d), j) \gets L_o$
					\State $O \gets (sid, \Ps, {\color{blue5}R / P_r}, m, \Pp, \Pp\ba, j, d)$
					\ch
					\If{$d = b \textbf{ and } j = |\Pp|$}
						\normalcolor
						\If{$m \neq \bot {\ch \textbf{ and } O \notin L_{tag}}$}
							\ch
							\State \Call{Send}{$P_r$, \enquote{Message $m$ received as reply}}
							\DocComment{Not forwarded to \Zz}
						\EndIf
						\normalcolor
					\Else
						\State $tid' \gets^R$ temporary ID
						\State \Call{Send}{\Po{j}, \enquote{$tid'$ received from \Po{j-1}{}}}
						\State Store $(tid', O)$ in $B_{o_j}$
					\EndIf
				\EndIf
			\EndMessage

			\algspace
			%\DocComment{$P_i$ (honest or corrupted) is done processing the onion $tid'$}
			\Message{\ForOnion}{$tid'$}{\textit{\Zz or \Ss via $P_i$}}
				\If{$(tid', \_) \in B_i$}
					\State Pop $(tid', O)$ from $B_i$
					\State \Call{\ShStep}{$O$}
				\ch
				\ElsIf{$(tid', \_) \in B^r_i$}
					\State Pop $(tid', m, tid)$ from $B^r_i$
					\State \Call{\NewReply}{$m, tid$}
				\EndIf
				\normalcolor
			\EndMessage

			\algspace
			\color{blue5}
			%\DocComment{\Ss tags the onion $tid$}
			\Message{Tag}{$tid$}{\Ss}
				\If{$(tid, \_, \_) \in L_o$}
					\State Retrieve $(tid, O, \_)$ from $L_o$
					\State Store $O$ in $L_{tag}$
				\EndIf
			\EndMessage
			\normalcolor

			\algspace
			%\DocComment{Give \Ss all information on an onion if \Ps is corrupted}
			\Procedure{\ShSend}{$\Ps, sid, {\color{blue5}R / P_r}, m, \Pp, \Pp\ba\!\!, tid, d$}
				\If{$d = f \textbf{ and } \Ps \in \text{Bad}$}
					\State \Call{Send}{\Ss,
						\enquote{$tid$ is from \Ps with
							$sid$, ${\color{blue5}R}$, $m$, $\Pp$, $\Pp\ba$, $d$}}
				\ElsIf{$d = b \textbf{ and } P_r \in \text{Bad}$}
					% We need to get the length before the first \State, or we'll be indented too far
					\makeatletter
					\setlength{\curindent}{\ALG@tlm}
					\makeatother

					\State \CallPar{Send}{\Ss,
						\enquote{$tid$ is reply from \Ps with
							$sid$, ${\color{blue5}P_r}$, $m$, $\Pp$, $\Pp\ba\!\!$, \\
					\phantom{\hspace{\curindent} \Call{Send}{\Ss}}
					$d$, replying to onion from $P_r$ with $ID_{fwd}(sid)$}}
				\EndIf
			\EndProcedure
		\end{algorithmic}
	\end{algorithm}

	\begin{algorithm}[ht]
		\caption{Ideal Functionality \FRN (2)}
		\label{alg:frn-2}
		\begin{algorithmic}
			\footnotesize
			%\DocComment{Process onion with honest successor relay \Po{j}}
			\Procedure{\textsc{Proc.ToRelay}}{$O = (sid, \Ps, {\color{blue5}R / P_r}, m, \Pp, \Pp\ba\!\!, i, d)$}
				\State $\Po{j} \gets \Po{k}$ with smallest $k > i$ such that $\Po{k} \notin$ Bad
				\State $tid \gets^R$ temporary ID
				\State \Call{Send}{\Ss, \enquote{\Po{i} sends $tid$ to \Po{j} via \Path{i+1}{j-1}{}}}
				\State \Call{Send}{\Po{i}, \enquote{Sent onion to \Po{i+1}{}}}
				\State \Call{\ShSend}{$\Ps, sid, {\color{blue5}R}, m, n, \Pp, tid, d$}
				\If{$d = b \textbf{ and } i = 0$}
					\State \Call{Send}{\Ss, \enquote{$tid$ belongs to $sid$}}
				\EndIf
				\State Add $(tid, O, j)$ to $L_o$
			\EndProcedure

			\algspace
			\ch
			%\DocComment{Allow replies to an onion to be created}
			\Procedure{\MakeReply}{$O = (sid, \Ps, R, m, \Pp, \Pp\ba, i, d), rid$}
				\State $tid \gets^R$ temporary ID
				\State Store $(tid, \Ps, \Pp, \Pp\ba, \Po{i}, sid)$ in $Back$
				\If{$i = |\Pp|$}
					\State Store $(rid, tid)$ in $Rep_{o_i}$
					\State \Call{Send}{\Ss, \enquote{Reply with reply ID $rid$}}
				\Else
					\State $\Pp[1]\ba\!\! \gets$ prefix of $\Pp\ba\!\!$ up to (including) the first honest relay
					\State \Call{Send}{\Ss, \enquote{Reply with $tid$, reply path begins with $\Pp[1]\ba$}}
				\EndIf
			\EndProcedure

			\algspace
			%\DocComment{\Ss learns the message on the final link}
			\Procedure{\LeakMess}{$O = (sid, \Ps, R, m, \Pp, \Pp\ba, i, d)$}
				\makeatletter
				\setlength{\curindent}{\ALG@tlm}
				\makeatother
				\If{$m = \bot$}
					\textbf{return}
				\EndIf
				\State \Call{\ShSend}{$P_s, sid, R, m, \Pp, \Pp\ba, \textsl{end}, d$}
				\If{$\Pp \neq ()$}
					\State $rid \gets^R$ temporary ID
					\State \Call{\MakeReply}{$O, rid$}
				\EndIf
				\If{$i = |\Pp|$}
					\Call{Send}{\Po{i}, \enquote{Sent message to $R$}}
				\Else
					\ \Call{Send}{\Po{i}, \enquote{Sent onion to \Po{i+1}{}}}
				\EndIf
				\vspace{-0.6ex}
				\State \CallPar{Send}{\Ss,
					\enquote{$\Po{i}$ sends onion with message $m$ to $R$ \\
						\phantom{\hspace{\curindent} \Call{Send}{\Ss}}
						via \Path{i+1}{n}{}}}
			\EndProcedure

			\algspace
			%\DocComment{\Ss (masquerading as $P_i \in \Nn$) delivers an arbitrary message to any receiver}
			\Message{\DelPlain}{$P_i, m, rid, R$}{\Ss}
				\State \Call{Send}{$R$, \enquote{Message $m$ received from $P_i$}}
				\If{$rid \neq \bot$}
					\State \Call{Send}{$R$, \enquote{Message is repliable with $rid$}}
				\EndIf
			\EndMessage

			\algspace
			%\DocComment{$R \in D$ decides to send a reply message to $rid$ via $P_i$}
			\Message{\InitReply}{$P_i, m, rid$}{\textit{\Zz or \Ss via $R$}}
				\State \Call{Send}{\Ss, \enquote{$R$ replies to $rid$ with message $m$ via $P_i$}}
			\EndMessage

			\algspace
			\DocComment{$P_i$ creates an onion from $R$'s reply request}
			\Message{\DelReply}{$R, P_i, m, rid$}{\Ss}
				\State \Call{Send}{$P_i$, \enquote{Reply $(m, rid)$ received from $R$}}
				\If{$(rid, \_) \in Rep_i$}
					\State $(rid, tid) \gets Rep_i$
					\State $tid' \gets^R$ temporary ID
					\State Store $(tid', m, tid)$ in $B^r_i$
					\State \Call{Send}{$P_i$, \enquote{Send reply onion with $tid'$}}
				\EndIf
			\EndMessage

			\algspace
			\DocComment{\Ss uses a $tid$ ID to bypass replying via an exit relay}
			\Message{\ByReply}{$m, tid$}{\textit{\Ss via $P_i$}}
				\If{$(tid, \ldots) \in Back$}
					\State \Call{\NewReply}{$m, tid$}
				\EndIf
			\EndMessage

			\algspace
			%\DocComment{\Ss learns the reply message as it is delivered to the corrupted sender}
			\Procedure{\LeakReply}{$O = (sid, \Ps, P_r, m, \Pp, \Pp\ba, i, d)$}
				% We need to get the length before the first \State, or we'll be indented too far
				\makeatletter
				\setlength{\curindent}{\ALG@tlm}
				\makeatother
				\State \CallPar{Send}{\Ss, \enquote{\Po{i} sends reply $tid$ with message $m$ to $P_r$ \\
				\phantom{\hspace{\curindent} \Call{Send}{\Ss}}
				via \Path{i+1}{n-1}{}}}
				\State \Call{Send}{\Po{i}, \enquote{Sent onion to \Po{i+1}{}}}
				\State \Call{\ShSend}{$\Ps, sid, P_r, m, \Pp, \Pp\ba, tid, b$}
			\EndProcedure

			\algspace
			%\DocComment{$\Po{i}$ has processed $O$, passes it to $\Po{i+1}$ or $R$}
			\Procedure{\textsc{Proc.NextStep}}{$O\! =\! (sid, \Ps, R / P_r, m, \Pp, \Pp\ba\!\!, i, d)$}
				\If {$\forall j > i: \Po{j} \in$ Bad \textbf{or} $i = |\Pp|$}
					%\DocComment{Either all remaining relays including the exit relay}
					%\DocComment{(or reply receiver) are corrupted, or \Po{i} is the exit relay}
					\If {$O \in L_{tag}$}
						\State \Call{\ShSend}{$\Ps, sid, R / P_r, m, \Pp, \textsl{tagged}, d$}
						\If {$i < n$}
							\State \Call{Send}{\Ss, \enquote{\Po{i} sends tagged via \Path{i+1}{n}{}}}
							\State \Call{Send}{\Zz, \enquote{\Po{i} sends onion to \Po{i+1}{}}}
						\Else \ \Call{Send}{\Zz, \enquote{Onion at \Po{i} fails integrity check}}
						\EndIf
					\Else
						\If {$d = f$} \Call{\LeakMess}{$O$}
						\Else \ \Call{\LeakReply}{$O$}
						\EndIf
					\EndIf
				\Else \ \Call{\ToRelay}{$O$}
				\EndIf
			\EndProcedure
			\normalcolor
		\end{algorithmic}
	\end{algorithm}

	\clearpage

	\section{Sphinx: Packet Format Description}
	\label{sapp:sphinx-packet-format}

	\compactequations
	This section
	is meant as a compact reference
	for the structure
	of the Sphinx packet,
	which is used
	in the following appendices.
	For a complete introduction
	to the Sphinx packet format,
	see%
	~\cite{sphinx}.
	A Sphinx packet
	is made of
	a header $\eta = (\alpha, \beta, \gamma)$
	and a payload $\delta$.
	$\alpha$ is built
	using the public keys
	$y_i:=g^{x_i}$
	for each relay $P_i$
	on the onion's path.
	The sender
	chooses a secret
	$x \in  \mathbb{Z}^*_q$
	and lets
	$\alpha_i = g^{xb_0\cdots b_{i-i}}$
	and $s_i = y_i^{xb_0\cdots b_{i-1}}$,
	where $b_i = h_b(\alpha_{i-1}, s_{i-1})$.
	$\alpha_i$ is the group element
	contained in the $i$-th layer
	of the header,
	and $s_i$ is the corresponding secret
	it shares with $P_i$%
	\footnote{%
		Note that
		our indices $i$
		for relay names
		are relative to
		a single packet's path
		for clarity.
	}.
	The $b_i$ are blinding factors
	that transform $\alpha$
	at each relay.
	They are calculated with
	a random oracle $h_b$.
	The remainder of the header
	is built using
	the shared secrets $s_i$
	after passing them
	through more random oracles:
	$h_\rho$, $h_\mu$,
	and $h_\pi$ are each used
	to key a different primitive.
	The $\beta_i$ are built
	starting at
	the final layer:
	\begin{flalign*}
	\beta_{\nu-1} := \{\ast/\Ps\|I\|rand_{(2(r-n)+2)\kappa - |R|}\} \oplus \\
	\{\rho(h_{\rho}(s_{n-1}))_{[\ldots(2(r-n)+3)\kappa-1]}\} \| \Phi_{n-1}%
	\footnotemark.
	\end{flalign*}
	\footnotetext{%
		Note that we use
		the randomized padding
		described in
		\autoref{sssec:sphinx-insecure} here.
	}
	In this definition,
	$\ast/\Ps$ is either
	a sentinel value indicating that
	this is a forward packet
	or the address
	of the original sender
	in a reply packet.
	$I$ is an identifier
	used by $\Ps$
	to recognize replies.
	$r$ is the global maximum path length
	in this Sphinx instance.
	$\rho$ is a PRG
	keyed with $h_\rho(s_i)$.
	$\Phi_i$ is padding formed
	via the repeated application
	of the $\rho$ PRG:
	$\Phi_0$ is empty,
	while
	\[\Phi_i = \{\Phi_{i-1}\|0_{2\kappa}\} \oplus \rho(h_\rho(s_i))_{[(2(r-i)+3)\kappa\ldots(2r+3)\kappa-1]}.\]
	The remaining $\beta_i$
	are built as
	\[\beta_i\! =\! \{P_{i+1}\|\gamma_{i+1}\|\beta_{i+1_{[\ldots(2r-1)\kappa-1]}}\} \oplus \rho(h_{\rho}(s_i))_{[\ldots(2r+1)\kappa-1]}.\]
	Effectively,
	the construction
	of the padding is designed
	such that
	$\Phi_i$ is a suffix
	of $\beta_i$.
	$\gamma_i$ is the MAC
	$\mu(h_\mu(s_i), \beta_i)$
	of $\beta_i$.
	Finally,
	a forward payload $\delta$ is formed
	as $\delta_{n-1} = \pi(h_\pi(s_{n-1}), 0_\kappa\|R\|m)$
	and $\delta_i = \pi(h_\pi(s_i), \delta_{i+1})$,
	where $\pi$ is a PRP
	keyed with $h_\pi(s_i)$.

	The packet sent by
	the sender
	is $((\alpha_0, \beta_0, \gamma_0), \delta_0)$.
	Each relay $P_i$
	then processes the packet it gets
	into $((\alpha_{i+1}, \beta_{i+1}, \gamma_{i+1}), \delta_{i+1})$%
	~\cite{sphinx}.
	\looseequations

	\section{Sphinx: Adapted KEM-IND-CCA}
	\label{app:Sphinx-KEM-Proof}
	\begin{proof}
		We will use
		a PPT attacker \Aa
		on the \SKEMCCA property
		for the \ROKEM
		to construct an attacker $\Bb^{\Oo[G]}$
		on the GDH assumption
		using the DDH oracle $\Oo[G]$.

		The GDH attacker \Bb
		maintains five lists:
		\begin{itemize}
			\item $L$: List of group elements $g^z$
				with their associated oracle outputs
				($h_*(g^z)$,
				$h_*(g^{x_1*z})$,
				$h_b(g^z, g^{x_1*z})$).
			\item $L_y$: List of up to $N$ tuples,
				one for each adversarial relay
				on the KEM path:
				Each holds $(\alpha_i, h_*(s_i), h_b(\alpha_i, s_i))$.
			\item $L_{b}$: List of $(g^a, g^z)$ element pairs
				with their corresponding $h_b(g^a, g^z)$ values.
			\item $L_{\Oo}$: List of $\alpha'$s that have been requested
				from \Oo.
			\item $L_h$: List of group elements that have been requested
				from $h_*$.
		\end{itemize}

		\Bb receives a CDH challenge $(g, g^{x_1}, g^{x_2})$
		from the GDH challenger \Cc.
		\Bb uses $g^{x_1}$ as the public key $PK$
		of the \enquote{honest relay}
		and $g^{x_2}$ as the challenge $\alpha_j$.
		The attacker sends $PK$ to \Aa
		and gives \Aa access to
		the programmed random oracles $h_*$ and $h_b$
		and the decapsulation oracle \Oo
		(which are described below).
		Next,
		\Bb receives $j$ and $y_i$ for $i \neq j$
		from \Aa
		and checks that
		the public keys are valid.
		\Bb now chooses $b_0, \ldots, b_n$
		randomly from $\mathbb{Z}^*_p$.
		To calculate $\alpha_i$ for $i < j$,
		\Bb calculates the inverses
		$b_0^{-1}, \ldots, b_{j-1}^{-1}$
		and uses them to form
		$\alpha_i = \alpha_j^{b_i^{-1}\dots b_{j-1}^{-1}}$.
		To make $\alpha_i$ for $i > j$,
		let $\alpha_i = \alpha_j^{b_j\dots b_{i-1}}$.
		Next,
		\Bb chooses $n$ random $r_0, \ldots, r_n$ values
		as outputs for $h_*(s_i)$.
		To remember these choices
		in the programmed random oracles,
		\Bb stores
		$(\alpha_i, r_i, b_i)$ in $L_y$
		for $i \neq j$.
		Now,
		it flips a bit $b$
		to determine whether it will simulate
		the KEM game for $b = 0$ or $b = 1$.
		If $b = 0$,
		\Bb sets $L[\alpha][1] := r_j$,
		and $L[\alpha][2] := b_j$.
		Finally,
		$\alpha_0, r_0, \ldots, r_{j-1}, r_{j+1}, \ldots, r_n$
		and $b_0, \ldots, b_{j-1}, b_{j+1}, \ldots, b_n$
		are sent to \Aa
		along with $(\alpha_j, r_j, b_j)$.

		\begin{algorithm}
		\caption{KEM attacker \Bb's oracles}
		\label{alg:kem-oracles}
		\small
		\begin{algorithmic}
			\Function{\Oo}{$\alpha'$}
				\If{$\alpha$ generated \textbf{and} $\alpha' = \alpha$} \textbf{abort}
				\EndIf
				\If{$L[\alpha'][1]$ is not set}
					\State Add $\alpha'$ to $L_{\Oo}$
					\For{$g^z \in L_h$}
						\If{$\Oo[G](g, PK, \alpha', g^z)$}
							%\DocComment{$g^z$ is the secret $s'$ to $\alpha'$}
							\State $L[\alpha'][1] \gets L[g^z][0]$
						\EndIf
					\EndFor
					\If{$L[\alpha'][1]$ is not set}
						\State $L[\alpha'][1] \gets^R \{0,1\}^{3\kappa}$
					\EndIf
					\For{$(g^a, g^z) \in L_b$ where $g^a = \alpha'$}
						\If{$\Oo[G](g, PK, \alpha', g^z)$}
						\State $L[\alpha'][2] \gets L_b[(\alpha', g^z)]$
						\EndIf
					\EndFor
					\If{$L[\alpha'][2]$ is not set}
						\State $L[\alpha'][2] \gets^R \mathbb{Z}^*_q$
					\EndIf
				\EndIf
				\State \textbf{return} $(L[\alpha'][1], L[\alpha'][2])$
			\EndFunction

			\Function{$h_*$}{$g^z$}
				\If{$L[g^z][0]$ is not set}
					\State Add $g^z$ to $L_h$
					\If{$\alpha$ generated \textbf{and} $\Oo[G](g, PK, \alpha, g^z)$}
						%\DocComment{\Aa has solved the CDH problem}
						\State $\textsl{bad} \gets 1$
					\EndIf
					\For{$(\alpha_i, r_i, \_) \in L_y$}
						\If{$\Oo[G](g, y_i, \alpha_i, g^z)$}
							%\DocComment{$g^z$ is the secret $s_i$ to $\alpha_i$ (for the key $y_i$)}
							\State $L[g^z][0] \gets r_i$
						\EndIf
					\EndFor
					\For{$\alpha' \in L_{\Oo}$}
						\If{$\Oo[G](g, PK, \alpha', g^z)$}
							%\DocComment{$g^z$ is the secret $s'$ to $\alpha'$}
							\State $L[g^z][0] \gets L[\alpha'][1]$ 
						\EndIf
					\EndFor
					\If{$L[g^z][0]$ is not set}
						\State $L[g^z][0] \gets^R \{0,1\}^{3\kappa}$
					\EndIf
				\EndIf
				\State \textbf{return} $L[g^z][0]$
			\EndFunction

			\Function{$h_b$}{$g^a, g^z$}
				\If{$L_b[(g^a, g^z)]$ is not set}
					\If{$\Oo[G](g, PK, g^a, g^z)$}
						\If{$\alpha$ generated \textbf{and} $g^a = \alpha$}
							%\DocComment{\Aa has solved the CDH problem}
							\State $\textsl{bad} \gets 1$
						\EndIf
						\If{$L[g^a][2]$ is set}
							\State $L_b[(g^a, g^z)] \gets L[g^a][2]$
						\EndIf
					\EndIf
					\If{$(g^a, \_, \_) \in L_y$ at index $i$ \\
					\phantom{\textbf{functi if}}\textbf{and} $\Oo[G](g, y_i, g^a, g^z)$}
						\State Retrieve $(g^a, \_, b_i)$ from $L_y$
						\State $L_b[(g^a, g^z)] \gets b_i$
					\EndIf
					\If{$L_b[(g^a, g^z)]$ is not set}
						\State $L_b[(g^a, g^z)] \gets^R \mathbb{Z}^*_q$
					\EndIf
				\EndIf
				\State \textbf{return} $L[(g^a, g^z)]$
			\EndFunction
		\end{algorithmic}
		\normalsize
		\end{algorithm}

		To simulate the decapsulation oracle \Oo
		and the random oracles $h_*$ and $h_b$,
		\Bb behaves as described
		in \autoref{alg:kem-oracles}.

		Using a bad-flag analysis,
		we can see that any attacker \Aa
		with a non-negligible advantage in the KEM game
		must trigger the $bad$ flag
		non-negligibly often,
		so \Bb can also win the GDH game
		non-negligibly often.
		The random oracles $h_*$ and $h_b$
		are set up to behave correctly
		in combination with \Oo
		except if a \enquote{collision}
		is created in the challenge creation phase
		(steps 3--7).
		Here,
		a collision refers to
		\Bb assigning two different random values
		to the same random oracle inputs
		on accident.
		This can occur
		in two ways:
		\begin{enumerate}
			\item The attacker
				already requested
				$\Oo(\alpha)$ in step two.
			\item The attacker
				already requested
				$h_*(s_i)$ or $h_b(\alpha_i, s_i)$
				for $0 \leq i \leq n$
				in step two or
				\Bb generates its own collision
				on accident when
				$s_{i_1} = s_{i_2}$
				for $i_1 \neq i_2$.
		\end{enumerate}

		In any of the above cases,
		\Bb simulates the oracles incorrectly%
		\footnote{%
			Unless \Bb happened
			to choose the same randomness
			in both cases,
			which only happens
			negligibly often.
		}.
		Let the number of requests
		\Aa makes
		to each oracle
		\Oo, $h_*$, and $h_b$
		be bounded by the polynomial $p(\kappa)$.
		The probability of each case
		is bounded by
		$1/q$
		and $\lessapprox \frac{(p(\kappa) + N)^2}{2q}$
		each for a collision
		on $h_*$ and $h_b$%
		\footnote{%
			This corresponds to
			an upper bound
			for the likelihood for a
			successful birthday attack
			given $p(\kappa) + N$ requests
			and a pre-image set with $q$ members%
			~\cite{intro_crypto_4}.
		}
		respectively.
		According to Sphinx's definition,
		the order $q$ of the group \Gg
		is on the order of $2^{2\kappa}$,
		meaning that
		both of these probabilities
		are negligible.
		
		If neither of these events occur,
		\Bb simulates the oracles correctly
		and wins the GDH game
		whenever \Aa wins the \SKEMCCA game.
	\end{proof}

	\newpage
	\section{Secure RSOR Definitions}
	\label{app:reor-defs}

	The following definition
	is adapted from \citeauthor*{or_replies}'s definition
	for repliable integrated-system-model OR
	(as shown in \autoref{ssec:rior}).
	Our changes compared
	to their definitions
	are shown in \tch{this style}.

	\begin{definition}
		\label{def:sec-rnre-protocol-full}
		An {\ch RSOR} protocol
		is secure
		in the \Ff[PKI]-\Ff[SC]-hybrid model
		if and only if
		it is built on
		a secure {\ch RSOR} scheme
		according to \autoref{def:sec-rnre-scheme}
		and operates
		the following way:

		\compactequations
		\begin{itemize}
			\item \emph{Setup:}
				Each {\ch relay} $P_i$
				generates a keypair
				$(SK_i, PK_i)$ $\gets G(1^\lambda)$
				and publishes $PK_i$
				by using \Ff[PKI].

			\item \emph{Sending a message:}
				If \Ps wants to
				send $m \in \mathcal{M}$
				to ${\ch R}$
				over the path
				$\Pp = (P_1, \ldots, P_n)$
				with $n\, {\ch \leq}\,  N$
				and wants to allow a reply
				over the path
				$\Pp\ba = (P_1, \ldots, P\ba_{\ch n\ba})$
				with $n\ba\, {\ch \leq}\, N$
				and $P\ba_{\ch n\ba} = \Ps$,
				it chooses
				a randomness \Rr
				and calculates
				\[O_1\! \gets\! \FormO(1, \Rr, m, {\ch R}, \Pp, \Pp\ba\!\!, PK_{\Pp}, PK\ba_{\Pp\ba\!\!})\]
				and sends $O_1$ to $P_1$
				using \Ff[SC].

			\item \emph{Processing an onion:}
				$P_i$ receives $O_i$
				and runs
				\[(O_j, P_j) \gets \Proc(SK_i, O_i, P_i). \]
				If $P_j = \bot$,
				$P_i$ {\ch outputs
					\enquote{Received $m = O_j$ as a reply}
					if $O_j \neq \bot$
					$[\ldots]$.
				}
				If $P_j \neq \bot$,
				$P_j$ is
				a valid relay name
				{\ch or receiver.}
				and $P_i$ generates
				a random $tid$
				and stores
				$(tid, (O_j, P_j))$
				in its outgoing buffer
				and notifies the environment
				about $tid$.

			\item \emph{Sending an onion:}
				When the environment
				instructs $P_i$
				to forward $tid$,
				$P_i$ looks up $tid$
				in its buffer.
				If $P_i$ does not
				find such an entry,
				it aborts.
				Otherwise,
				it finds
				$(tid, (O_j, P_j))$.
				{\ch
					If $P_j$ is a relay name,
					it%
				}
				sends
				$O_j$ to $P_j$
				using \Ff[SC].
				{\ch If $P_j = R$ for a receiver $R$,
					$P_i$ checks if
					$\Reply(m\ba, O_i, P_i, SK_i) \neq \bot$
					for an arbitrary $m\ba$.
					If so,
					$P_i$ chooses a random $rid$,
					stores $(rid, O_i)$
					in its reply buffer,
					and sends $(O_j, rid)$ to $P_j$
					without \Ff[SC].
					If not,
					$P_i$ sends $(O_j, \bot)$ to $P_j$
					without \Ff[SC].
				}

			{\ch \item \emph{Receiving a message:}
				When a receiver
				receives $(m, rid)$
				from a relay $P_i$,
				it outputs
				\enquote{Received message $m$ from $P_i$}
				to the environment.
				If $rid \neq \bot$,
				it additonally outputs
				\enquote{It is repliable with $rid$}.
			}

			{\ch \item \emph{Sending a reply message:}
				When the environment instructs $R$
				to reply to $P_i$ with $rid$ and $m\ba$,
				$R$ sends $(m\ba, rid)$ to $P_i$
				without \Ff[SC].
			}

			\item \emph{\ch Creating a reply onion:}
				{\ch
					When $P_r$ receives $(m\ba, rid)$
					from a receiver $R$,
					$P_r$ looks up $rid$
					in its reply buffer.
					If there is no entry with $rid$,
					$P_r$ stops.
					If it finds $(rid, O)$
					in its buffer,
				}
				it calculates
				\[(O\ba_1, P\ba_1) \gets \Reply(m\ba, O, P_r, SK_r)\]
				and sends $O\ba_1$ to $P\ba_1$ using \Ff[SC].
		\end{itemize}
		\looseequations
	\end{definition}

\section{Sphinx Onion Property Proofs}

	\subsection{\texorpdfstring{RSOR-Tagging Layer Unlinkability (\TFLU)}{RSOR-Forward Layer Unlinkability}}
	\label{sapp:tfluproof}

	\textsc{Theorem \ref{thm:tflu}.}
	\textit{%
		Sphinx satisfies \TFLU
		under the GDH assumption.
	}

	\begin{proof}
		We prove that an adversary
		cannot distinguish
		the $b = 0$ scenario
		from the $b = 1$ scenario
		through a hybrid argument
		starting at $b=0$
		and ending at $b=1$.
		For clarity,
		we separate the proof
		into two cases:
		One where $j = n$
		and the other
		where $j < n$.
		Initially,
		this may seem problematic
		since the adversary
		chooses $j$ adaptively
		after the first round
		of oracle accesses,
		so we cannot predict
		which it will choose beforehand
		for our reductions.
		However,
		every step in our proofs
		is either common to both cases
		(so it does not require
		predicting $j$)
		or only applies
		after the adversary
		has made its choice.
		Note that \citeauthor*{sphinx} index
		Sphinx packet layers starting at 0%
		~\cite{sphinx}.
		Here,
		this means that $P_0 (= \Ps)$ would send
		the Sphinx layer with $\alpha_0$
		in its header.

		\noindent \textbf{Case 1 $(j = n)$:}
		In this case,
		the honest relay
		on the forward path
		is also the exit relay
		of the onion
		and thus the sender
		of the reply onion.

		\textbf{Hybrid \Hh[_0]:}
		This hybrid is just
		the $b = 0$ scenario of \TFLU
		with $j = n$.

		\textbf{Hybrid \Hh[_1]:}
		As a first step,
		we replace the secrets used
		at the honest relay
		(the exit relay)
		with randomness.
		\Hh[_1] replaces
		the random oracle outputs
		$h_*(s_{n-1})$
		with random $\{0,1\}^\kappa$ bitstrings
		when building the reply header%
		\footnote{%
			It also randomizes $b_{n-1}$,
			but we do not need that
			for this proof.
		}.
		\textbf{Proc} requests
		for onions with the challenge
		$\alpha_{n-1}$
		are also served using these random keys
		to keep the oracle's behavior
		consistent.

		\underline{$\Hh[_0] \approx_I \Hh[_1]$:}
		The difference
		between these two hybrids
		reduces to
		$\text{Exp}^\text{RO-KEM-IND-CCA}_{\text{RO-KEM},\Aa}(\kappa)$:
		The \SKEMCCA attacker \Aa
		uses its inputs from the challenger
		and the assumed hybrid distinguisher \Dd
		to build the challenge onion
		and serves requests
		to the \textbf{Proc}
		and \textbf{Reply} oracles
		using its decapsulation oracle.
		\textbf{Proc} requests
		with the challenge $\alpha_{n-1}$
		are served using
		the keys provided by
		the challenger.

		\textbf{Hybrid \Hh[_2]:}
		To ensure that
		only the challenge onion
		is recognized
		for \enquote{challenge processing}
		by the exit relay,
		the \textbf{Proc} oracle
		at the relay
		now returns $(\bot, \bot)$
		on every request
		with $\alpha_{n-1}$
		in its header
		except if
		the rest of the header
		also matches
		the expected header of
		the challenge onion.

		\underline{$\Hh[_1] \approx_I \Hh[_2]$:}
		This reduces to
		$\text{Exp}^\text{sEUF-CMA-vq}_{\mu,\Aa}(\kappa)$
		using the fact that
		the PRF $\mu$ can be viewed
		as a MAC
		with the randomized key
		$h_\mu(s_{n-1})$.
		Any request
		to the \textbf{Proc} oracle
		at the exit relay
		with $\alpha_{n-1}$
		must contain
		a valid $\gamma$ MAC
		for the $\beta$ in the header
		to be processed.
		To notice a difference
		between the two hybrids,
		a distinguisher must submit
		such a request
		with a modified $\beta$
		or $\gamma$.
		Since $\mu$ is sEUF-CMA-vq-secure,
		this request
		constitutes a MAC forgery.

		\textbf{Hybrid \Hh[_3]:}
		This hybrid swaps
		$\pi(h_\pi(s_{n-1}), \cdot)$
		with a random permutation (RP).
		Note that $h_\pi(s_{n-1})$
		is already a random key
		before this hybrid.
		Since Sphinx requires $\pi$
		to be a strong PRP%
		~\cite{sphinx},
		both $\pi$ and $\pi^{-1}$
		can be used as RPs
		after the replacement.

		\underline{$\Hh[_2] \approx_I \Hh[_3]$:}
		A distinguisher \Dd
		for these hybrids
		is easily converted
		into an attacker \Aa
		on $\text{Exp}^\text{prp}_{\pi,\Aa}(\kappa)$.

		\textbf{Hybrid \Hh[_4]:}
		In this hybrid,
		the honest exit relay $P_H$
		only sends the challenge reply
		in response to a \textbf{Reply} request
		if the onion received
		in the corresponding \textbf{Proc} request
		has a payload
		that matches the expected payload
		$\delta_{n-1}$
		of the forward onion
		exactly.

		\underline{$\Hh[_3] \approx_I \Hh[_4]$:}
		Since $\pi^{-1}(h_\pi(s_{n-1}), \cdot)$
		is now an RP,
		every input
		is mapped to a random output.
		In order for
		a distinguisher
		to notice a difference
		between the two hybrids,
		it must submit a modified payload
		attached to the challenge header
		that is accepted by the oracle.
		After decrypting the payload,
		the oracle checks that
		the first $\kappa$ bits
		of the contents are all zero
		and discards the onion
		it that is not the case.
		A manipulated payload
		only starts with $0_\kappa$
		negligibly often,
		so the distinguisher
		only has a negligible chance
		of success.

		\textbf{Hybrid \Hh[_5]:}
		This hybrid
		replaces the contents
		of the forward payload
		with the contents
		as they would be
		in the $b = 1$ scenario.
		The original contents
		are $0_\kappa\|R\|\eta_0\|\tilde{k}\|m$,
		while the replacements
		are $0_\kappa\|\bar{R}\|pad\fa_{\kappa,N}\|\bar{m}$
		with a random receiver $\bar{R}$,
		a random message $\bar{m}$,
		and padding
		instead of a reply header.
		When building
		the challenge reply onion,
		$P_H$ still uses
		the original $\eta_0$ reply header.
		The \textbf{Proc} oracle
		also outputs the original
		message and receiver.

		\underline{$\Hh[_4] \approx_I \Hh[_5]$:}
		The \textbf{Proc}
		and \textbf{Reply} oracles
		at the exit relay
		behave the same way
		in \Hh[_4] and \Hh[_5].
		$P_H$ consistently uses
		$\eta_0$ as the challenge reply header
		in both hybrids.
		The replacement
		of the payload contents
		reduces to
		$\text{Exp}^\text{1-LR-CPA}_{\pi,\Aa}(\kappa)$.

		\textbf{Hybrid \Hh[_6]:}
		In this hybrid,
		we rewind the temporary changes
		made in the hybrids
		\Hh[_4],
		\Hh[_3],
		\Hh[_2],
		and \Hh[_1].

		\underline{$\Hh[_5] \approx_I \Hh[_6]$:}
		Apply the previous arguments
		in reverse.
		This concludes
		the $j = n$ case
		of the \TFLU proof.

		\noindent\textbf{Case 2 $(j < n)$:}
		For the second part
		of this proof,
		the honest relay
		on the forward path
		is not the exit relay,
		i.e.,
		$j < n$.

		\textbf{Hybrid \Hh[_0]:}
		The $b = 0$ case of \TFLU
		with $j < n$.

		\textbf{Hybrid \Hh[_1]:}
		This hybrid performs
		the same steps
		as the hybrids
		\Hh[_1] and \Hh[_2]
		in the $j = n$ case.

		\underline{$\Hh[_0] \approx_I \Hh[_1]$:}
		See the corresponding hybrids
		in the $j = n$ case.

		\textbf{Hybrid \Hh[_2]:}
		In this hybrid,
		we replace
		$\pi(h_\pi(s_{j-1}), \cdot)$
		with an RP.

		\underline{$\Hh[_1] \approx_I \Hh[_2]$:}
		Analogous to Case 1's
		$\Hh[_2] \approx_I \Hh[_3]$.

		\textbf{Hybrid \Hh[_3]:}
		Previously (in \Hh[_2]),
		if $\delta'_{j-1}$ matches
		the original $\delta_{j-1}$,
		the resulting output
		is the correct $\delta_j$.
		Otherwise,
		the RP causes the output
		to be a uniformly random
		$\{0,1\}^{l_\pi(\kappa)}$ string.

		In \Hh[_3],
		instead of actually processing
		$\delta'_{j-1}$,
		the honest relay only checks
		whether $\delta'_{j-1} = \delta_{j-1}$.
		If so,
		$\delta_j$ is used
		as the output onion's payload.
		Otherwise,
		a random
		$\{0,1\}^{l_\pi(\kappa)}$ string
		is output instead.

		\underline{$\Hh[_2] \approx_I \Hh[_3]$:}
		If a distinguisher \Dd
		chooses to submit
		the correct challenge payload,
		the output in both hybrids
		is identical.
		If \Dd sends
		a manipulated payload,
		\Hh[_2] outputs
		a new RP output
		while \Hh[_3] produces
		a completely random string.
		These two distributions
		are only distinguishable
		if \Hh[_3] happens to choose
		$\delta_j$ as its
		random output,
		which \Hh[_2] would never do.
		The probability
		of that happening
		is negligible.

		\textbf{Hybrid \Hh[_4]:}
		This hybrid is analogous
		to Case 1's \Hh[_5],
		replacing the contents
		of $\delta_{j-1}$
		(originally
		$\delta_j$)
		with
		$0_\kappa\|\bar{R}\|pad\fa_{\kappa,N}\|\bar{m}$,
		$\bar{R}$ and $\bar{m}$
		being a random receiver
		and message.

		\underline{$\Hh[_3] \approx_I \Hh[_4]$:}
		Analogous to \Hh[_5]
		in Case 1.

		\textbf{Hybrid \Hh[_5]:}
		In the honest relay's challenge processing,
		\Hh[_5] always produces
		the same challenge header
		(the one belonging to
		the challenge onion's layer $O_j$)
		without actually processing
		the header input
		the relay is given.
		The challenge output of $P_H$
		now only depends on
		whether the payload was manipulated
		(i.e., tagged).

		\underline{$\Hh[_4] \approx_I \Hh[_5]$:}
		Due to \Hh[_1],
		the honest relay
		only performs
		the challenge processing steps
		on headers that match
		the challenge header exactly.
		Thus,
		both hybrids always output
		the identical challenge header
		for the challenge onion.

		\textbf{Hybrid \Hh[_6]:}
		This hybrid
		replaces the PRG output
		$\rho(h_\rho(s_{j-1}))$
		with a random string.

		\underline{$\Hh[_5] \approx_I \Hh[_6]$:}
		Given a distinguisher \Dd
		for the two hybrids,
		construct an attacker \Aa
		on $\text{Exp}^{prg}_{\rho,\Aa}(\kappa)$.

		\textbf{Hybrid \Hh[_7]:}
		In this hybrid,
		we replace the first
		$(2(N-j)+3)\kappa$ bits
		of the contents
		of $\beta_{j-1}$.
		These bits correspond to
		the address of the relay $P_{j+1}$,
		the MAC $\gamma_j$,
		and $\beta_{j[\ldots(2N-1)\kappa-1]}$.
		The rest of $\beta_j$
		constitutes padding
		that we leave unchanged.
		The replacement is
		$\{\ast\|0_\kappa\|rand_{[(2(N-j)+2)\kappa-|\delta|]}\}$.

		\underline{$\Hh[_6] \approx_I \Hh[_7]$:}
		Since $\rho(h_\rho(s_{j-1}))$
		is a random string,
		this change can be reduced to
		$\text{Exp}^\text{1-LR-CPA}_{\text{OTP},\Aa}$.
		The new contents
		are exactly what $\beta_{j-1}$
		would contain
		if $P\ba_H$
		were the last relay
		on the path.
		The suffix 
		\[[(2(N-j)+3)\kappa\ldots (2N+1)\kappa-1]\]
		of $\beta_{j-1}$
		is $\Phi_{j-1}$
		by construction%
		\footnote{%
			Technically,
			Sphinx uses the output
			of $\rho$ twice
			when building a header:
			Once to generate the padding
			and a second time
			to encrypt $\beta$.
			However,
			these use different substrings
			of the PRG output.
			\Aa can thus submit
			the two different $\beta$ contents
			with an appropriate zero padding
			to the 1-LR-CPA challenger
			to extract the random string
			required for the padding calculation.
		}.
		With this change,
		the layers
		$\beta_0$, \ldots, $\beta_{j-1}$
		are now independent of
		the later layers
		$\beta_j$, \ldots, $\beta_{n-1}$.

		\textbf{Hybrid \Hh[_8]:}
		The second part
		of the challenge onion
		still contains the padding
		$\Phi_0$, \ldots, $\Phi_{j-1}$
		nested in
		$\beta_j$'s $\Phi_j$ padding.
		To alleviate this,
		\Hh[_8] replaces $\Phi_j$
		with a random string
		of length $2j\kappa$.

		\underline{$\Hh[_7] \approx_I \Hh[_8]$:}
		In \Hh[_7],
		$\Phi_j$ is calculated
		from $\Phi_{j-1}$ as
		\[\Phi_j \gets \rho(h_\rho(s_{j-1}))_{[(2(N-j)+3)\kappa\ldots]} \oplus \{\Phi_{j-1}\|0_{2\kappa}\}.\]
		Since $\rho(h_\rho(s_{j-1}))$
		is a random string,
		the replacement
		\[\Phi_j \gets \rho(h_\rho(s_{j-1}))_{[(2(N-j)+3)\kappa\ldots]}\]
		is indistinguishable
		from the original.
		As a result,
		this change
		reduces to
		$\text{Exp}^\text{1-LR-CPA}_{\text{OTP},\Aa}$.

		\textbf{Hybrid \Hh[_9]:}
		This hybrid
		replaces the KEM instance used
		for the second part
		of the challenge onion
		after the honest relay.
		Previously,
		$\alpha_j = \alpha_{j-1}^{b_{j-1}}$
		and $s_j = y_j^{xb_0\cdots b_{j-1}}$.
		Now,
		\Hh[_9] picks a new
		$x' \gets^R \mathbb{Z}_q^*$,
		setting
		$\alpha_j = g^{x'}$
		and $s_j = y_j^{x'}$
		and adjusting the later
		$\alpha_i$ and $s_i$
		following them accordingly.

		\underline{$\Hh[_8] \approx_I \Hh[_9]$:}
		\Hh[_1] randomizes $b_{j-1}$
		into a unformly distributed element
		of $\mathbb{Z}_q^*$.
		It follows that
		$\alpha_{j-1}^{b_{j-1}}$
		and $g^{x'}$
		are identically distributed.
		The same argument holds
		for the later $\alpha$s
		and secrets.

		\textbf{Hybrid \Hh[_{10}]:}
		This hybrid
		\enquote{fixes} the second part
		of the onion
		so that it becomes
		a complete onion
		starting at $P_S$ again.
		To that end,
		\Hh[_{10}] starts building
		new $O'_0$, \ldots, $O'_{j-1}$
		onion layers
		that follow the same path
		as the original $O_0$, \ldots, $O_{j-1}$.
		These new layers
		are built as a prefix
		to $O_j$,
		so the payload content
		of $\delta'_{j-1}$ is
		$\delta_j$
		and $\beta'_{j-1}$ is formed
		with $\beta_j$
		in its contents.
		The random oracle outputs
		$h_*(s'_{j-1})$
		and $h_b(\alpha'_{j-1}, s'_{j-1})$
		are randomized.

		Most importantly,
		$\alpha_j$ is now formed
		as $\alpha_{j-1}^{\prime b'_{j-1}}$,
		with the secrets being built
		analogously.

		\underline{$\Hh[_9] \approx_I \Hh[_{10}]$:}
		The new layers
		$O'_0$, \ldots, $O'_{j-1}$
		are never actually given
		to the adversary.
		Their construction is thus
		entirely invisible to
		the attacker
		except for the change
		in how $\alpha_j$ is formed.
		Using $\Hh[_8] \approx_I \Hh[_9]$'s argument,
		$\alpha_j$ still has
		the same distribution
		in both hybrids.

		\textbf{Hybrid \Hh[_{11}]:}
		This hybrid
		replaces
		$\rho(h_\rho(s'_{j-1}))$
		with a random string.

		\underline{$\Hh[_{10}] \approx_I \Hh[_{11}]$:}
		See $\Hh[_5] \approx_I \Hh[_6]$.

		\textbf{Hybrid \Hh[_{12}]:}
		In \Hh[_8],
		$\Phi_j$ is replaced
		with a random string
		instead of containing
		the previous $\Phi_{j-1}$.
		Now,
		we replace that random string
		with
		\[\rho(h_\rho(s'_{j-1}))_{[(2(N-j)+3)\kappa\ldots]} \oplus \{\Phi_{j-1}\|0_\kappa\},\]
		so that
		$\Phi_j$ is formed
		as the $j$-th layer
		of padding in the
		$O'_0$, \ldots, $O_j$, \ldots, $O_{n-1}$ onion.

		\underline{$\Hh[_{11}] \approx_I \Hh[_{12}]$:}
		Analogous to \Hh[_8].

		\textbf{Hybrid \Hh[_{13}]:}
		This hybrid
		replaces the randomized oracle outputs
		for $s'_{j-1}$
		with the actual outputs
		$h_*(s'_{j-1})$
		and $h_b(\alpha'_{j-1}, s'_{j-1})$.

		\underline{$\Hh[_{12}] \approx_I \Hh[_{13}]$:}
		See $\Hh[_1] \approx_I \Hh[_0]$.

		\textbf{Hybrid \Hh[_{14}]:}
		In this hybrid,
		we rewind
		all of the temporary changes
		made in the previous hybrids
		in the reverse order:
		\Hh[_{11}],
		\Hh[_6],
		\Hh[_2],
		and \Hh[_1].

		\underline{$\Hh[_{13}] \approx_I \Hh[_{14}]$:}
		Apply the previous arguments
		in reverse.
		This concludes
		the $j < n$ case
		of the \TFLU proof.

		We have proven that
		Sphinx satisfies
		\TFLU.
	\end{proof}

	\subsection{\texorpdfstring{RSOR-Backw. Layer Unlinkability (\NBLU)}{RSOR-Backward Layer Unlinkability}}
	\label{sapp:nbluproof}

	\begin{theorem}
		Sphinx satisfies \NBLU
		under the GDH assumption.
	\end{theorem}

	\begin{proof}
		We prove that an adversary
		cannot distinguish
		the $b = 0$ scenario
		from the $b = 1$ scenario
		through a hybrid argument
		starting at $b=0$
		and ending at $b=1$.
		For clarity,
		we separate the proof
		into two cases:
		One where $j\ba = 0$
		and the other
		where $j\ba > 0$.

		\noindent\textbf{Case 1 $(j\ba = 0)$:}
		In this case,
		the honest relay
		on the return path
		is identical to
		the exit relay
		on the forward path
		and is thus also
		the sender of the reply onion.

		\textbf{Hybrid \Hh[_0]:}
		This hybrid is just
		the $b = 0$ case of \NBLU
		with $j\ba = 0$.

		\textbf{Hybrid \Hh[_1]:}
		\Hh[_1] replaces
		the random oracle outputs
		$h_*(s\ba_{n\ba-1})$
		with random $\{0,1\}^\kappa$ bitstrings
		when building the reply header.
		Any \textbf{Proc} requests
		for onions with the challenge
		$\alpha\ba_{n\ba-1}$
		are also served using these random keys.

		\underline{$\Hh[_0] \approx_I \Hh[_1]$:}
		This difference reduces
		to \SKEMCCA.
		See hybrid \Hh[_1]
		in the \TFLU proof
		for details.

		\textbf{Hybrid \Hh[_2]:}
		To ensure that
		only the challenge reply onion
		is \enquote{absorbed}
		by the reply receiver,
		the \textbf{Proc} oracle
		now returns $(\bot, \bot)$
		on every request
		with the challenge
		$\alpha\ba_{n\ba-1}$
		except if the rest of the header
		also matches the challenge reply
		(in that case,
		no output is produced at all).

		\underline{$\Hh[_1] \approx_I \Hh[_2]$:}
		See hybrid \Hh[_2]
		in the \TFLU proof.

		\textbf{Hybrid \Hh[_3]:}
		We repeat the changes
		in hybrids \Hh[_1]
		and \Hh[_2]
		for the honest relay $P_H$
		to randomize $h_*(s_{n-1})$
		and reject challenge onions
		with modified headers
		in \textbf{Proc} at $P_H$.

		\underline{$\Hh[_2] \approx_I \Hh[_3]$:}
		Analogous to
		$\Hh[_0] \approx_I \Hh[_1]$
		and $\Hh[_1] \approx_I \Hh[_2]$.

		\textbf{Hybrid \Hh[_4]:}
		This hybrid exchanges
		$\pi(h_\pi(s_{n-1}), \cdot)$
		with an RP.

		\underline{$\Hh[_3] \approx_I \Hh[_4]$:}
		A distinguisher \Dd
		for these hybrids
		is easily converted
		into an attacker \Aa on
		$\text{Exp}^\text{prp}_{\pi,\Aa}(\kappa)$.

		\textbf{Hybrid \Hh[_5]:}
		In this hybrid,
		the honest exit relay $P_H$
		only sends the challenge reply
		in response to a \textbf{Reply} request
		if the onion received
		in the corresponding \textbf{Proc} request
		has a payload
		that matches the expected payload
		$\delta_{n-1}$
		of the forward onion
		exactly.

		\underline{$\Hh[_4] \approx_I \Hh[_5]$:}
		See hybrid \Hh[_4]
		in \TFLU's
		Case 1.

		\textbf{Hybrid \Hh[_6]:}
		This hybrid replaces
		the reply header $\eta_0$
		and symmetric key $\tilde{k}$
		in the contents of $\delta_{n-1}$
		with a new reply header
		that uses the same path,
		but different randomness
		and a random $\tilde{k}'$
		when building the forward onion.
		The challenge reply onion's
		reply header
		is no longer read from
		the payload of the forward onion.
		Instead,
		the actual reply header
		and symmetric key
		are built \enquote{at} $P_H$
		when the challenge \textbf{Reply}
		is requested.

		\underline{$\Hh[_5] \approx_I \Hh[_6]$:}
		The behavior of the \textbf{Proc}
		and \textbf{Reply} oracles
		is indistinguishable
		between the two hybrids,
		since both process the onion
		and reply to it
		using the same header.
		The only other change
		is to the contents
		of $\delta_{n-1}$.
		We can construct
		an attacker \Aa
		on $\text{Exp}^\text{1-LR-CPA}_{\pi,\Aa}(\kappa)$
		using any distinguisher \Dd.

		\textbf{Hybrid \Hh[_7]:}
		In this hybrid,
		replace $\pi(h_\pi(s\ba_{n\ba-1}), \cdot)$
		and $\pi(\tilde{k}, \cdot)$
		with RPs
		when building
		the reply onion.

		\underline{$\Hh[_6] \approx_I \Hh[_7]$:}
		Analogous to
		$\Hh[_3] \approx_I \Hh[_4]$.

		\textbf{Hybrid \Hh[_8]:}
		Previously,
		the first layer $O\ba_1$
		of the reply onion
		had a payload
		encrypted with the RP
		$\pi(\tilde{k}, \cdot)$.
		We now replace this permutation
		with 
		$\pi(h_\pi(s\ba_0), \pi(h_\pi(s\ba_1), \cdots\pi(h_\pi(s\ba_{n\ba-1}), \cdot)\cdots))$
		while encrypting the same contents.
		The new permutation
		corresponds to
		how a forward onion payload
		is encrypted at the sender.

		\underline{$\Hh[_7] \approx_I \Hh[_8]$:}
		Since chaining the permutations
		$\pi(h_\pi(s\ba_i), \cdot)$
		after the RP
		$\pi(h_\pi(s\ba_{n\ba-1}), \cdot)$
		results in a new RP,
		we have simply replaced
		one RP with another.

		\textbf{Hybrid \Hh[_9]:}
		Until now,
		the contents
		of the reply onion payload were
		$0_\kappa\|\text{pad}\ba_{\kappa,N}\|m\ba$,
		where $m\ba$ is
		the adversary-chosen message.
		We replace them with
		$0_\kappa\|\bar{R}\|\text{pad}\fa_{\kappa,N}\|\bar{m}$,
		where $\bar{R}$ and $\bar{m}$
		are randomly chosen receivers
		and messages.

		\underline{$\Hh[_8] \approx_I \Hh[_9]$:}
		Analogous to
		$\Hh[_5] \approx_I \Hh[_6]$.

		\textbf{Hybrid \Hh[_{10}]:}
		This hybrid replaces
		$\rho(h_\rho(s\ba_{n\ba-1}))$
		with a random string
		when building the
		\enquote{reply} onion header.

		\underline{$\Hh[_9] \approx_I \Hh[_{10}]$:}
		Reduce to
		$\text{Exp}^\text{prg}_{\rho,\Aa}(\kappa)$.

		\textbf{Hybrid \Hh[_{11}]:}
		When building
		the \enquote{reply} onion header,
		\Hh[_{11}] uses $*\|0_\kappa$
		instead of $\Ps\|I$
		in the contents of $\beta\ba_{n\ba-1}$.

		\underline{$\Hh[_{10}] \approx_I \Hh[_{11}]$:}
		Since the change
		in $\beta\ba_{n\ba-1}$'s contents
		does not affect the padding,
		we can reduce this change to
		$\text{Exp}^\text{1-LR-CPA}_{\text{OTP},\Aa}$
		without further provisions.
		At this stage,
		$O\ba_1$ is built
		just like $\bar{O}_1$
		in the $b = 1$ case
		of \NBLU.

		\textbf{Hybrid \Hh[_{12}]:}
		This hybrid rewinds
		all of the temporary changes
		made in the previous hybrids
		\Hh[_{10}],
		\Hh[_7],
		\Hh[_5],
		\Hh[_4],
		\Hh[_3],
		\Hh[_2],
		and \Hh[_1]
		in that order.

		\underline{$\Hh[_{11}] \approx_I \Hh[_{12}]$:}
		Apply the previous arguments
		in reverse.
		This concludes
		the $j\ba = 0$ case
		of the \NBLU proof.

		\noindent\textbf{Case 2 $(j\ba > 0)$:}
		Now,
		we consider the case
		where the honest relay
		is on the reply path
		of the challenge onion.

		\textbf{Hybrid \Hh[_0]:}
		The $b\!=\!0$ scenario of \NBLU
		with $j\ba\! >\! 0$.

		\textbf{Hybrid \Hh[_1]:}
		This hybrid performs
		the same steps as
		the hybrids
		\Hh[_1], \Hh[_2], and \Hh[_3]
		in Case 1.

		\underline{$\Hh[_0] \approx_I \Hh[_1]$:}
		See the corresponding hybrids
		in Case 1.

		\textbf{Hybrid \Hh[_2]:}
		In this hybrid,
		we exchange
		the two permutations
		$\pi^{-1}(h_\pi(s\ba_{j\ba-1}), \cdot)$
		and $\pi(h_\pi(s\ba_{n\ba-1}), \cdot)$
		for RPs.

		\underline{$\Hh[_1] \approx_I \Hh[_2]$:}
		A distinguisher \Dd
		for this hybrid
		is easily converted
		into attackers \Aa and \Bb on
		$\text{Exp}^\text{prp}_{\pi^{-1},\Aa}(\kappa)$.
		and $\text{Exp}^\text{prp}_{\pi,\Bb}(\kappa)$.

		\textbf{Hybrid \Hh[_3]:}
		In this hybrid,
		we change how
		the challenge onion's payload
		$\delta\ba_{j\ba-1}$
		is processed at $P_H$.
		Normally,
		Sphinx calculates
		$\delta\ba_{j\ba}$
		by applying
		the $\pi^{-1}$ (P)RP
		to the payload.
		We replace $\pi^{-1}$ with the RP
		$\pi(h_\pi(s\ba_{j\ba}), \pi(h_\pi(s\ba_{j\ba+1}), \cdots\pi(h_\pi(s\ba_{n\ba-1}), \cdot)\cdots))$.

		\underline{$\Hh[_2] \approx_I \Hh[_3]$:}
		Analogous to Case 1's \Hh[_8].

		\textbf{Hybrid \Hh[_4]:}
		Now,
		instead of running the RP
		on $\delta\ba_{j\ba-1}$
		during the challenge processing,
		\Hh[_4] runs it on
		$0_\kappa\|\bar{R}\|\text{pad}\fa_{\kappa,N}\|\bar{m}$
		for a random receiver $\bar{R}$
		and message $\bar{m}$.
		This completely replaces
		the original,
		adversary-chosen payload
		with a forward payload.

		\underline{$\Hh[_3] \approx_I \Hh[_4]$:}
		See $\Hh[_5] \approx_I \Hh[_6]$
		in Case 1.

		\textbf{Hybrid \Hh[_5]:}
		In this hybrid,
		$P_H$ does not process
		the challenge onion.
		Instead,
		it outputs the header of $O\ba_{j\ba}$
		and the payload
		as defined in
		\Hh[_3] and \Hh[_4].

		\underline{$\Hh[_4] \approx_I \Hh[_5]$:}
		The processing output
		to the adversary
		is identical in both
		\Hh[_4] and \Hh[_5].

		\textbf{Hybrid \Hh[_6]:}
		\Hh[_6] replaces
		$\rho(h_\rho(s\ba_{j\ba-1}))$
		with a random string
		of the same length
		when building $O\ba_1$.
		
		\underline{$\Hh[_5] \approx_I \Hh[_6]$:}
		See Case 1's \Hh[_{10}].

		\textbf{Hybrid \Hh[_7]:}
		In this hybrid,
		we replace the
		first $(2(N-j\ba)+3)\kappa$ bits
		of the contents of
		$\beta\ba_{j\ba-1}$
		with randomness
		when building $O\ba_1$.
		This corresponds to
		the next relay address
		and the next MAC
		as well as
		the $(2(N-(j\ba-1))+3)\kappa$-bit prefix of
		$\beta\ba_{j\ba}$
		that does not contain padding.
		$\beta\ba_{j\ba}$ itself
		is still used for
		the $O\ba_{j\ba}$
		reinserted at $P_H$.

		\underline{$\Hh[_6] \approx_I \Hh[_7]$:}
		Reduce this change to
		$\text{Exp}^\text{1-LR-CPA}_{\text{OTP},\Aa}(\kappa)$
		like in Case 1's \Hh[_{11}].
		After this change,
		$\beta\ba_{j\ba-1}$ can be built
		without using any of the secrets
		$s\ba_{j\ba-1}$, \ldots, $s\ba_{n\ba-1}$
		or the random oracle outputs
		derived from them.

		\textbf{Hybrid \Hh[_8]:}
		This hybrid replaces the keys
		used to generate $O\ba_{j\ba}$.
		Instead of choosing
		$\alpha\ba_{j\ba} = \alpha_{j\ba-1}^{\leftarrow b\ba_{j\ba-1}}$
		and $s\ba_{j\ba} = y_{j\ba}^{x\ba b\ba_0 \cdots b\ba_{j\ba-1}}$,
		\Hh[_8] picks
		a new $x' \gets^R \mathbb{Z}^*_q$
		and lets $\alpha\ba_{j\ba} = g^{x'}$
		and $s\ba_{j\ba} = y_{j\ba}^{x'}$.
		The later $\alpha\ba_{i}$,
		$s\ba_{i}$, $i > j\ba$
		are calculated accordingly
		\footnote{%
			In \Hh[_1],
			the random oracle outputs
			for $s\ba_{n\ba-1}$
			at the reply receiver $P_S$
			are randomized.
			\Hh[_8] remains consistent
			with this behavior
			by also randomizing
			the random oracle outputs
			for the \enquote{new} $s\ba_{n\ba-1}$.
			Since the \enquote{original} $s\ba_{n\ba-1}$
			is not in use,
			there is still exactly
			one set of randomized oracle outputs
			at $P_S$.
		}.

		\underline{$\Hh[_7] \approx_I \Hh[_8]$:}
		See hybrid \Hh[_9]
		in \TFLU's Case 2.

		\textbf{Hybrid \Hh[_9]:}
		In \Hh[_9],
		we move the first
		$2j\ba\kappa$ bits
		of $\Phi_{n\ba-1}$
		(corresponding to $\Phi_{j\ba}$)
		into the $rand$-padding
		inside $\beta\ba_{n\ba-1}$.
		To do so,
		the $rand$-padding
		is extended by $2j\ba\kappa$ bits
		and $\Phi_{n\ba-1}$
		is truncated to
		$\Phi_{n\ba-1[2j\ba\kappa\ldots]}$.

		\underline{$\Hh[_8] \approx_I \Hh[_9]$:}
		For this step,
		assume a distinguisher \Dd
		for the two hybrids.
		We will construct
		an attacker \Aa
		on $\text{Exp}^{\text{1-LR-CPA}}_{\text{OTP},\Aa}(2j\kappa)$.
		\Aa submits
		\begin{align*}
			&\rho(h_\rho(s\ba_{j\ba}))_{[(2(N-(j\ba+1))+3)\kappa\ldots (2(N-1)+3)\kappa - 1]}\\
			\oplus &\cdots & \\
			\oplus &\rho(h_\rho(s\ba_{n-2}))_{[(2(N-(n-1))+3)\kappa\ldots (2(N-(n-1-j\ba))+3)\kappa - 1]}
		\end{align*}
		and
		\[\rho(h_\rho(s\ba_{n-1}))_{[(2(N-n)+3)\kappa\ldots (2(N-(n-j\ba)) + 3)\kappa - 1]}\]
		to the challenger
		and uses the challenge ciphertext
		as the substring
		\[[(2(N-(j\ba+1))+3)\kappa\ldots (2(N-1)+3)\kappa-1]\]
		of $\beta\ba_{n-1}$.
		$\Phi\ba_{j\ba}$
		is already a random string
		due to it being
		the result of an XOR operation
		with the random string
		$\rho(h_\rho(s\ba_{j\ba-1}))$,
		so the first scenario
		simulates \Hh[_3]
		and the second
		simulates \Hh[_4].

		\textbf{Hybrid \Hh[_{10}]:}
		This hybrid
		replaces $\rho(h_\rho(s\ba_{n\ba-1}))$
		with a random string
		when building $O\ba_{j\ba}$.

		\underline{$\Hh[_9] \approx_I \Hh[_{10}]$:}
		Analogous to
		this case's
		$\Hh[_5] \approx_I \Hh[_6]$.

		\textbf{Hybrid \Hh[_{11}]:}
		Just like in Case 1's \Hh[_{11}],
		we use $*\|0_\kappa$
		instead of $\Ps\|I$
		in the contents of $\beta\ba_{n\ba-1}$.

		\underline{$\Hh[_{10}] \approx_I \Hh[_{11}]$:}
		Analogous to Case 1's
		$\Hh[_{10}] \approx_I \Hh[_{11}]$.

		\textbf{Hybrid \Hh[_{12}]:}
		This hybrid rewinds
		the temporary changes
		made in the previous hybrids:
		\Hh[_{10}],
		\Hh[_7],
		\Hh[_6],
		\Hh[_2],
		and \Hh[_1]
		are unwound
		in that order.
		Note that unwinding \Hh[_7]
		means replacing
		the random contents
		in $\beta\ba_{j\ba-1}$
		with the
		\enquote{legitimate} rest
		of the reply header,
		not the $\beta$s
		that were transformed
		into the $\bar{O}_1$ header.

		\underline{$\Hh[_{11}] \approx_I \Hh[_{12}]$:}
		Apply the previous arguments
		in reverse.
		This concludes the $j\ba > 0$ case
		of the \NBLU proof.

		We have now shown
		that $\Hh[_0] \approx_I \Hh[_{12}]$
		with $\Hh[_0] = \NBLU_{b=0}$
		and $\Hh[_{12}] = \NBLU_{b=1}$
		for both the
		$j\ba = 0$
		and $j\ba > 0$ cases,
		proving that
		Sphinx satisfies \NBLU.
	\end{proof}

	\subsection{\texorpdfstring{RSOR-Tail Indistinguishability (\NTI)}{RSOR-Tail Indistinguishability}}
	\label{sapp:ntiproof}

	\begin{theorem}
		Sphinx satisfies \NTI
		under the GDH assumption.
	\end{theorem}
	\begin{proof}
		We prove that
		an adversary cannot distinguish
		the $b = 0$ scenario
		from the $b = 1$ scenario
		through a hybrid argument
		starting at $b = 0$
		and ending at $b = 1$.
		We gradually transform
		the $O_{j}$ onion
		into the $\bar{O}_0$ onion
		in successive hybrids.

		\textbf{Hybrid \Hh[_0]:}
		This hybrid is just the $b = 0$ case of \NTI.

		\textbf{Hybrid \Hh[_1]}:
		If $j = 0$,
		the following hybrids
		do nothing.
		Skip to hybrid \Hh[_6]
		in that case.
		In this hybrid,
		we begin truncating the forward path.
		When building $O_{j}$,
		choose $\alpha_i := g^{x'}$
		and $s_i := y_i^{x'}$
		with a random $x' \in \mathbb{Z}^*_q$.

		\underline{$\Hh[_0] \approx_I \Hh[_1^1]$}:
		See hybrid \Hh[_9]
		in the \TFLU proof's
		Case 2.
		Note that $b_{j-1}$
		is a random oracle output
		that is never used elsewhere.

		\textbf{Hybrid \Hh[_2]:}
		When building the Sphinx packet,
		replace $h_\rho(s_{j-1})$ with
		a random $\{0,1\}^\kappa$-bitstring.
		$h_\rho(s_{j-1})$ is only required
		to calculate $\Phi_j$.

		\underline{$\Hh[_1^j] \approx_I \Hh[_2]$:}
		$s_{j-1}$ behaves like
		a uniformly random group element.
		By definition of a random oracle,
		these hybrids are indistinguishable.

		\textbf{Hybrid \Hh[_3]:}
		Replace $\rho(h_\rho(s_{j-1}))$
		with a random string.
		Since $\Phi_j$ is calculated
		from an XOR operation
		with $\rho(h_\rho(s_{j-1}))$,
		it is now also a random string.

		\underline{$\Hh[_2] \approx_I \Hh[_3]$:}
		See hybrid \Hh[_6]
		in the \TFLU's
		Case 2.

		\textbf{Hybrid \Hh[_4]:}
		When building $\beta_{n-1}$,
		\Hh[_3] extends
		the random bits in its contents
		by $2j\kappa$ extra random bits
		and truncates $\Phi_{n-1}$
		by the same amount.

		\underline{$\Hh[_3] \approx_I \Hh[_4]$:}
		See hybrid \Hh[_9]
		in \NBLU's
		Case 2.

		\textbf{Hybrid \Hh[_5]:}
		This hybrid does not generate
		$\alpha_0$, \ldots, $\alpha_{j-1}$,
		$\beta_0$, \ldots, $\beta_{j-1}$,
		$\gamma_0$, \ldots, $\gamma_{j-1}$,
		$\delta_0$, \ldots, $\delta_{j-1}$,
		or $\Phi_0$, \ldots, $\Phi_{j-1}$.

		\underline{$\Hh[_4] \approx_I \Hh[_5]$:}
		The parts of the packet destined
		for the path prefix
		are not used in $\Hh[_4]$,
		so not generating them
		in the first place goes unnoticed
		by any distinguisher.
		
		\textbf{Hybrid \Hh[_6]:}
		If $j\ba = n\ba$,
		we can skip the following hybrids
		because the original
		and truncated paths
		are identical.
		We thus assume
		$j\ba < n$
		in the following.
		This hybrid replaces
		$h_\rho(s\ba_{j\ba-1})$,
		$h_\mu(s\ba_{j\ba-1})$,
		and $h_\pi(s\ba_{j\ba-1})$
		with random $\{0,1\}^\kappa$ strings.
		In order to process \textbf{Proc} requests
		for $P\ba_H$ with $\alpha$s
		that are identical to
		the $\alpha$ of the challenge header,
		\Hh[_6] uses the new random keys.

		\underline{$\Hh[_5] \approx_I \Hh[_6]$:}
		See hybrid \Hh[_1]
		in \TFLU's
		Case 1.

		\textbf{Hybrid \Hh[_7]:}
		This hybrid adjusts the processing
		of onions
		at the second honest relay $P\ba_{j\ba}$
		by returning $(\bot, \bot)$
		for any \textbf{Proc} request
		with the challenge $\alpha$
		unless the entire header
		matches the one of
		the challenge onion.

		\underline{$\Hh[_6] \approx_I \Hh[_7]$:}
		See hybrid \Hh[_2]
		in \TFLU's
		Case 1.

		\textbf{Hybrid \Hh[_8]:}
		This hybrid replaces
		$\rho(h_\rho(s\ba_{j\ba-1}))$
		with a random string
		when forming $\beta\ba_{j\ba-1}$
		and $\Phi_{j\ba}$.

		\underline{$\Hh[_7] \approx_I \Hh[_8]$:}
		The difference between
		these two hybrids reduces to 
		$\text{Exp}^{prg}_{\rho,\Aa}(\kappa)$.

		\textbf{Hybrid \Hh[_9]:}
		This hybrid replaces the actual contents
		of $\beta\ba_{j\ba-1}$.
		In $\Hh[_8]$,
		these were
		\[\{n\ba_{j\ba}\|\gamma\ba_{j\ba}\|\beta\ba_{j\ba {[\ldots(2N-1)\kappa-1]}}\}.\]
		$\Hh[_9]$ replaces
		the first
		$(2(N-j\ba)+3)\kappa$ bits
		of that with
		\[\{\Ps\|I\ba\|rand_{[(2(N-j\ba)+2)\kappa-|\Delta|]}\}.\]
		The rest of $\beta\ba_{j\ba-1}$
		is unchanged padding.

		\underline{$\Hh[_8] \approx_I \Hh[_9]$:}
		See $\Hh[_3] \approx_I \Hh[_4]$.

		\textbf{Hybrid \Hh[_{10}]:}
		This hybrid reverts
		the temporary changes
		in the previous hybrids:
		\Hh[_8],
		\Hh[_7],
		and \Hh[_6]
		are unwound
		in that order.

		\underline{$\Hh[_9] \approx_I \Hh[_{10}]$:}
		Apply the previous arguments
		in reverse.

		\textbf{Hybrid \Hh[_{13}]:}
		This hybrid does not generate
		$\alpha\ba_{j\ba}$, \ldots, $\alpha\ba_{n\ba-1}$,
		$\beta\ba_{j\ba}$, \ldots, $\beta\ba_{n\ba-1}$,
		$\gamma\ba_{j\ba}$, \ldots, $\gamma\ba_{n\ba-1}$,
		$h_\pi(s\ba_{j\ba})$, \ldots, $h_\pi(s\ba_{n\ba-1})$,
		or $\Phi\ba_{j\ba}$, \ldots, $\Phi\ba_{n\ba-1}$.

		\underline{$\Hh[_{12}] \approx_I \Hh[_{13}]$:}
		These components are no longer required
		to form the challenge packet:
		The later $\beta\ba$s have been replaced
		by the new contents of $\beta\ba_{j\ba-1}$
		along with the later $\Phi\ba$s.

		$\Hh[_{13}]$ is the $b = 1$ case
		of \NTI.
		Since $\Hh[_0] \approx_I \Hh[_{13}]$,
		Sphinx satisfies \NTI.
	\end{proof}

\section{\texorpdfstring{\FRN UC Realization Proof}{FRNRE UC Realization Proof}}
	\label{app:frn-uc-proof}

	\noindent
	\textbf{\autoref{thm:rnre-props}.}
	\textit{
		An RSOR protocol
		according to
		\autoref{ssec:reor-schemes}
		with a PRP-encrypted payload
		that satisfies RSOR-Correctness,
		Tagging-Forward Layer Unlinkability,
		RSOR-Backward Layer Unlinkability,
		and RSOR-Tail Indistinguishability
		securely realizes \FRN.
	}

	This proof
	is taken from
	\citea{or_replies}
	and modified to
	fit RSOR.
	Our modifications
	are given in
	\tch{this style}.
	Some parts of the proof
	also make use of elements
	from \citeauthor*{break_onion}'s proof
	of UC-realization
	in~\cite{break_onion} ---
	these parts
	are explicitly marked
	with citations.

	\begin{proof}
		For UC-realization,
		we show that
		every attack
		on the real world protocol $\Pi$
		can be simulated
		by an ideal world attack
		without the environment
		being able to distinguish those.
		We first describe the simulator S.
		Then we show
		indistinguishability of
		the environment's view
		in the real and ideal world.

	\paragraph*{Constructing Simulator \Ss}

	\Ss interacts with
	the ideal functionality {\ch \FRN}
	as the ideal world adversary,
	and simulates
	the real-world honest parties
	for the real world adversary \Aa.
	All outputs \Aa does
	are forwarded
	to the environment by \Ss.
	First,
	\Ss carries out
	the trusted set-up stage:
	It generates
	public and private key pairs
	for all the real-world honest parties.
	\Ss then sends
	the respective public keys
	to \Aa
	and receives
	the real world corrupted parties' public keys
	from \Aa.
	The simulator \Ss
	maintains four internal data structures:
	\begin{itemize}
		\item The $r$-list
			consisting of tuples
			of the form
			(\textit{{\ch onion}, {\ch prevRelay,} nextRelay, tid{\ch, $a$})}.
			Each entry
			in this list
			corresponds to a stage
			in processing an onion
			that belongs
			to a communication
			of an honest sender
			{\ch or an onion
			that was injected
			into \FRN
			by \Ss
			}.
			By \enquote{stage},
			we mean that
			the next action
			to this onion
			is adversarial
			(i.e.,
			it is sent
			over a link
			or processed by
			an adversarial router).

		\item The $O$-list
			containing onions
			sent by corrupted senders
			together with
			the information
			about the communication
			(\textit{{\ch onionlist, path, currentPosition}, \textit{information})}.

		\item The $Reply$-list
			containing reply information
			together with
			the forward id
			for communications
			with a corrupted sender
			$({\ch sid}_{fwd}, \text{reply information})$.

		\item The $C$-list
			containing reply information
			together with
			a tid for communications
			with an honest sender
			$(P_i, reply, tid)$.
	\end{itemize}

	\textbf{\Ss's behavior
	on a message from \FRN:}
	\textit{In case
	the received output
	belongs to
	an adversarial sender's communication:}

	\noindent\textbf{Case I:}\hspace{0pt plus 20pt}
			\enquote{\textsl{start}
				belongs to {\ch reply}
				from \Ps
				with ${\ch sid}, P_r, m, n, {\ch \Pp, \Pp\ba}, d$,
				{\ch replying to onion
				from $P_r$ with $sid$}%
			};
			an honest relay
			is replying to
			an onion of
			a corrupted sender.
			\Ss knows that
			the next output
			\enquote{Onion $tid$
				in direction $d$
				from \ldots
			}
			includes the first part
			of this backward path,
			that he chose
			to consist of
			{\ch the correct sequence
				of honest relays
				potentially followed by
			}
			one adversarial relay
			{\ch $[\ldots]$}.
			To construct
			the right real world reply onion,
			\Ss looks up
			the reply information
			$({\ch sid}, \textit{replyinfo})$
			for this ${\ch sid}$
			in the $Reply$-list
			and uses the information
			to construct
			{\ch the reply onion
				\[(O_1, P_1) \gets \Reply(m, \textit{replyinfo}, {\ch \Ps, SK_s})\]
				followed by
				the next onion layers
				as far as \Ss
				can process them
				with the secret keys
				of the honest relays.
				Since the sender
				of the forward onion
				is corrupted,
				there must be
				at least one adversarial relay
				on the reply path
				of the onion,
				so \Ss will be able
				to process the onion
				up to the pair $(O', P')$
				with an adversarial $P'$.
				This results in
				a list of onions
				$\mathcal{O} = (O_1, \ldots, O_{last})$
				and a list of relays
				$\Pp = (P_1, \ldots, P_{last})$.
				\Pp is identical
				to the reply path
				of the onion
				that was already created
				in the ideal world.%
			}
			{\ch \Ss} sends $O_1$ to $P_1$,
			if $P_1$ is adversarial,
			or to \Aa's party
			representing the link between
			$P_{\ch s}$ and $P_1$,
			if $P_1$ is honest.
			(Note that $P_{\ch s}$
			cannot be adversarial
			for this output
			as then both sender and receiver
			would be corrupted,
			which only activates
			cases \textbf{VIIIb} and \textbf{II}
			(as it works
			without including
			any reply onion
			from the view of
			the ideal world).
			{\ch
				If the next relay $P_1$
				on the reply onion's path
				is honest,
				then \Ss needs to
				be able to associate
				that onion layer
				and the ones
				following it
				in the real world
				with the ideal-world onion
				as the layers are processed
				and sent along
				the honest relays.
				Conversely,
				if $P_1$
				is adversarial,
				then the onion
				leaves \Ss's control
				after this case%
				\footnote{%
					If $P_1$ is adversarial,
					the reply path
					of the ideal-world onion
					\Ss created
					for this reply
					in case \textbf{VIII}
					will also end
					at $P_1$.
				}.

			\let\origenumi\labelenumi
			\begin{enumerate}
				\item If the first relay $P_1$
					is an honest relay,
					\Ss adds
					the tuple
					$(\mathcal{O}, \Pp, 0, (\Ps, sid, P', m, \Pp, ()))$
					onto the $O$-list
					and the tuple
					$(O_1, \Ps, P_1, tid, a)$
					to the $r$-list,
					where $tid$ is the ID
					that \Ss received
					along with the output
					from \FRN
					and $a$ is the index
					of the $O$-list entry.%

				\item If $P_1$ is adversarial,
					\Ss does no additional work.%
			\end{enumerate}%
			\let\labelenumi\origenumi%
			}

		\noindent\textbf{Case II:}\hspace{0pt plus 20pt}
			\enquote{\textsl{start} belongs to
				onion from \Ps
				with ${\ch sid}, {\ch R}, m, {\ch \Pp, \Pp\ba, d}$%
			}.
			This is just the result
			of \Ss's reaction
			to an onion
			from \Aa
			that was not
			the protocol-conform processing
			of an honest sender's communication
			(Case \textbf{VIII}).
			\Ss does nothing.

		\noindent\textbf{Case IIIa:}
			any output
			together with
			\enquote{$tid$ belongs to
				onion{\ch /reply} from \Ps
				with ${\ch sid}, {\ch R/}P_r, m, {\ch \Pp, \Pp\ba, d}$%
			}
			for $tid \notin \{\textsl{start}, \textsl{end}\}$.
			This means
			an honest relay
			is done processing
			an onion
			received from the adversary \Aa
			that was not
			the protocol-conform processing
			of an honest sender's communication
			(processing that follows Case \textbf{VII}).
			\Ss finds
			(\textit{{\ch onionlist, path, $c := \textit{currentPosition}$}, information})
			with these inputs
			as information
			{\ch and
				$\Pp[][c] = \Po{i}$
				where \Po{i}
				is the relay
				that sent
				the onion $tid$
				in \FRN%
				~\cite{break_onion}
			}
			in the $O$-list
			(notice that there
			has to be
			such an entry).
			{\ch Let $a$
				be the index
				of the entry
				in the $O$-list%
				~\cite{break_onion}.

				\Ss must now
				send the correct onion
				from the list
				in the $O$-list entry
				over the next link
				while keeping track of it
				so that it can reassociate it
				with the $O$-list entry
				when it next receives it
				as the following honest relay.
				To this end,
				\Ss stores
				$(\mathcal{O}[c], \Pp[][c], \Pp[][c+1], tid, a)$
				to the $r$-list
				and sends
				$\mathcal{O}[c]$
				to the link to
				$\Pp[][c+1]$
				from $\Pp[][c]$%
				~\cite{break_onion}.
			}

		{\ch \noindent\textbf{Case IIIb:}
			any output
			together with
			\enquote{\textsl{end} belongs to
				onion{\ch /reply} from \Ps
				with ${\ch sid}, {\ch R/}P_r, m, {\ch \Pp, \Pp\ba, d}$%
			}.
			This case occurs
			whenever one of the onions
			\Ss creates in \FRN
			in case \textbf{VIII}
			reaches the end
			of its path
			and either an adversarial relay
			or the receiver of the onion
			comes next.
			\Ss can tell the difference
			by examining whether
			another relay
			$\Pp[][c+1]$
			and onion
			$\mathcal{O}[c]$
			remain in the lists
			of the $O$-list entry
			(\textit{onionlist, path, $c := \textit{currentPosition}$, information})
			corresponding to this onion
			(\Ss finds the entry
			like in case \textbf{IIIa}).

			\let\origenumi\labelenumi
			\begin{enumerate}
				\item If another relay follows,
					\Ss sends
					$\mathcal{O}[c]$
					to the link to
					$\Pp[][c+1]$
					from 
					$\Pp[][c]$%
					~\cite{break_onion}.

				\item If there is no next relay,
					then \Ss must
					send the message
					contained in the onion
					to the receiver
					along with
					the correct reply ID
					if the onion
					is repliable
					(if it is,
					\Ss also receives
					\enquote{Reply with reply ID $rid$}
					in the output).
					\Ss sends $(m, rid)$
					(or $(m, \bot)$
					if the onion
					is not repliable)
					to the link to $R$
					from \Po{i}
					in the real world,
					where \Po{i}
					is the relay
					that sent the message
					in \FRN,
					triggering this case.
			\end{enumerate}
			\let\labelenumi\origenumi%
		}
			{\ch \noindent\textbf{Case IIIc:}
				Any output
				together with
				the new output
				\enquote{\textsl{tagged} belongs to
					onion from \Ps
					with $sid, R, m, n, \Pp{}$%
				}.
				If \Ss receives
				this message,
				the final honest relay
				on the forward path
				of a corrupted sender's onion
				just processed
				a tagged onion
				over either
				the final path segment
				consisting of
				only corrupted relays
				or the final link
				to the receiver itself.
				Depending on
				which is the case,
				\Ss performs different actions:

				\let\origenumi\labelenumi
				\begin{enumerate}
					\item If $\Pp[][c+1]$
						is set:
						The onion
						is not
						at the exit relay yet
						and the tagging
						will not
						be discovered
						by an honest router.
						\Ss behaves
						like in case \textbf{IIIa}.
					\item If $\Pp[][c+1]$
						is not set:
						The honest exit relay
						has noticed the tagging.
						The protocol would discard
						such an onion,
						so no action
						is required from \Ss.
				\end{enumerate}
				\let\labelenumi\orig
			}

	\textit{In case
		the received output
		belongs to
		an honest sender's communication:
	}

		\noindent\textbf{Case IV:}
			\enquote{\ch
				\Po{i} sends
				onion $tid$
				to \Po{i+1}
				via $()$%
			}.
			In this case,
			\Ss needs to
			make it look
			as though an onion was passed
			from the honest party \Po{i}
			to the honest party \Po{i+1}:
			\Ss picks the path
			$\Pp = (\Po{i}, \Po{i+1})$
			and random message $m_{rdm}$.
			\Ss honestly picks
			a randomness \Rr{}
			{\ch and a random receiver $R$}
			and calculates
			\[O_1 \gets {\ch \FormO(1, \Rr, m_{rdm}, R, \Pp, (), PK_{\Pp}, ())}\]
			and sends
			the onion $O_1$
			to \Aa's party
			representing the link
			between the honest relays
			as if it was sent
			from \Po{i} to \Po{i+1}.
			\Ss stores
			$({\ch O_1}, {\ch \Po{i}}, \Po{i+1}, tid, {\ch \bot})$
			on the $r$-list.

		\noindent\textbf{Case V:}\hspace{0pt plus 20pt}
			\enquote{\ch
				\Po{i} sends
				onion $tid$
				to \Po{j}
				via \Path{i+1}{j-1}{}%
			}.
			To handle this case,
			\Ss picks the path
			$\Pp = \Path{i+1}{j-1}$,
			a randomness \Rr{}
			{\ch and a random receiver $R$}
			and a message $m_{rdm}$
			and calculates
			\[O_1 \gets {\ch \FormO(1, \Rr, m_{rdm}, R, \Pp, (), PK_{\Pp}, ())}\]
			and sends
			the onion $O_1$
			to \Po{i+1},
			as if it came
			from \Po{i}.
			\Ss stores
			$({\ch O_{j-i-1}}, {\ch \Po{j-1}}, \Po{j}, tid, {\ch \bot})$
			on the $r$-list.

		\noindent\textbf{Case VIa:}
			\Ss receives the message
			\enquote{\ch
				\Po{i} sends onion
				with message $m$
				to $R$
				via \Path{i+1}{j-1}{}%
			}.
			{\ch The behavior
				in this case
				depends on
				whether
				the onion is repliable
				and whether \Po{i}
				is the onion's exit relay
				or not:
			}
			\let\origenumi\labelenumi
			\begin{enumerate}
				{\ch \item The onion is not repliable
					and \Po{i} is
					its exit relay.
					In this case,
					\Ss sends $(m, \bot)$
					to the adversary's link
					to $R$ as \Po{i}
					in the real world.%
				}

				{\ch \item The onion is not repliable
					and \Po{i} is not
					its exit relay.
					Here,
					\Ss needs to build an onion
					that will carry the message
					across the remaining
					adversarial relays
					to the receiver.
					This happens just like
					in the original case:
				}
				\Ss picks the path
				$\Pp = \Path{\ch i+1}{\ch n}$,
				randomness \Rr,
				calculates
				\[O_1 \gets {\ch \FormO(1, \Rr, m, R, \Pp, (), PK_{\Pp}, ())}\]
				and sends
				the onion $O_1$
				to \Po{i+1},
				as if it came from \Po{i}.

				{\ch \item The onion is repliable
					and \Po{i} is
					the exit relay.
					\Ss receives
					the additional output
					\enquote{%
						Reply with reply ID $rid$%
					}
					and sends $(m, rid)$
					to the adversary's link
					to $R$ as \Po{i}
					in the real world.
				}

				{\ch \item The onion is repliable
					and \Po{i} is
					not the exit relay.
					Now,
					\Ss needs to use
					the extra output
					\enquote{%
						Reply with $tid$.
						Its reply path
						begins with $\Pp\ba$%
					}
					to construct an onion
					that will carry the message
					to the receiver
					and allow it to reply
					such that the reply onion
					will follow the beginning
					of the reply path
					to the first honest relay
					on it,
					where \Ss will expect it.
				}
				\Ss picks the path
				$\Pp = \Path{\ch i+1}{\ch n}$,
				randomness \Rr,
				calculates
				\[O_1\! \gets\! {\ch \FormO(1,\! \Rr,\! m,\! R,\! \Pp,\! \Pp\ba\!\!\!, PK_{\Pp\!},\! PK_{\Pp\ba\!\!})}\]
				and sends
				the onion $O_1$
				to the relay \Po{i+1},
				as if it was
				sent from
				the relay \Po{i}.
				Further,
				\Ss stores
				$(\Pp\ba.last, info, tid)$
				with
				$\textit{info} = ({\ch n} + \Pp\ba\!\!.lastPos, \Rr, m, {\ch R}, \Pp, \Pp\ba, PK_{\Pp}, PK_{\Pp\ba})$
				on the $C$-list.
				(Note that,
				as this is an honest communication,
				$\Pp\ba.last$ is honest.)
			\end{enumerate}
			\let\labelenumi\origenumi
			
			{\ch \noindent\textbf{Case VIb:}
				\Ss receives
				the message
				\enquote{\Po{i} sends tagged onion
				via \Path{i+1}{n}{}}.
				This means that
				an honest relay
				has processed
				a tagged onion
				from an honest sender
				on the forward path
				and is delivering it
				to the exit relay.
				\Path{i+1}{n}
				is never empty
				when \Ss receives
				this input
				since \FRN guards
				against that case
				with the condition
				$i < n$.
				To translate
				this onion
				into the real world,
				\Ss behaves
				like in case \textbf{VIa2},
				but doesn't learn
				the message
				or the receiver
				of the onion
				and must additionally
				tag the new onion
				before sending it.
				To this end,
				\Ss chooses
				a random $m_{rdm} \in M$
				and $R_{rdm} \in D$
				and builds the onion
				using those.
				Before sending it,
				\Ss tags the onion.
			}

			{\ch \noindent\textbf{Case IX:}
				\enquote{$R$ replies
					to $rid$
					with message $m$
					via $P_i$%
				}.
				\Ss needs to
				recreate this reply
				in the real world
				by sending $(m, rid)$
				to the adversary's link
				to $P_i$
				from the honest receiver $R$.%
			}

	\textbf{\Ss's behavior on a message from \Aa:}
	\Ss,
	as real world honest party $P_i$,
	received an onion
	$O = (\tilde{\eta}, \tilde{\delta})$
	{\ch or a message-reply ID-pair $(m, rid)$}
	from \Aa
	as adversarial player $P_a$.

		\noindent\textbf{Case VIIa:}
			$({\ch (\tilde{\eta}, \tilde{\delta})}, {\ch P_{i-1}}, P_i, tid, {\ch a})$
			is on the $r$-list
			for some $tid$.
			In this case,
			$O$ is the protocol-conform processing
			of an onion
			from an honest sender's communication.
			\Ss runs
			$\Proc(SK_{P_i}, O, P_i)$.
			If it returns a fail
			($O$ is a replay
			or modification
			that is detected and dropped
			by $P_i$),
			\Ss does nothing.
			Otherwise,
			\Ss sends the message
			${\ch \DelOnion(tid)}$
			to {\ch \FRN}
			{\ch and increments
				\textit{currentPosition}
				for the $a$-th entry
				in the $O$-list
				if $a \neq \bot$%
				~\cite{break_onion}.
			}

		{\ch \noindent\textbf{Case VIIb:}
			$({\ch (\tilde{\eta}, \delta')}, {\ch P_{i-1}}, P_i, tid, {\ch a})$
			is on the $r$-list
			for some $tid$
			and a $\delta' \neq \tilde{\delta}$.
			From this,
			we know that
			\Aa has tagged
			the onion
			in flight.
			\Ss calculates
			$\Proc(SK_{P_i}, O, P_i)$.
			If it returns a fail
			(e.g., $O$ is a replay
			that is detected
			and dropped by the protocol),
			\Ss does nothing.
			If \Proc
			does not
			return a fail,
			\Ss calls
			$\Tag(tid)$
			to tag the onion
			in the ideal world
			as well
			before calling
			$\DelOnion(tid)$.
			If $a \neq \bot$,
			\Ss increases the position
			of the $a$-th entry
			in the $O$-list.
			To forward the tag
			in the real world,
			\Ss also replaces
			the onion list
			$\mathcal{O} = (O_1, \ldots, O, \ldots, O_k)$
			in the $O$-list entry
			with
			the new onion list
			$\mathcal{O}' = (O_1, \ldots, O', \ldots, O'_k)$,
			where $O'_i$
			is the result
			of processing $O'$
			repeatedly
			like in case \textbf{VIII}.%
		}%

		\noindent\textbf{Case VIII:}
			$(\tilde{\eta}, {\ch P_{i-1}}, P_i, tid, {\ch a})$
			is not on the $r$-list
			for any $tid$.
			{\ch This onion
				must have been sent
				by the $P_a$ relay itself
				since the links
				between relays
				are secure channels
				due to \Ff[SC].
				In order to
				replicate this onion
				in the ideal world,
				\Ss must first process it
				until it cannot be processed
				any further
				because either:
				1) the next relay
				is a receiver,
				in which case
				the reply path
				must also be processed,
				2) there is no next relay
				because the onion
				is a reply,
				3) the next relay
				is adversarial,
				so \Ss does not
				know the necessary keys,
				4) or processing the onion fails~\cite{break_onion}.

				In any case,
				\Ss processes $O$ repeatedly
				until it has
				the final result $(O', P')$
				along with
				the list of onions
				$\mathcal{O} = (O_1, \ldots, O_{last})$
				and relays $\Pp = (P_i, \ldots, P_{last})$
				encountered along the way.
				The following behavior
				depends on
				what form $(O', P')$ takes~\cite{break_onion}:
			}
			\let\origenumi\labelenumi
			\begin{enumerate}
				{\ch \item $(O', P') = (m, R)$:
					First,
					\Ss checks whether
					$O_{last}$ is repliable.
					If so,
					then \Ss forms a reply
					\[(O\ba_1, P\ba_1) \gets \Reply(m', O_{last}, P_i, SK_i)\]
					to it
					with an arbitrary message $m'$
					and processes it
					until it cannot proceed,
					resulting in
					a list of reply onions
					$\mathcal{O}\ba = (O\ba_1, \ldots, O\ba_k)$
					and relays
					$\Pp\ba = (P\ba_1, \ldots, P\ba_k)$.
					On the other hand,
					if the onion
					is not repliable,
					let $\mathcal{O}\ba = \Pp\ba = ()$.

					Now,
					\Ss creates the ideal-world onion
					by sending the message
					$\NewOnion(R, m, \Pp, \Pp\ba\!)$
					to \FRN
					in the role of $P_a$.
					After doing so,
					\Ss immediately delivers
					the first onion
					by calling $\DelOnion(tid)$
					with the $tid$ ID
					it received from \FRN
					without running case \textbf{IIIa}.
					It does so because
					\Aa has already delivered
					the onion from $P_a$
					to $P_i$.
					\Ss also stores
					$(\mathcal{O}, \Pp, 0, (P_a, sid, R, m, \Pp, \Pp\ba))$
					on the $O$-list~\cite{break_onion}
					and $(sid, O_{last})$
					on the $Reply$-list.%
				}

				\item ${\ch (O', P') = (m, \bot)}$:
					{\ch $[\ldots]$}
					$P_i$ is the recipient
					and $O'$
					{\ch is a message.}
					This means
					the adversary
					possibly replied
					to an honest senders
					forward onion
					{\ch with a corrupted exit relay}.
					\Ss checks for all
					$(P_i, reply, tid)$ tuples
					in the $C$-list
					to see if
					$\tilde{\eta}$ matches
					any $reply$-info
					on this list.
					If so
					(it was a reply to $tid$),
					\Ss sends the message
					${\ch \ByReply(m, tid)}$
					to \FRN on $P_a$'s behalf
					and,
					as \Aa already
					delivered this message
					to the honest party,
					sends
					${\ch \DelOnion(tid')}$
					for the belonging $tid'$.
					Otherwise,
					{\ch this onion is
						an unsolicited reply
						to an honest sender
						and will be ignored.
					}
					\Ss (creating this onion in the \FRN)
					sends {\ch the invalid onion
						on behalf of $P_a$:%
					}
					${\ch \NewOnion(\bot, \bot, \Pp, ())}$
					and ${\ch \DelOnion(tid)}$
					for the corresponding $tid$
					without running case \textbf{IIIa}.
					(Notice that \Ss knows
					which $tid$ and ${\ch sid}$
					belongs to this communication
					as it is started
					at an adversarial party $P_a$).
					{\ch \Ss adds
						$(\mathcal{O}, \Pp, 0, (P_a, sid, \bot, \bot, \Pp, ()))$
						to the $O$-list.
					}

				\item ${\ch (O', P' \neq \bot)}$:
					{\ch $P'$ is the next adversarial relay
						and $O'$ is the onion
						it should receive.
					}
					\Ss picks a message
					$m \in \mathcal{M}$
					{\ch and a new random receiver $R$}.
					\Ss sends
					on $P_a$'s behalf
					the message,
					${\ch \NewOnion(R, m, \Pp, ())}$
					(notice that
					this onion cannot be replied to)
					{\ch $[\ldots]$}
					and ${\ch \DelOnion(tid)}$
					for the belonging $tid$
					to \FRN{}
					{\ch without running case \textbf{IIIa}}
					(notice that \Ss knows
					the $tid$ as in case (a)).
					Now, \Ss adds the entry
					${\ch (\mathcal{O}\|(O'), \Pp\|(P'), 0, (P_a, sid, R, m, \Pp, ()))}$
					to the $O$-list.

				{\ch \item $(O', P') = (\bot, \bot)$:
					This onion failed to be processed
					at $P_{last}$,
					so \Ss must also send
					an invalid onion
					that takes this path
					in \FRN:
					\Ss sends
					$\NewOnion(\bot, \bot, \Pp, ())$
					in the role of $P_a$.
					If the header
					of the onion
					processes correctly
					as a last layer header,
					but the payload does not,
					\Ss calls $\Tag(tid)$
					on the corresponding
					$tid$ ID.
					Then,
					\Ss follows it with
					$\DelOnion(tid)$
					without running case \textbf{IIIa}.
					\Ss adds
					$(\mathcal{O}, \Pp, 0, (P_a, sid, \bot, \bot, \Pp, ()))$
					to the $O$-list.
				}
			\end{enumerate}
			\let\labelenumi\origenumi

			{\ch \noindent\textbf{Case X:}
				\Ss receives $(m, rid)$
				as $R$
				from the link to
				the exit relay $E$.
				\Ss sends the message
				$\DelPlain(E, m, rid, R)$
				to \FRN.%
			}

			{\ch \noindent\textbf{Case XI:}
				\Ss receives $(m, rid)$
				as the honest relay $P_i$
				from the link to
				the receiver $R$.
				\Ss sends the message
				$\DelReply(R, P_i, m, rid)$
				to \FRN.%
			}

	\noindent\textbf{Indistinguishability:}

	\textbf{Notation:}
	\Hh[_i] describes
	the first hybrid
	that replaces
	a certain part of any communication
	for the first communication.
	In \Hh[_i^{<x}]
	this part of the communication
	is replaced
	for the first $x - 1$ communications.
	Finally in \Hh[_i^\ast]
	this part of the communication
	is replaced in all communications.

	\textbf{Hybrid \Hh[_0]:}
	This machine sets up
	the keys for the honest parties
	(so it has their secret keys).
	Then it interacts
	with the environment
	and \Aa
	on behalf of the honest parties.
	It invokes
	the real protocol
	for the honest parties
	in interacting with \Aa.

	\textbf{Replacing between honest - Forward Onion:}
	We replace the onion layers
	in the way they appear
	in the communication.
	So the first onion layers
	(close to the sender)
	are replaced first.

	\textbf{Hybrid \Hh[_1]:}
	In this hybrid,
	for the first one forward communication
	the onion layers
	from its honest sender
	to the next honest relay
	on the forward path
	(relay or receiver)
	are replaced with
	random onion layers
	embedding the same path.
	More precisely,
	this machine acts like \Hh[_0]
	except that the consecutive onion layers
	$O_1, O_2, \ldots, O_j$
	from an honest sender $P_0$
	to the next honest relay $P_j$
	are replaced with $\bar{O}_1$
	and its following processings
	by calculating
	(with honestly chosen
	randomness $\Rr'$
	{\ch and a random receiver $R_{rdm}$})
	\[\bar{O}_1 \gets {\ch \FormO(1, \Rr', m_{rdm}, R_{rdm}, \Pp^{\ch \prime}, (), PK_{\Pp^{\ch \prime}}, ())}\]
	where $m_{rdm}$ is a random message,
	$\Pp^{\ch \prime} = (P_1, \ldots, P_j)$.
	\Hh[_1] now maintains a new $\bar{O}$-list,
	which it uses to
	recognize replacement onions.
	It stores
	$(info = (\Rr', m_{rdm}, {\ch R_{rdm}}, \Pp^{\ch \prime}\!, (), PK_{\Pp}, ()), {\ch \delta_j}, P_j, (O_1^R, P_{j+1}))$\\
	there,
	where $info$ are the randomness and parameters
	used for the {\ch replacement} onion's creation,
	{\ch $\delta_j$ is the payload
	of the $j$-the onion layer,}
	and $O_1^R$ is calculated as
	\[O_1^R\! \gets\! {\ch \FormO(j + 1, \Rr, m, R, \Pp, \Pp\ba\!\!, PK_{\Pp}, PK_{\Pp\ba\!\!})},\]
	where the randomness,
	{\ch receiver,}
	paths and message
	are chosen as in
	the original sender's call in \Hh[_0].
	If an onion $\tilde{O}$
	is sent to $P_j$,
	the machine tests
	if processing results in a fail
	(replay/modification detected and dropped).
	If it does not,
	\Hh[_1] uses
	\[\RO(j, \tilde{O}, \Rr', m_{rdm}, {\ch R_{rdm}}, \Pp^{\ch \prime}, \Pp\ba, PK_{\Pp}, PK_{\Pp\ba})\]
	for every recognize-information
	stored in the $\bar{O}$-list
	where the {\ch third} entry is $P_j$.
	If it finds a match,
	{\ch it compares $\tilde{O}$'s payload
	to the $\delta_j$
	in the entry.
	If those also match,}
	the belonging $O_1^R$
	is sent to $P_{j+1}$
	as the processing result of $P_j$.
	{\ch If only the payload comparison fails,
	the onion has been tagged.
	\Hh[_1] adds
	the original onion's input parameters
	to a \textit{Tag}-list
	for later recognition.
	If $P_j$ is
	the onion's exit relay
	($P_{j+1} = \bot$),
	\Hh[_1] produces
	no output.
	Otherwise,
	\Hh[_1] recreates the tag
	on $O_1^R$
	and sends it
	to $P_{j+1}$.
	If no onion
	with a matching header
	is found,
	}
	$\Proc(SK_{P_j}, \tilde{O}, P_j)$
	is used.

	\underline{$\Hh[_0] \approx_I \Hh[_1]$:}
	The environment gets notified
	when an honest party receives
	an onion layer
	(and about their repliability)
	and inputs
	when this party is done.
	As we just
	exchange onion layers with others
	(with the same repliability),
	the behavior to the environment
	is indistinguishable for both machines.
	\Aa observes
	the onion layers after $P_0$
	and,
	if it sends an onion to $P_j$,
	the result of the processing
	after the honest relay.
	Depending on the behavior of \Aa,
	three cases occur:
	\Aa drops the onion
	belonging to this communication
	before $P_j$,
	\Aa behaves protocol-conform
	and sends the expected onion to $P_j$
	or \Aa modifies the expected onion
	before sending it to $P_j$.
	Notice that dropping the onion
	leaves the adversary
	with no further output.
	Thus,
	we can focus
	on the other cases:

	We assume there exists
	a distinguisher \Dd
	between \Hh[_0] and \Hh[_1]
	and construct a successful attack
	on {\ch \TFLU}.

	The attack receives
	key and name
	of the honest relay
	and uses the input
	of the replaced communication
	as choice for the challenge,
	where it replaces the name
	of the first honest relay
	with the one
	that it got
	from the challenger.
	For the other relays,
	the attack decides on the keys
	as \Aa (for corrupted)
	and the protocol (for honest) do.
	It receives $\tilde{O}$
	from the challenger.
	The attack uses \Dd.
	For \Dd,
	it simulates
	all communications except the one
	chosen for the challenge,
	with the oracles and knowledge
	of the protocol and keys.
	For simulating the challenge communication,
	the attack hands $\tilde{O}$ to \Aa
	as soon as \Dd instructs to do so.
	To simulate further for \Dd
	it uses $\tilde{O}$
	to calculate the later layers
	and does any actions \Aa does
	on the onion.

	\Aa either sends
	the honest processing of $\tilde{O}$
	to the challenge router
	or \Aa modifies it.
	The attack uses the oracle
	to simulate the further processing
	of $\tilde{O}$
	or its modification.
	{\ch If \Aa chooses to
		tag $\tilde{O}$,
		then the challenger
		will output
		an onion with
		a random payload
		in both \Hh[_0]
		and \Hh[_1]
		since the tag
		completely randomizes
		the payload contents.
	}

	Thus,
	either the challenger chose $b = 0$
	and the attack behaves like
	\Hh[_0] under \Dd;
	or the challenger chose $b = 1$
	and the attack behaves like
	\Hh[_1] under \Dd.
	The attack outputs
	the same bit as \Dd does
	for its simulation
	to win with the same advantage
	as \Dd can distinguish the hybrids.

	\textbf{Hybrid \Hh[_1^{<x}]:}
	In this hybrid,
	for the first $x - 1$ forward communications,
	onion layers
	from an honest sender
	to the next honest relay
	on the forward path
	are replaced with
	a random onion
	sharing this path.
	[Note that $\Hh[_1] = \Hh[_1^{<2}]$
	and let $\Hh[_1^\ast]$
	be the hybrid where
	the replacement happened
	for all communications.]

	\underline{$\Hh[_1^{<x-1}] \approx_I \Hh[_1^{<x}]$:}
	Analogous to above.
	Apply argumentation of indistinguishability
	($\Hh[_0] \approx_I \Hh[_1]$)
	for every replaced subpath.

	{\ch \textbf{Hybrid \Hh[_{2a}]:}
	In the two hybrids
	\Hh[_{2a}] and \Hh[_{2b}],}
	for the first forward communication
	for which, in the adversarial processing,
	no recognition-falsifying modification
	(i.e. on $\eta$) occurred
	and other modification
	does not result in a fail,
	onion layers between
	two consecutive honest relays
	on the forward path
	(the second might be the {\ch exit relay})
	are replaced with
	random onion layers
	embedding the same path.
	{\ch We do so
		in two hybrid steps
		because we require
		both the \NTI
		and \TFLU properties
		to truncate the forward path
		of the onion
		before the first
		of the two honest relays
		and then replace
		the onion layers
		between the honest relays.
	}
	Additionally,
	for all forward communications,
	replacements between
	the sender and the first relay
	happen as in $\Hh[_1^\ast]$.
	More precisely,
	{\ch \Hh[_{2a}]} acts like $\Hh[_1^\ast]$
	except for the processing
	of $O_j$.
	{\ch The onion layers
	$O_{j+1}$, \ldots, $O_{n}$, $O\ba_{1}$, \ldots, $O\ba_{n\ba}$
	are replaced with
	$\bar{O}_1$, \ldots, $\bar{O}_{n-j}$, $\bar{O}\ba_{1}$, \ldots, $\bar{O}\ba_{n\ba}$;
	the hybrid sends
	$\bar{O}_1$ instead of $O_{j+1}$.
	The replacement
	is formed as
	\[\bar{O}_1 \gets \FormO(1, \Rr', m, R, \Pp', \Pp\ba\!\!, PK_{\Pp'}, PK_{\Pp\ba\!\!})\]
	with an honestly chosen
	randomness $\Rr'$
	and $\Pp' = (P_{j+1}, \ldots, P_{n})$.
	If the onion's information
	is on the \textit{Tag}-list,
	the tag is recreated
	on $\bar{O}_1$
	before it is sent.
	}

	{\ch \underline{$\Hh[_1^\ast] \approx_I \Hh[_{2a}]$:}
	\Hh[_{2a}] replaces
	the onion layers
	on the path after
	an honest relay
	and does so
	for an onion
	that has already had
	all of its subpaths
	between honest relays
	prior to this
	honest relay
	replaced before.
	The original onion layers
	before this honest relay
	are thus
	never output
	to the adversary
	while the layers
	used as replacements
	are chosen independently
	at random.
	We can reduce
	the new replacement
	to \NTI
	with $j_{\NTI} = j$
	and $j\ba_{\NTI} = n\ba$
	as the positions
	of the honest relays
	in the challenge
	thanks to this.
	Since the layers
	before $P_j$
	are independent
	of the challenge,
	the attack
	can recognize tagged payloads
	on those layers
	and recreate the tag
	on the challenge onion
	from the \NTI challenger
	if necessary.
	}

	\textbf{Hybrid \Hh[_{2b}]:}
	{\ch In this hybrid,
	we perform
	the actual replacement
	of the onion layers
	between the honest relays
	for the onion
	whose forward path
	was truncated
	in \Hh[_{2a}].
	}
	In essence,
	the consecutive onion layers
	{\ch $\bar{O}_{1}, \ldots, \bar{O}_{j'-j}$}
	from a communication
	of an honest sender,
	starting at the next honest relay $P_j$
	to the next following honest relay $P_{j'}$,
	are replaced with
	{\ch $\hat{O}_1, \ldots, \hat{O}_{j'-j}$}
	by sending {\ch $\hat{O}_1$}.
	Thereby,
	for honestly chosen
	randomness $\Rr'$
	{\ch and a random receiver $R_{rdm}$}:
	\[{\ch \hat{O}_1}\! \gets\! {\ch \FormO(1, \Rr', m_{rdm}, R_{rdm}, \Pp^{\ch \prime\prime}, (), PK_{\Pp^{\ch \prime\prime}}, ())},\]
	where $m_{rdm}$ is a random message,
	$\Pp^{\ch \prime\prime} = (P_{\ch j+1}, \ldots, P_{j'})$
	is the path between
	the honest relays.
	\Hh[_{2b}] stores
	\begin{flalign*}
		(({\ch \Rr''}, m_{rdm}, {\ch R_{rdm}}, \Pp^{\ch \prime\prime}, (), PK_{\Pp^{\ch \prime\prime}}, ()), \\
		{\ch \delta_{j'-j}}, P_{j'}, (O_1^R, P_{j'+1})),
	\end{flalign*}
	where {\ch $\bar{O}_1^R$}
	is calculated with
	\[{\ch \FormO(j'\!-\!j\! + \!1,\! \Rr',\! m,\! {\ch R},\! \Pp^{\ch \prime},\! \Pp\ba\!\!\!, PK_{\Pp^{\ch \prime}},\! PK_{\Pp\ba\!\!})},\]
	where the randomness,
	{\ch receiver},
	paths and message are chosen
	as {\ch \Hh[_{2a}] chose them},
	on the $\bar{O}$-list.
	Like in \Hh[_1^\ast],
	if an onion $\tilde{O}$
	is sent to $P_{j'}$,
	processing is first checked for a fail.
	If it does not fail,
	\Hh[_{2b}] checks
	\[{\ch \RO(j'\! -\! j, \tilde{O}, \Rr'', m_{rdm}, R_{rdm}, \Pp'', (), PK_{\Pp''}, ())}\]
	for any info on the $\bar{O}$-list
	where the second entry is $P_{j'}$.
	If it finds a match,
	{\ch it checks
	whether the original onion's information
	is on the \textit{Tag}-list
	and compares
	$\tilde{O}$'s payload
	to $\delta_{j'-j}$.
	If the information
	is on the list
	or the payloads
	do not match,
	one of the previous replacements
	were tagged
	or the current replacement
	was tagged.
	If the original onion's information
	is not on the \textit{Tag}-list yet,
	it is added.
	If $P_j$ is
	the onion's exit relay,
	\Hh[_{2b}] does nothing.
	Otherwise,
	\Hh[_{2b}] recreates the tag
	on $O_1^R$
	and sends it
	to $P_{j'+1}$.
	On the other hand,
	if the onion
	was not tagged until now,}
	the belonging $O_1^R$
	is used as the processing result of $P_{j'}$.
	{\ch If no onion
	with a matching header
	is found,}
	$\Proc(SK_{P_{j'}}, \tilde{O}, P_{j'})$
	is used.

	{\ch \underline{$\Hh[_{2a}] \approx_I \Hh[_{2b}]$:}}
	\Hh[_{2b}] replaces,
	for one communication
	(and all its replays),
	the first subpath
	between two consecutive honest relays
	after an honest sender.
	The output to \Aa
	includes the earlier
	(by \Hh[_1^\ast]) replaced onion layers $\bar{O}_{earlier}$
	before the first honest relay
	(these layers are identical
	in {\ch \Hh[_{2a}]} and \Hh[_{2b}])
	that take the original subpath
	but are otherwise chosen randomly;
	the original onion layers
	after the first honest relay
	for all communications
	not considered by \Hh[_{2b}]
	(output by \Hh[_1^\ast])
	or in case of the communication
	considered by \Hh[_{2b}],
	the newly drawn random replacement
	(generated by \Hh[_{2b}]);
	and the processing after $P_{j'}$.

	{\ch Similarly to
		our argument
		for \Hh[_{2a}],
		the $\bar{O}_{earlier}$ layers
		are random
		and independent of
		the replaced layers,
		so they can be built
		without needing
		the \TFLU challenger.}

	Thus,
	all that is left
	are the original/replaced onion layer
	after the first honest relay
	and the processing afterwards.
	This is the same output
	as in $\Hh[_0] \approx_I \Hh[_1]$.
	Hence,
	if there exists a distinguisher
	between {\ch \Hh[_{2a}]} and \Hh[_{2b}]
	there exists an attack on {\ch \TFLU}.

	\textbf{Counting explanation for $\Hh[_2^{<x}]$:}
	{\ch From here,
		we refer to
		the combination of
		\Hh[_{2a}] and \Hh[_{2b}]
		as \Hh[_2]
		for convenience.
	}
	Communication paths consist
	of multiple possible honest subpaths
	(paths from an honest relay
	to the next honest relay).
	We count (and replace)
	all these subpaths
	from the subpath closest to the sender
	until the one closest to the receiver.
	We first replace
	all such subpaths for the first communication,
	then for the second and so on.
	Below we use $< x$
	to signal how many such subpaths
	will be replaced
	in the current hybrid.
	[Note that $\Hh[_2] = \Hh[_2^{<2}]$
	and let \Hh[_2^\ast] be the hybrid
	where the replacement happened
	for all such subpaths.]

	\textbf{Hybrid \Hh[_2^{<x}]:}
	In this hybrid,
	the first $x-1$ honest subpaths
	(honest relay to next honest relay)
	of honest senders' forward communications
	is replaced with a random onion
	sharing the path.
	Additionally,
	for all forward communications,
	replacements between the sender
	and the first relay happen
	as in \Hh[_1^\ast].
	If \Aa previously
	(i.e., in onion layers
	up to the honest relay
	starting the selected subpath)
	modified $\eta$ of an onion layer
	in this communication
	or modifies other parts
	such that processing fails,
	the communication is skipped.

	\underline{$\Hh[_2^{<x-1}] \approx_I \Hh[_2^{<x}]$:}
	Analogous to above.

	\textbf{Replacing between Honest - Backward Onion}

	On the backward path,
	we replace the last onion layers first,
	then the second last and so on.
	Each machine only starts replacing
	at a certain point
	and if a message does not come that far
	(it is modified or dropped),
	they simply do not use any replacement.
	For all following hybrids,
	the replacements on the forward path
	are done as in \Hh[_2^\ast].

	\textbf{Hybrid \Hh[_1\ba]:}
	Similar to \Hh[_1],
	but this time one backward communication
	between the last honest relay
	({\ch which could be the exit relay})
	until the honest (forwards) sender
	is replaced.
	More precisely,
	this machine acts like \Hh[_2^\ast]
	except that the consecutive onion layers
	$O\ba_{j+1}, \ldots, O\ba_{\ch n\ba}$
	from a reply to an honest (forward) sender
	from the last honest relay $P\ba_j$
	to the (forward) sender $P\ba_{\ch n\ba} = P_0$
	are replaced with
	$\bar{O}_1, \ldots, \bar{O}_{\ch n\ba-j}$
	with (for an honestly chosen $\Rr'$
	{\ch and a random receiver $R_{rdm}$}):
	\[\bar{O}_1 \gets {\ch \FormO(1, \Rr', m_{rdm}, R_{rdm}, \Pp^{\ch \prime}, (), PK_{\Pp^{\ch \prime}}, ())}\]
	where $m_{rdm}$ is a random message,
	$\Pp^{\ch \prime} = (P\ba_{\ch j+1}, \ldots, P\ba_{\ch n\ba})$
	is the path from $P\ba_{\ch j+1}$ to $P\ba_{\ch n\ba}$.
	\Hh[_1\ba] stores
	$(\textit{info}, P\ba_{\ch n\ba} = P_0, m_{rdm})$
	on the $\bar{O}$-list.
	When looking up entries
	(with \RO)
	on the $\bar{O}$-list,
	\Hh[_1\ba] checks
	the belonging last entry
	to be an onion
	before sending it to the next relay.

	\underline{$\Hh[_2^\ast] \approx_I \Hh[_1\ba]$:}
	The environment is notified
	when an honest party receives an onion layer
	and inputs when this party is done.
	As we just exchange onion layers with others
	(with the same repliability),
	the behavior to the environment
	is indistinguishable for both machines.

	\Aa observes the onion layers
	before $P_j\ba$
	and,
	if it sends an onion to $P\ba_{\ch n\ba}$,
	the result of the processing
	after the honest relay.
	Depending on the behavior of \Aa,
	three cases occur:
	\Aa drops the onion
	belonging to this communication
	before $P\ba_{\ch n\ba}$,
	\Aa behaves protocol-conform
	and sends the expected onion
	to $P\ba_{\ch n\ba}$
	or \Aa modifies the expected onion
	before sending it to $P\ba_{\ch n\ba}$.
	Notice that dropping the onion
	leaves the adversary
	with no further output.
	Thus,
	we can focus on the other cases.

	We assume there exists a distinguisher \Dd
	between \Hh[_2^\ast] and \Hh[_1\ba]
	and construct a successful attack
	on {\ch \NBLU}:

	The attack receives
	key and name of the honest relay
	and uses the input of the replaced communication
	as the choice for the challenge,
	where it replaces the name of the honest relay
	with the one that it got from the challenger.
	For the other relays,
	the attack decides on the keys
	as \Aa (for corrupted)
	and the protocol (for honest) do.
	It receives $O_1$ from the challenger
	and forwards it to \Aa
	for the corrupted first relay
	(on the forward path).
	The attack simulates all other communications
	with oracles
	(or their replacements
	as in the games before)
	and at some point,
	as \Aa replies to $O_1$
	(after receiving its processing $O_{n+1}$),
	so does our attack.
	The reply is processed
	(with the knowledge of the keys)
	until the honest relay,
	where the replaced onion layers start
	and this processed reply is forwarded
	to the oracle of the challenger
	as $O$ to process it.
	The challenger returns $\tilde{O}$.
	The attack sends $\tilde{O}$,
	as the processing of the answer,
	to \Aa as soon as
	\Dd instructs to do so.
	To simulate further for \Dd
	it uses $\tilde{O}$
	to calculate the later layers
	and does any actions \Aa does
	on the onion.
	Further, the attack simulates
	all other communications
	with the oracles and knowledge
	of the protocol and keys
	(or the random replacement onions,
	if replaced before).

	Thus,
	either the challenger chose $b = 0$
	and the attack behaves like
	\Hh[_2^\ast] under \Dd;
	or the challenger chose $b = 1$
	and the attack behaves like
	\Hh[_1\ba] under \Dd.
	The attack outputs the same bit
	as \Dd does for its simulation
	to win with the same advantage
	as \Dd can distinguish the hybrids.

	\textbf{Hybrid \Hh[_1^{<x\leftarrow}]:}
	In this hybrid,
	for the first $x-1$ backward communications,
	onion layers from the last honest relay
	to the honest sender
	(=backwards receiver)
	are replaced with a random onion
	sharing this path.
	The replacement is again stored
	on the $\bar{O}$-list as before.

	\underline{$\Hh[_1^{<x-1\leftarrow}] \approx_I \Hh[_1^{<x\leftarrow}]$:}
	Analogous to above.
	Apply argumentation of indistinguishability
	($\Hh[_2^\ast] \approx_I \Hh[_1\ba]$)
	for every replaced subpath.

	{\ch \textbf{Hybrid \Hh[_{2a}\ba]:}
	In the hybrids
	\Hh[_{2a}\ba] and \Hh[_{2b}\ba],}
	for the first backward communication
	(and all its replays)
	for which,
	in the adversarial processing,
	no recognition-falsifying modification occurred
	and other modification
	did not lead to failed processing,
	onion layers between
	the two last consecutive honest relays
	(the first might be
	the forward receiver
	(=backward sender))
	are replaced with random onion layers
	embedding the same path.
	{\ch We separate this hybrid
	into two just like \Hh[_2]
	for the same reason.
	Let $j$ be the index
	of the first of
	the two honest relays in question
	and $j'$ the index
	of the second.
	}
	{\ch \Hh[_{2a}\ba]} acts like \Hh[_1^{\ast\leftarrow}]
	except for the processing of
	{\ch $O_{j''}$,
	where $j''$ is the index
	of the last honest relay
	on the forward path
	of the onion.
	Due to \Hh[_2],
	the onion's path
	starts at
	the second-to-last honest relay
	on the original onion's forward path.
	The onion layers
	$O_{j''+1}$, \ldots, $O_n$, $O_1$, \ldots, $O_{j'}$
	are replaced with
	$\bar{O}_{j''+1}$, \ldots, $\bar{O}_n$, $\bar{O}_1$, \ldots, $\bar{O}_{j'}$,
	where $\bar{O}_{j''+1}$
	is formed as
	\[\FormO(j''+1, \Rr', m, R, \Pp, \Pp^{\leftarrow\prime}\!\!, PK_{\Pp}, PK_{\Pp^{\leftarrow\prime}\!\!})\]
	with honestly chosen
	randomness $\Rr'$
	and $\Pp^{\leftarrow\prime} = (P\ba_1, \ldots, P\ba_{j'})$,
	effectively cutting off
	the end of the onion's reply path.
	If the onion's information
	is on the \textit{Tag}-list,
	the tag is recreated
	on $\bar{O}_{j''+1}$
	before it is sent.
}

	{\ch \underline{$\Hh[_1^{\ast\leftarrow}] \approx_I \Hh[_{2a}\ba]$:}
		We construct
		an attacker on \NTI
		with $j_{\NTI} = 0$
		and $j\ba_{\NTI} = j'$
		from any distinguisher
		of the two hybrids
		the same way
		as for \Hh[_{2a}].
		Due to the previous hybrid,
		the original onion layers
		after relay $P_{j'}$
		are never output
		to the adversary
		since they have been replaced
		with random layers.
	}

	\textbf{Hybrid \Hh[_{2b}\ba]:}
	{\ch In this hybrid,
	we perform
	the actual replacement
	of the onion layers
	between the honest relays
	for the onion
	whose backward path
	was truncated in
	\Hh[_{2a}\ba].
	}
	In essence,
	the consecutive onion layers
	{\ch $\bar{O}\ba_{j+1}, \ldots, \bar{O}\ba_{j'}$}
	from a backward communication
	of an honest (forward) sender,
	starting at the second last honest relay $P_j\ba$
	to the next following honest relay $P\ba_{j'}$
	(on the backward path),
	are replaced with {\ch $\hat{O}_1, \ldots, \hat{O}_{j'-j}$}.
	Thereby for an
	honestly chosen {\ch $\Rr''$}
	{\ch and a random receiver $R_{rdm}$:}
	\[\hat{O}_1 \gets {\ch \FormO(1, \Rr'', m_{rdm}, R_{rdm}, \Pp^{\ch \prime}, (), PK_{\Pp^{\ch \prime}}, ())}\]
	where $m_{rdm}$ is a random message,
	$\Pp^{\ch \prime} = (P\ba_{\ch j+1}, \ldots, P\ba_{j'})$
	is the path from $P\ba_{\ch j}$ to $P\ba_{j'}$.

	Further,
	the hybrid calculates (and stores)
	another replacement for the next part
	after the current replacement
	$(\tilde{P}\ba_{j'+1}, \tilde{O}_k)$
	(by exploiting the fact that
	the sender knows the backward path
	and can infer the message from any layer)
	as in the hybrid \Hh[_1^{\ast\leftarrow}] before.
	Then it also stores
	$(\textit{info}, P\ba_{j'}, (\tilde{O}_k, \tilde{P}\ba_{j'+1}))$
	to the $\bar{O}$-list
	(to ensure the replacement
	of the later path
	is used as well).
	As before,
	the $\bar{O}$-list will be checked
	to pick the right processing
	of an onion.

	\underline{$\Hh[_{2a}\ba] \approx_I \Hh[_{2b}\ba]$:}
	\Hh[_{2b}\ba] replaces for one backward communication,
	the last subpath between
	two consecutive honest relay
	before an honest (forward) sender.
	The output to \Aa includes
	the later (by \Hh[_1^{\ast\leftarrow}])
	replaced onion layers $\bar{O}_{later}$
	after the second honest relay
	(these layers are identically generated
	in {\ch \Hh[_{2a}\ba]} and \Hh[_{2b}\ba])
	that take the original subpath
	but are otherwise chosen randomly;
	the original onion layers
	after the first of the honest relays
	for all communications
	not considered by \Hh[_{2b}\ba]
	(output by {\ch \Hh[_{2a}\ba]})
	or,
	in case of the communication considered by \Hh[_{2b}\ba],
	the newly drawn random replacement
	(generated by \Hh[_{2b}\ba]);
	and the processing before
	the first honest relay $P\ba_j$.

	{\ch Similarly to
	our argument for
	\Hh[_{2a}\ba],
	the $\bar{O}_{later}$ layers
	are random
	and independent of
	the replaced layers,
	so they can be built
	without needing
	the \NBLU challenger.}

	Thus,
	all that is left
	are the original/replaced onion layer
	after the honest relay
	and the original layers before.
	This is the same output as in $\Hh[_2^\ast] \approx_I \Hh[_1\ba]$.
	Hence,
	if there exists a distinguisher
	between {\ch \Hh[_{2a}\ba]} and \Hh[_{2b}\ba],
	there exists an attack on {\ch \NBLU}.

	\textbf{Hybrid \Hh[_2^{<x\leftarrow}]:}
	{\ch From here,
	we refer to
	the combination of
	\Hh[_{2a\ba}] and \Hh[_{2b}\ba]
	as \Hh[_2\ba]
	for convenience.}
	In this hybrid,
	for the first $x - 1$ honest subpaths
	on backwards communications
	are replaced with a random onion
	sharing the path
	and the other replacements calculated as before
	and all are stored on the $\bar{O}$-list.
	If \Aa previously
	(i.e.,
	in onion layers
	up to the honest relay
	starting the selected subpath)
	modified $\eta$ of an onion layer
	in this communication
	or modified another part
	such that processing fails,
	the communication is skipped.

	\underline{$\Hh[_2^{<x-1\leftarrow}] \approx_I \Hh[_2^{<x\leftarrow}]$:}
	Analogous to above.

	\textbf{Onion replacement for corrupted receivers}

	We replace the missing part
	between the onion layers
	already replaced on the forward path
	and the onion layers already replaced
	on the backward path.
	{\ch Due to the structure
		of RSOR,
		this part exists
		for every onion
		and includes
		the link between
		the exit relay
		and the receiver.
		Note that
		if the exit relay
		is the last honest relay
		on the onion's path,
		then no replacement
		can take place
		because the message
		is simply sent to the receiver
		in plaintext.
		We thus only consider onions
		that still have at least
		one adversarial relay
		left on their path
		(which also implies
		that their exit relay
		is adversarial).
	}

	\textbf{Hybrid \Hh[_3]:}
	In this hybrid,
	for the first forward communication
	for which,
	in the adversarial processing,
	no recognition-falsifying modification
	(i.e., a modification on $\eta$)
	occurred
	(and no other modification
	caused the processing to fail)
	so far,
	forward onion layers
	from its last honest relay
	to the corrupted {\ch exit relay}
	are replaced with random onions
	sharing this path, {\ch receiver},
	and message
	and the first part of the reply-path.
	More precisely,
	this machine acts like \Hh[_2^{\ast\leftarrow}]
	except for the processing of $O_j$;
	in essence,
	the consecutive onion layers
	$O_{j+1}, \ldots, O_{\ch n}$
	from a communication of an honest sender,
	starting at the last honest relay $P_j$
	to the corrupted {\ch exit relay} $P_{\ch n}$
	are replaced with
	$\bar{O}_1, \ldots, \bar{O}_{\ch n-j}$;
	Thereby,
	for an honestly chosen $\Rr'$:
	\[\bar{O}_1 \gets {\ch \FormO(1, \Rr', m, R, \Pp^{\ch \prime}, \Pp^{\ch \prime \leftarrow}\!\!, PK_{\Pp^{\ch \prime}}, PK_{\Pp^{\ch \prime \leftarrow}\!\!})},\]
	where $m$ is the message
	of this communication,
	{\ch $R$ is the receiver,}
	$\Pp^{\ch \prime} = (P_j, \ldots, P_{\ch n})$
	is the path from $P_j$ to $P_{\ch n}$
	and $\Pp^{\ch \prime \leftarrow}$
	is the first part of the reply-path
	(until the first honest relay),
	that a reply to the original onion
	would have taken.
	{\ch If the onion's information
	is on the \textit{Tag}-list,
	then the tag
	is recreated
	on $\bar{O}_1$
	before it is sent.}

	\Hh[_3] further checks for
	every onion (ending at) $\Pp\ba.last$,
	if it was a reply to these replaced onion layers
	(by using the information \textit{info} stored and \RO).
	If so,
	it uses its knowledge
	about the original forward onion
	(before replacement)
	and the sender
	to construct the belonging original reply.
	With it,
	it computes the replacement
	of the later onion layers
	for this communication
	as in hybrid \Hh[_2^{\ast\leftarrow}]
	and stores the corresponding information
	on the $\bar{O}$-list.
	As before,
	the $\bar{O}$-list will be checked
	to pick the right processing of an onion.

	\underline{$\Hh[_2^{\ast\leftarrow}] \approx_I \Hh[_3]$:}
	Similar to $\Hh[_1^\ast] \approx_I \Hh[_2]$,
	the forward onion layers before $P_j$
	are independent and hence can be simulated
	for the distinguisher by an attack
	on {\ch \NTI}.
	{\ch If the adversary
	tags one of those layers,
	the attack can recognize
	the tag since it knows
	the expected payload
	it should be receiving
	at each honest relay.
	The attack can then
	recreate the tag
	after getting
	the challenge onion
	from the \NTI challenger.}
	Similarly to $\Hh[_1^{\ast\leftarrow}] \approx_I \Hh[_2\ba]$,
	the backward onion layers after $\Pp\ba.last$
	are independent and hence can be simulated
	for the distinguisher by an attack
	on {\ch \NTI}.
	The remaining outputs suffice to
	construct an attack on {\ch \NTI}
	similar to the one on {\ch \TFLU}
	in \Hh[_1^\ast] and \Hh[_2].

	\textbf{Hybrid \Hh[_3^{<x}]:}
	In this hybrid,
	for the first $x - 1$ forward communications
	for which,
	in the adversarial processing,
	no recognition-falsifying modification
	(and no other modification
	that results in failed processing)
	occurred so far,
	the onion layers between
	its last honest relay to corrupted {\ch exit relay}
	are replaced with random onion layers
	sharing the path,
	message,
	{\ch receiver,}
	and first part of the reply path.

	$\Hh[_3^{<x-1}] \approx_I \Hh[_3^{<x}]$:
	Analogous to above.

	\textbf{Hybrid \Hh[_4]:}
	This machine acts the way
	that \Ss acts
	in combination with \FRN.
	Note that \Hh[_3^\ast] only behaves differently from \Ss in
	(a) routing onions through the honest parties
	and (b)
	where it gets its information needed
	for choosing the replacement onion layers:
	(a) \Hh[_3^\ast] actually routes them
	through the real honest parties
	that do all the computation.
	\Hh[_4] instead runs the way
	that \FRN and \Ss operate:
	there are no real honest parties,
	and the ideal honest parties do not do
	any crypto work.
	(b) \Hh[_3^\ast]
	gets inputs directly from the environment
	and gives output to it.
	In \Hh[_4],
	the environment instead gives
	inputs to \FRN
	and \Ss gets the needed information
	(i.e.,
	parts of path
	and the included message,
	if the receiver is corrupted)
	from outputs of \FRN
	as the ideal world adversary.
	\FRN gives the outputs
	to the environment as needed.

	\underline{$\Hh[_3^\ast] \approx_I \Hh[_4]$:}
	For the interaction with the environment
	from the protocol/ideal functionality,
	it is easy to see
	that the simulator directly gets
	the information it needs
	from the outputs of the ideal functionality
	to the adversary:
	Whenever an honest relay
	is done processing,
	it needs the path
	from it to the next honest relay
	or path from it to the corrupted receiver
	and in this case also the message
	and beginning of the backward path.
	This information is given to \Ss by \FRN.

	Further,
	in the real protocol,
	the environment is notified by honest relays
	when they receive an onion
	together with some random ID
	that the environment sends back to signal
	that the honest relay
	is done processing the onion.
	The same is done
	in the ideal functionality.
	Notice that the simulator ensures that
	every communication is simulated in \FRN
	such that those notifications arrive at the environment
	without any difference
	(this includes them having the same repliability).

	For the interaction
	with the real world adversary,
	we distinguish the outputs
	in communications from
	honest and corrupted senders.
	\vspace{-2pt}
	\begin{enumerate}
		\setcounter{enumi}{-1}
		\item Corrupted (forward) senders:
			In the case of a corrupted sender,
			both \Hh[_3^\ast] and \Hh[_4]
			(i.e., $\Ss+\FRN$)
			do not replace any onion layers
			except that with negligible probability
			a collision on the $\bar{O}$-list
			resp. $O$-list occurs.
			(Notice that even for honest receivers
			(and thus backward senders)
			layers following the protocol
			can be and are created.)

		\item Honest senders:
			\let\origenumii\labelenumii
			\renewcommand{\labelenumii}{1.\arabic{enumii})}
			\begin{enumerate}
				\item No recognition-falsifying modification of the onion
					by the adversary happens
					(and if modification happens at all,
					the processing does not fail
					[note that a failing processing
					is the same as dropping;
					see 1.3)]):
					All parts of the path
					are replaced with
					randomly drawn onion layers $\bar{O}_i$.
					The way those layers are chosen is identical
					for \Hh[_3^\ast] and \Hh[_4] (i.e., $\Ss + \FRN$).

				{\ch \item The onion
					is tagged
					by the adversary:
					If the tagging occurs
					on the forwards path
					between two honest relays,
					the onion's information
					is added to
					the \textit{Tag}-list
					at the next honest relay.
					This list exactly corresponds
					to the $L_{tag}$ list
					in \FRN
					for forward onions
					and the onion is treated accordingly
					at the last honest relay:
					If that relay
					is the exit relay,
					no output is produced.
					If it is not the exit relay,
					the steps in \Hh[_3^\ast]
					ensure that
					the corrupted exit relay
					receives a tagged onion.
					If the tagging occurs
					on the backwards path,
					the onion's information
					is not added to
					the \textit{Tag}-list
					and the final onion layer
					delivered to
					the reply receiver
					is not tagged.
					This does not change
					the output to \Zz
					because received replies
					are never output
					to the environment
					by honest relays.
				}

				\item Some recognition-falsifying modification of the onion
					or a drop or insert happens:
					As soon as a recognition-falsifying modification happens,
					both \Hh[_3^\ast] and \Hh[_4] continue to use
					the bit-identical onion
					for the further processing
					except when,
					with negligible probability,
					a collision on the $\bar{O}$-list
					resp. $O$-list occurs.
					In case of a dropped onion,
					it is simply not processed further
					in any of the two machines.
			\end{enumerate}
			\let\labelenumii\origenumii

	\end{enumerate}

	Note that the view of the environment
	in the real protocol
	is the same as its view
	in interacting with \Hh[_0].
	Similarly,
	its view in the ideal protocol
	with the simulator
	is the same as its view
	in interacting with \Hh[_4].
	As we have shown
	indistinguishability in every step,
	we have indistinguishability
	in their views.
	\end{proof}

\end{document}